\documentclass[11pt]{article}

\usepackage{fullpage}

\usepackage{indentfirst}
\usepackage{amsfonts}
\usepackage{amsmath}
\usepackage{amsthm}
\usepackage{amssymb}
\usepackage{amsthm}
\usepackage{color}
\usepackage{verbatim}
\usepackage{epsfig}
\usepackage{graphicx}
\usepackage[numbers]{natbib}
\usepackage{hyperref}
\usepackage{bm}
\usepackage{xspace}
\usepackage[capitalise]{cleveref}
\usepackage{framed}

\usepackage[labelfont=bf]{caption}

\definecolor{corlinks}{RGB}{0,0,150}
\definecolor{cormenu}{RGB}{200,0,0}
\definecolor{corurl}{RGB}{200,0,0}

\hypersetup{
colorlinks=true,
urlcolor=corlinks,
linkcolor=corlinks,
menucolor=corlinks,
citecolor=corlinks,
pdfborder= 0 0 0
}
\newtheorem{theorem}{Theorem}

\newtheorem{lemma}[theorem]{Lemma}
\newtheorem{corollary}[theorem]{Corollary}
\newtheorem{definition}[theorem]{Definition}
\newtheorem{proposition}[theorem]{Proposition}
\newtheorem{remark}[theorem]{Remark}

\newtheorem{claim}[theorem]{Claim}

\def\colorful{1}

\ifnum\colorful=1

\fi
\ifnum\colorful=0

\fi

\newcommand{\SIZE}{\mathsf{SIZE}}
\newcommand{\poly}{\mathsf{poly}}

\newcommand{\Ptime}{\ensuremath{{\sf P}}\xspace}
\newcommand{\Ppoly}{\ensuremath{{\sf P/poly}}\xspace}

\newcommand{\NC}{\ensuremath{{\sf NC}^1}\xspace}
\newcommand{\AC}{\ensuremath{{\sf AC}^0}\xspace}

\newcommand{\NPtime}{\ensuremath{{\sf NP}}\xspace}
\newcommand{\NP}{\ensuremath{{\sf NP}}\xspace}
\newcommand{\coNP}{\ensuremath{{\sf coNP}}\xspace}
\newcommand{\QPtime}{\ensuremath{{\sf QP}}\xspace}
\newcommand{\NQPtime}{\ensuremath{{\sf NQP}}\xspace}
\newcommand{\TIME}{\ensuremath{{\sf TIME}}\xspace}
\newcommand{\NTIME}{\ensuremath{{\sf NTIME}}\xspace}

\newcommand{\Formula}{\ensuremath{{\sf Formula}}\xspace}
\newcommand{\Circuit}{\ensuremath{{\sf Circuit}}\xspace}

\newcommand{\PARITY}{\ensuremath{{\sf PARITY}}\xspace}
\newcommand{\MCSP}{\ensuremath{{\sf MCSP}}\xspace}
\newcommand{\MKtP}{\ensuremath{{\sf MKtP}}\xspace}

\newcommand{\Succinct}{\ensuremath{{\sf Succinct}\text{-}\ensuremath{\MCSP}}\xspace}

\newcommand{\ttable}{\ensuremath{{\sf tt}}\xspace}

\newcommand{\polylog}{\mathop{\mathrm{polylog}}}

\newcommand{\cC}{\mathcal C}
\newcommand{\cD}{\mathcal D}

\newcommand{\cF}{\mathcal F}
\newcommand{\cG}{\mathcal G}

\newcommand{\cR}{\mathcal R}

\newcommand{\N}{\mathbb{N}}

\newcommand{\eps}{\varepsilon}
\newcommand{\Size}{\mathsf{Size}}

\newcommand{\boldsym}{\boldsymbol}
\newenvironment{proofsketch}{\begin{proof}[Proof Sketch]}{\end{proof}}

\newcommand{\Nat}{\mathbb{N}}
\newcommand{\binset}{\{0,1\}}
\newcommand{\Kt}{\mathsf{Kt}}
\newcommand{\DT}{\mathsf{DT}}

\newcommand{\GapAND}{\mathsf{GapAND}}
\newcommand{\FM}{\mathsf{Formula}}
\newcommand{\Andreev}{\mathsf{Andreev}}

\renewcommand{\restriction}{\mathord{\upharpoonright}}

\def\ShowAuthNotes{1}
\ifnum\ShowAuthNotes=1
\newcommand{\authnote}[2]{\ \\ \textcolor{red}{\parbox{0.9\linewidth}{[{\footnotesize {\bf #1:} { {#2}}}]}}\newline}
\else
\newcommand{\authnote}[2]{}
\fi

\newcommand{\lnote}[1]{\authnote{Lijie}{#1}}

\newenvironment{proofof}[1]{\begin{proof}[Proof of #1]}{\end{proof}}

\begin{document}

\title{Beyond Natural Proofs:\\Hardness Magnification and Locality\vspace{0.3cm}}

\author{
\hspace{-1.6cm} Lijie Chen\thanks{\texttt{\href{lijieche@mit.edu}{lijieche@mit.edu}}}\vspace{0.05cm}\\ \hspace{-1.6cm} \small{Massachusetts Institute of Technology}
\and
Shuichi Hirahara\thanks{\texttt{\href{mailto:s_hirahara@nii.ac.jp}{s\_hirahara@nii.ac.jp}}}\vspace{0.05cm}\\ \small{National Institute of Informatics}
\and
 \hspace{0.8cm}Igor C. Oliveira\thanks{\texttt{\href{mailto:igor.oliveira@warwick.ac.uk}{igor.oliveira@warwick.ac.uk}}}\vspace{0.05cm}\\ \hspace{0.8cm} \small{University of Warwick} \vspace{0.2cm}
\and
J\'{a}n Pich\thanks{\texttt{\href{mailto:jan.pich@cs.ox.ac.uk}{jan.pich@cs.ox.ac.uk}}}\vspace{0.05cm}\\ \small{University of Oxford}\vspace{0.3cm}
\and Ninad Rajgopal\thanks{\texttt{\href{mailto:ninad.rajgopal@cs.ox.ac.uk}{ninad.rajgopal@cs.ox.ac.uk}}}\vspace{0.05cm}\\ \small{University of Oxford}\vspace{0.0cm}
\and Rahul Santhanam\thanks{\texttt{\href{mailto:rahul.santhanam@cs.ox.ac.uk}{rahul.santhanam@cs.ox.ac.uk}}}\vspace{0.05cm}\\ \small{University of Oxford}~\\
}
\maketitle

\vspace{-0.8cm}

\begin{abstract}
Hardness magnification reduces major complexity separations (such as $\mathsf{\mathsf{EXP}} \nsubseteq \mathsf{NC}^1$) to proving lower bounds for some natural problem $Q$ against weak circuit models. Several recent works \citep{OS18_mag_first, MMW_STOC_paper, CT19_STOC, OPS19_CCC, CMMW_CCC_paper, DBLP:conf/icalp/Oliveira19, Magnification_FOCS19} have established results of this form. In the most intriguing cases, the required lower bound is known for problems that appear to be significantly easier than $Q$, while $Q$ itself is susceptible to lower bounds but these are not yet sufficient for magnification. 

In this work, we provide more examples of this phenomenon, and investigate the prospects of proving new lower bounds using this approach. In particular, we consider the following essential questions associated with the hardness magnification program:

\vspace{0.2cm}
\noindent -- \emph{Does hardness magnification avoid the natural proofs barrier of Razborov and Rudich \emph{\citep{DBLP:journals/jcss/RazborovR97}}?}

\vspace{0.1cm}
\noindent -- \emph{Can we adapt known lower bound techniques to establish the desired lower bound for $Q$?}

\vspace{0.2cm}

We establish that some instantiations of hardness magnification overcome the natural proofs barrier in the following sense: slightly superlinear-size circuit lower bounds for certain versions of the minimum circuit size problem \MCSP imply the non-existence of natural proofs. 
As a corollary of our result, we show that certain magnification theorems not only imply strong worst-case circuit lower bounds but also rule out the existence of efficient learning algorithms.

Hardness magnification might sidestep natural proofs, but we identify a source of difficulty when trying to adapt existing lower bound techniques to prove strong lower bounds via magnification. This is captured by a \emph{locality barrier}: existing magnification theorems \emph{unconditionally} show that the problems $Q$ considered above admit highly efficient circuits extended with small fan-in oracle gates, while lower bound techniques against weak circuit models quite often easily extend to circuits containing such oracles. This explains why direct adaptations of certain lower bounds are unlikely to yield strong complexity separations via hardness magnification.
\end{abstract}

\newpage
\tableofcontents
\newpage

\section{Introduction}\label{s:introduction}

Proving circuit size lower bounds for explicit Boolean functions is a central problem in Complexity Theory. Unfortunately, it is also notoriously hard, and   arguments ruling out a wide range of approaches have been discovered. The most prominent of them is the {\em natural proofs barrier} of Razborov and Rudich \cite{DBLP:journals/jcss/RazborovR97}. 

A candidate approach for overcoming this barrier was investigated recently by Oliveira and Santhanam \cite{OS18_mag_first}. {\em Hardness Magnification} identifies situations where strong circuit lower bounds for explicit Boolean functions (e.g.~$\NP\not\subseteq\Ppoly$) follow from much weaker (e.g.~slightly superlinear) lower bounds for specific natural problems. As discussed in \citep{OS18_mag_first}, in some cases the lower bounds required for magnification are already known for explicit problems, but not yet for the problem for which the magnification theorem holds. This approach to lower bounds has attracted the interest of several researchers, and a number of recent works have proved magnification results \cite{MMW_STOC_paper, CT19_STOC, OPS19_CCC, CMMW_CCC_paper, DBLP:conf/icalp/Oliveira19, Magnification_FOCS19} (see also \citep{Srinivasan03, AK10, LW13, DBLP:journals/eccc/MullerP17} for related previous work). We provide a concise review of existing results in Appendix \ref{ss:previous_work}.

In this work, we are interested in understanding the prospects of proving new lower bounds using hardness magnification, including potential barriers.

\subsection{Hardness Magnification Frontiers}\label{ss:intro_background_motivation} 

While hardness magnification is a broad phenomenon, its most promising instantiations seem to occur in the setting of circuit classes such as \NC. The potential of hardness magnification stems from establishing the following scenario.
\bigskip

\medskip

\setlength\fboxrule{2pt}
\fbox{\parbox{433pt}{
\vspace{6pt}

\quad {\bf HM Frontier:} There is a natural problem $Q$ and a computational model $\mathcal{C}$ such that:
\vspace{-1pt}
\begin{itemize}
\item[1.] (\emph{Magnification}) 
$Q \notin\mathcal{C}$ implies $\NP\not\subseteq \NC$ or a similar breakthrough.
\item[2.] (\emph{Evidence of Hardness}) $Q \notin\mathcal{C}$ under a standard conjecture.
\item[3.] (\emph{Lower Bound against $\mathcal{C}$}) $L \notin\mathcal{C}$, where $L$ is a simple function like \PARITY. 
\item[4.] (\emph{Lower Bound for $Q$}) $Q \notin\mathcal{C}^-$, 
where $\mathcal{C}^-$ is slightly weaker than $\mathcal{C}$.
\end{itemize}
\vspace{-4pt}}}

\bigskip

\medskip

A frontier of this form provides hope that the required lower bound in Item 1 is true (thanks to Item 2), and that it might be within the reach of known techniques (thanks to Items 3 and 4, which provide evidence that we can analyse the circuit model and the problem). HM frontiers have been already achieved in earlier works with a striking example appearing in \citep{OPS19_CCC} (see also \citep{Magnification_FOCS19}). Despite the number of works in this area, we note that the HM frontier is achieved only by some magnification theorems (Item 3 is often unknown; e.g. in the case of results in \citep{AK10, CT19_STOC}). 

In order to make our subsequent discussion more concrete, we provide five examples of HM frontiers. Some of these results are new or require an extension of previous work, and the relevant statements will be explained in more detail in Section \ref{a:improved_magnification_MCSP}. The list of frontiers is not meant to be exhaustive, but we have tried to cover different computational models.

\setlength\fboxrule{0.6pt}
\fbox{\parbox{436pt}{
\vspace{6pt}
\quad {\bf (A) HM Frontier for $\mathsf{MKtP}[n^c, 2n^c]$ and} $\mathsf{AC}^0$-$\mathsf{XOR}$\textbf{:}
\medskip

\quad A1. If $\mathsf{MKtP}[n^c,2n^c] \notin \mathsf{AC}^0$-$\mathsf{XOR}[N^{1.01}]$ for large $c>1$ then $\mathsf{EXP} \nsubseteq \mathsf{NC}^1$ (Section \ref{ss:frontier_ACXOR}).

\quad A2. $\mathsf{MKtP}[n^c,2n^c]\notin\mathsf{P}/\mathsf{poly}$ for large enough $c$ under exponentially secure PRFs \citep{DBLP:journals/jcss/RazborovR97}.

\quad A3. $\mathsf{Majority} \notin \mathsf{AC}^0$-$\mathsf{XOR}[2^{N^{o(1)}}]$ (immediate from \citep{Razborov87, DBLP:conf/stoc/Smolensky87}).

\quad A4. $\mathsf{MKtP}[n^c,2n^c] \notin \mathsf{AC}^0$ for any sufficiently large constant $c$  (Section \ref{ss:frontier_ACXOR}).
\vspace{3pt}}}
\medskip

\vspace{0.3cm}

\noindent {\normalsize \textbf{A.} $\mathsf{MKtP}[s,t]$ refers to the promise problem of determining if an $N$-bit input has Levin Kolmogorov complexity at most $s$ versus at least $t$ (cf.~\citep{OPS19_CCC}). Here $N=2^n$. The $\mathsf{AC}^0$-$\mathsf{XOR}$ model is the extension of $\mathsf{AC}^0$ where gates at the bottom layer of the circuit can compute arbitrary parity functions. $\mathsf{AC}^0$-$\mathsf{XOR}[s]$ denotes $\mathsf{AC}^0$-$\mathsf{XOR}$ circuits of size $s$ where the size is measured as the number of gates. This circuit class has received some attention in recent years (cf.~\citep{DBLP:journals/jcss/CheraghchiGJWX18})}, and a few basic questions about $\mathsf{AC}^0$ circuits with parity gates (such as constructing PRGs of seed lengh $o(n)$ and learnability using random examples) remain open for $\mathsf{AC}^0$-$\mathsf{XOR}$ as well.

\vspace{0.4cm}

\setlength\fboxrule{0.6pt}
\fbox{\parbox{436pt}{
\vspace{6pt}

\quad {\bf (B) HM Frontier for $\MCSP[2^{n^{1/3}},2^{n^{2/3}}]$ and} $\Formula$-$\mathsf{XOR}${\bf:}
\medskip

\quad B1. $\MCSP[2^{n^{1/3}},2^{n^{2/3}}]\notin\Formula$-$\mathsf{XOR}[N^{1.01}]$ implies $\mathsf{NQP}\not\subseteq\NC$ (Section \ref{ss:frontier_MCSP_formula_xor}).

\quad B2. $\MCSP[2^{n^{1/3}},2^{n^{2/3}}]\notin\mathsf{P}/\mathsf{poly}$ under standard cryptographic assumptions 
\citep{DBLP:journals/jcss/RazborovR97}.

\quad B3. $\mathsf{InnerProduct}\notin\Formula$-$\mathsf{XOR}[N^{1.99}]$ (immediate consequence of \cite{Tal17}).

\quad B4. $\MCSP[2^{n^{1/3}},2^{n^{2/3}}]\notin \Formula[N^{1.99}]$ (\cite{HS17}; see also \citep{OPS19_CCC}).
\vspace{3pt}}}
\medskip

\vspace{0.3cm}

\noindent {\normalsize \textbf{B.} In the statements above, $\mathsf{NQP}$ stands for nondeterministic quasi-polynomial time, $\mathsf{InnerProduct}$ is the Boolean function defined as $\mathsf{InnerProduct}(x,y)= \sum_i x_i \cdot y_i\ (\mathsf{mod}\ 2)$, where $x,y \in \{0,1\}^N$, $\Formula$-$\mathsf{XOR}[s]$ refers to the class of Boolean formulas over the De Morgan basis with at most $s$ leaves, where each leaf is an XOR of arbitrary arity over the inputs\footnote{Note that $\mathsf{Formula}$-$\mathsf{XOR}[N^{1.01}] \subseteq \mathsf{Formula}[N^{3.01}]$. A better understanding of the former class is therefore necessary before we can understand the power and limitations of super-cubic formulas, which is a major open question in circuit complexity.}, and $\MCSP[s,t]$ denotes a promise problem over $N= 2^n$ input bits with YES inputs being truth-tables of Boolean functions on $n$ inputs which are computable by circuits of size $s$, and NO instances being truth-tables of Boolean functions which are hard for circuits of size $t$.}

\vspace{0.4cm}

\fbox{\parbox{436pt}{
\vspace{6pt}
\quad {\bf (C) HM Frontier for $\MCSP[2^{n^{1/2}}/10n,2^{n^{1/2}}]$ and} $\mathsf{Almost}$-$\mathsf{Formulas}$\textbf{:}
\medskip

\quad C1. $\MCSP[\frac{2^{n^{1/2}}}{10n},2^{n^{1/2}}]\notin N^{0.01}$-$\mathsf{Almost}$-$\Formula[N^{1.01}]$ implies $\mathsf{NP}\not\subseteq\NC$ (Section \ref{ss:almost_formulas_magnification}).

\quad C2. $\MCSP[\frac{2^{n^{1/2}}}{10n},2^{n^{1/2}}]\notin\mathsf{P}/\mathsf{poly}$ under standard cryptographic assumptions 
\citep{DBLP:journals/jcss/RazborovR97}.

\quad C3. $\mathsf{PARITY}\notin N^{0.01}$-$\mathsf{Almost}$-$\Formula[N^{1.01}]$  (Section \ref{ss:almost_formulas_magnification}).

\vspace{-0.075cm}

\quad C4. $\MCSP[2^{n^{1/2}}/10n,2^{n^{1/2}}]\notin \Formula[N^{1.99}]$ (\cite{HS17}; see also \citep{OPS19_CCC}).
\vspace{3pt}}}
\medskip

\vspace{0.3cm}

\noindent {\normalsize \textbf{C.} An almost-formula is a circuit with a bounded number of gates of fan-out larger than $1$. More precisely, a $\gamma$-$\mathsf{Almost}$-$\mathsf{Formula}[s]$ is a circuit containing at most $s$ AND, OR, NOT gates of fan-in at most $2$, and among such gates, at most $\gamma$ of them have fan-out larger than $1$. Consequently, this class naturally interpolates between formulas and circuits. This magnification frontier can be seen as progress towards establishing magnification theorems for worst-case variants of $\mathsf{MCSP}$ in the regime of sub-quadratic formulas (see the discussion in \citep{OPS19_CCC}).

\vspace{0.4cm}

\fbox{\parbox{436pt}{
		\vspace{6pt}
		\quad {\bf (D) HM Frontier for $\MCSP[2^{\sqrt{n}}]$ and  one-sided error randomized formulas:}
		\medskip
		
		\quad D1. $\MCSP[2^{\sqrt{n}}] \notin \GapAND_{O(N)}$-$\FM[N^{2.01}] \Rightarrow \NQPtime \notin \NC$ (\cref{ss:frontier_MCSP_formula_xor}).
		
		\quad D2. $\MCSP[2^{\sqrt{n}}] \notin \mathsf{P}/{\mathsf{poly}}$ under standard cryptographic assumptions \citep{DBLP:journals/jcss/RazborovR97}.
		
		\quad D3.1. $\Andreev_N \notin \GapAND_{O(N)}$-$\FM[N^{2.99}]$ (implicit in~\cite{Hastad98}).
		
		\quad D3.2. $\MCSP[2^n / n^4] \notin \GapAND_{O(N)}$-$\FM[N^{2.99}]$ (implicit in~\cite{CheraghchiKLM19_eccc_journals}).
		
		\quad D4. $\MCSP[2^{\sqrt{n}}] \notin \GapAND_{O(N)}$-$\FM[N^{1.99}]$ (\cite{CJW_tight_threshold}, building on~\cite{HS17,OPS19_CCC}).
\vspace{3pt}}}
\medskip

\vspace{0.3cm}

\noindent {\normalsize \textbf{D.} $\GapAND_{N}$ is the promise function on $N$ bits such that it outputs $1$ when all input bits are $1$, and outputs $0$ when at most $1/10$ of the input bits are $1$. $\GapAND_{O(N)}$-$\FM[s]$ denotes circuits with $\GapAND_{O(N)}$ gate at the top with formulas of size $s$ being inputs of the top gate. Therefore, $\GapAND_{O(N)}$-$\FM$ can be seen as \emph{randomized formulas with one-sided error}.\footnote{Suppose there is a $\GapAND_{O(N)}$-$\FM$ circuit computing a function $f \colon \{0,1\}^N \to \{0,1\}$. Consider a uniform distribution of all sub-formulas below the top $\GapAND_{O(N)}$ gate. Then for any input $x$, if $f(x) = 1$ then a sample formula from that distribution always outputs $1$ on $x$, otherwise it outputs $0$ with probability at least $0.9$ on $x$. On the other hand, it is possible to derandomize a distribution of formulas computing $f$ with one-sided error using a top $\GapAND_{O(N)}$ gate.} The most interesting aspect of this magnification frontier is that the gap between the known hardness result and the magnification threshold is \emph{nearly-tight} ($N^{2-\eps}$ versus $N^{2+\eps}$).\footnote{This tight threshold is first observed in~\cite{CJW_tight_threshold}, we include it here to show that the barrier discussed in this paper also applies to this particular setting.}
	
	\vspace{0.4cm}

\setlength\fboxrule{0.6pt}
\fbox{\parbox{436pt}{
\vspace{6pt}
\quad {\bf (E) HM Frontier for $(n - k)$\normalfont{-}$\mathsf{Clique}$ \textbf{and} $\mathsf{AC}^0$\textbf{:}}
\medskip

\quad E1. If $(n-k)$-$\mathsf{Clique} \notin \mathsf{AC}^0[m^{1.01}]$ for some $k = (\log n)^C$, then $\mathsf{NP} \nsubseteq \mathsf{NC}^1$ (Section \ref{ss:frontier_clique}).

\quad E2. (Non-uniform) ETH $\Rightarrow$ $(n-k)$-$\mathsf{Clique} \notin \mathsf{P}/\mathsf{poly}$ for some $k = (\log n)^C$ (Section \ref{ss:frontier_clique}).

\quad E3. $\mathsf{Parity} \notin \mathsf{AC}^0$ \citep{ajtai1983sup, DBLP:journals/mst/FurstSS84}.

\quad E4. $(n-k)$-$\mathsf{Clique} \notin \mathsf{mP}/\mathsf{poly}$ for some $k = (\log n)^C$ (\citep{DBLP:journals/dm/AndreevJ08}; see Section \ref{ss:frontier_clique}).
\vspace{3pt}}}
\medskip

\vspace{0.3cm}

\noindent {\normalsize \textbf{E.} } The $\ell$-$\mathsf{Clique}$ problem is defined on graphs on $n$ vertices in the adjacency matrix representation of size $m = \Theta(n^2)$. (The statements above refers to the regime of very large clique detection.) The class $\mathsf{mP}/\mathsf{poly}$ refers to monotone circuits of polynomial size. In this frontier we are modifying Item 4 from HM frontier so that instead of slightly weaker $\mathcal{C^-}$ we consider an incomparable $\mathcal{C^-}$. This frontier is however particularly interesting, as items E1 and E4 connect hardness magnification to a basic question about the power of \emph{non-monotone} circuits when computing \emph{monotone} functions (see \citep{DBLP:conf/stoc/ChenOS17, DBLP:conf/innovations/GoosKRS19} and references therein): Is every monotone function in $\mathsf{AC}^0$ computable by a monotone (unbounded depth) boolean circuit of polynomial size? If this is the case, $\mathsf{NP} \nsubseteq \mathsf{NC}^1$ would follow.

\vspace{0.4cm}

Note that these hardness magnification frontiers offer different approaches to proving lower bounds against $\mathsf{NC}^1$. \\

\noindent \textbf{Essential Questions.} Do magnification theorems bring us closer to strong circuit lower bounds? In order to understand the limits and prospects of hardness magnification, the following questions are relevant. 

\begin{itemize}
\item[\textbf{Q1.}] \emph{Naturalization.} Is hardness magnification a \emph{non-naturalizing} 
approach to circuit lower bounds? 
If we accept standard cryptographic assumptions, non-naturalizability is a necessary property of any successful approach to strong circuit lower bounds.\footnote{We assume familiarity of the reader with the natural proofs framework of \citep{DBLP:journals/jcss/RazborovR97}.}

\item[\textbf{Q2.}] \emph{Extending known lower bounds.} Can we adapt an existing lower bound proof from Items 3 and 4 in some HM frontier to show the lower bound required from Item 1 in that HM frontier? Is it possible to establish the required lower bounds via a reduction from $L$ to $Q$? 

\item[\textbf{Q3.}] \emph{Improving existing magnification theorems.} Can we close the gap between Items 1 and 4 in HM frontier by establishing a magnification theorem that meets \emph{known} lower bounds, such as the ones appearing in Item 4? 
\end{itemize}


In the next sections, we present results that shed light into all these questions.

\subsection{Hardness Magnification and Natural Proofs}\label{ss:non_natural}

The very existence of the natural proofs barrier provides a direction for proving strong circuit lower bounds: one can proceed by \emph{refuting the existence of natural properties}.\footnote{A similar perspective has been employed in proof complexity in attempts to approach strong proof complexity lower bounds by extending the natural proofs barrier (see \cite{KForc11, Razb15}).} In other words, a way to avoid natural proofs is to prove that there are no natural proofs. It is also easy to see that \Ppoly-natural properties useful against \Ppoly can be turned into natural properties with much higher constructivity, e.g.~into linear-size natural properties useful against circuits of polynomial-size. If read contrapositively, this gives a form of hardness magnification.

The initial hardness magnification theorem of Oliveira and Santhanam \cite{OS18_mag_first} proceeds in a similar fashion. It proposes to approach $\NP\not\subseteq\Ppoly$ by deriving slightly superlinear circuit lower bounds for specific problems such as an \emph{approximate version} of 
\MCSP, which 
asks to distinguish truth-tables of Boolean functions computable by small circuits 
from truth-tables of Boolean functions which are hard to approximate 
by small circuits. 
Interestingly, this approach does not seem to naturalize, as it appears to yield strong lower bounds only for certain problems, and not for most of them. (The same heuristic argument appears in \citep{AK10}.) However, this is only an informal argument, and we would like to get stronger evidence that the natural proofs barrier does not apply here.


We show that hardness magnification for approximate \MCSP can be used to conclude the \emph{non-existence} of natural proofs against polynomial-size circuits. More precisely, we prove that if approximate \MCSP requires slightly superlinear-size circuits, then there are no \Ppoly-natural properties against \Ppoly. This strongly suggests that the natural proofs barrier isn't relevant to the magnification approach. Indeed, there remains the possibility that the weak circuit lower bound for \MCSP in the hypothesis of the result can be shown using naturalizing techniques (as there aren't any strong enough plausible cryptographic conjectures known that rule this out), and yet by using magnification to 'break' naturalness, we could get strong circuit lower bounds and even conclude the non-existence of natural proofs!

The core of our proof is the following new hardness magnification theorem:~if approximate \MCSP requires slightly superlinear-size circuits, then not only $\NP\not\subseteq\Ppoly$ but \emph{it is impossible even to learn efficiently}. We can then refute the existence of natural proofs by applying the known translation of natural properties to learning algorithms \cite{DBLP:conf/coco/CarmosinoIKK16}. Similar implications hold with a \emph{worst-case gap version} of \MCSP (in the sense of HM Frontiers B and C but with different parameters) instead of approximate \MCSP, following an idea from \citep{Hir18}.

Interestingly, all the implications above are actually \emph{equivalences}. In particular, the existence of natural properties is equivalent to the existence of highly efficient circuits for computing approximate \MCSP and worst-case gap \MCSP with certain parameters (cf.~Theorem \ref{thm:equivalences}). This extends a known characterization of natural properties: Carmosino et al.~\cite{DBLP:conf/coco/CarmosinoIKK16} showed that \Ppoly natural proofs against \Ppoly are equivalent to learning \Ppoly by subexponential-size circuits, which was in turn shown to be equivalent by Oliveira and Santhanam \cite{DBLP:conf/coco/OliveiraS17} to the non-existence of non-uniform pseudorandom function families of sub-exponential security. The connection of hardness magnification to learning and pseudorandom function generators might be of independent interest, since it extends the consequences of magnification into two central areas in Complexity Theory.

  \begin{theorem}[Equivalences for Hardness Magnification]
  \label{thm:equivalences}
  The following statements are equivalent:\footnote{See Preliminaries (Section \ref{s:preliminaries}) for definitions.}
  \begin{itemize}
  \item[\emph{(\emph{a})}] \textbf{Hardness of approximate \emph{MCSP} against almost-linear size circuits.}\\There exist $c \geq 1$, $0 < \gamma < 1$, and
    $\varepsilon > 0$ such that $\MCSP[(n^c,0), (2^{n^\gamma}, n^{-c})] \notin \Circuit[N^{1 + \varepsilon}]$.
  \item[\emph{(\emph{b})}] \textbf{Hardness of worst-case \emph{MCSP} against almost-linear size circuits.}\\There exists $c \geq 1$ and $\varepsilon > 0$
    such that $\MCSP[n^c,2^n/n^c] \notin \Circuit[N^{1 + \varepsilon}]$.
  \item[\emph{(\emph{c})}] \textbf{Hardness of sub-exponential size learning using non-adaptive queries.}\\
    There exist $\ell \geq 1$ and $0 < \gamma < 1$ such that $\Circuit[n^\ell]$ cannot be learned up to error $O (1/n^{\ell})$ under the uniform distribution by circuits of size $2^{O(n^\gamma)}$ using non-adaptive membership queries.
  \item[\emph{(\emph{d})}] \textbf{Non-existence of natural properties against polynomial size circuits.}\\
    For some $d \geq 1$ there is no $\Circuit[\mathsf{poly}(N)]$-natural property useful against $\Circuit[n^d]$. 
  \item[\emph{(\emph{e})}] \textbf{Existence of non-uniform \emph{PRFs} secure against sub-exponential size circuits.}\\
    For every constant $a \geq 0$, there exists $d \geq 1$, a sequence $\mathfrak{F} = \{\cF_n\}_{n \geq 1}$ of families $\cF_n$ of $n$-bit boolean functions $f_n \in \Circuit[n^d]$, and a sequence of probability distributions $\mathfrak{D} = \{\cD_n\}_{n \geq 1}$ supported over $\mathcal{F}_n$ such that, for infinitely many values of $n$, $(\cF_n, \cD_n)$ is pseudo-random function family that $(1/N^{\omega(1)})$-fools \emph{(}oracle\emph{)} circuits of size $2^{a \cdot n}$.
  \end{itemize}
\end{theorem}

\vspace{0.2cm}

\noindent The proof of this result appears in Section \ref{ss:proof_equivalences}. We highlight below the most interesting  implications of Theorem \ref{thm:equivalences}. (Note that some of them have appeared in other works in similar or related forms.)

\begin{itemize}
\item[(\emph{a}) $\rightarrow$ (\emph{d}):] The initial hardness magnification result from \citep[Theorem 1]{OS18_mag_first} (stated for circuits) implies the \emph{non-existence} of natural proofs useful against polynomial-size circuits, indicating that the natural proofs barrier might not be relevant to the magnification approach. 
\item[(\emph{a}), (\emph{b}) $\leftrightarrow$ (\emph{d}):] \emph{Any} \Ppoly natural property useful against \Ppoly can be transformed into an almost-linear size natural property that is simply the approximate $\MCSP[(n^c,0),(2^{n^\gamma}, n^{-c})]$ or worst-case gap $\MCSP[n^c, 2^n/n^c]$. (Note the different regime of circuit size parameters for these problems.)
\item[(\emph{a}), (\emph{b}) $\leftrightarrow$ (\emph{c}):] A weak-seeming hardness assumption for worst-case gap and approximate versions of $\MCSP$ implies a strong non-learnability result: polynomial-size circuits cannot be learned over the uniform distribution even non-uniformly in sub-exponential time.
\item[(\emph{a}), (\emph{b}) $\leftrightarrow$ (\emph{e})] Hardness magnification for $\MCSP$ also yields cryptographic hardness in a certain regime.
\end{itemize}

We note that the use of non-adaptive membership queries in Theorem \ref{ss:proof_equivalences} Item (c) is not essential. It follows from \citep{DBLP:conf/coco/CarmosinoIKK16} that, in the context of learnability of polynomial size circuits under the uniform distribution in sub-exponential time, adaptive queries are not significantly more powerful than non-adaptive queries.\footnote{In a bit more detail, one can easily extract a natural property from a learner that uses adaptive queries. In turn, closer inspection of the technique of \citep{DBLP:conf/coco/CarmosinoIKK16} shows that a non-adaptive learner can be obtained from a natural property.}

\vspace{0.1cm}

\vspace{0.2cm}

\noindent \textbf{Towards a more robust theory.} While Theorem \ref{thm:equivalences} formally connects hardness magnification and natural properties, it would be very interesting to understand to which extent different hardness magnification theorems are provably non-naturalizable. This would provide a more complete answer to Question Q1 asked above. For instance, Theorem \ref{thm:equivalences} leaves open whether hardness magnification for worst-case versions of \MCSP such as $\mathsf{MCSP}[n^c, 2^{n^\varepsilon}]$ refutes natural proofs as well. Note that one way of approaching this question would be to study reductions from $\mathsf{MCSP}[n^{c}, 2^{n^\gamma}]$ to its approximate version $\MCSP[(n^{c'},0), (2^{n^{\gamma'}}, n^{-c'})]$.\footnote{More precisely, the existence of a reduction from $\mathsf{MCSP}[n^{c}, 2^{n^\gamma}]$ to $\MCSP[(n^{c'},0), (2^{n^{\gamma'}}, n^{-c'})]$ shows that lower bounds for the former problem yield lower bounds for the latter. Since any such lower bound must be non-naturalizable by Theorem \ref{thm:equivalences}, we obtain the same consequence for $\mathsf{MCSP}[n^{c}, 2^{n^\gamma}]$. (Note that in the context of hardness magnification it is also important to have highly efficient reductions.)} In Section \ref{ss:redmcsp}, we observe that this question is related to the problem of basing \emph{hardness of learning} on \emph{worst-case assumptions} such as $\Ptime\ne\NP$ (cf.~\citep{DBLP:conf/focs/ApplebaumBX08}). We refer to the discussion in Section \ref{ss:redmcsp} for more details.

\subsection{The Locality Barrier} \label{ss:locality_barrier}

The results from the preceding section show that hardness magnification can go beyond natural proofs. Is there another barrier that makes it difficult to establish lower bounds via magnification? In this section, we present a general argument to explain why the lower bound techniques behind A3-E3, A4-D4 in the magnification frontiers from Section \ref{ss:intro_background_motivation} cannot be adapted (without significantly new ideas) to establish the required lower bounds in Items A1-E1, respectively. We refer to it as the \emph{locality barrier}. While we will focus on these particular examples to make the discussion concrete, we believe that this barrier applies more broadly.

In order to explain the locality barrier, let's consider the argument behind the proof of B1 presented in Section \ref{ss:frontier_MCSP_formula_xor}. Recall that this result shows that if $\MCSP[2^{n^{1/3}},2^{n^{2/3}}]\notin\Formula$-$\mathsf{XOR}[N^{1.01}]$ then $\mathsf{NQP}\not\subseteq\NC$. This and other known hardness magnification theorems are established in the contrapositive.  The core of the argument is to prove that there are highly efficient $\mathsf{Formula}$-$\mathsf{XOR}$ circuits that reduce an input to $\MCSP[2^{n^{1/3}},2^{n^{2/3}}]$ of length $N = 2^n$ to deciding whether certain strings of length $N'$ (much smaller than $N$) belong to a certain language $L'$. Then, under the assumption that $\mathsf{NQP} \subseteq \NC$, one argues that $L'$ has polynomial size formulas. Finally, since $N' \ll N$, we can employ such formulas and still conclude that $\MCSP[2^{n^{1/3}},2^{n^{2/3}}]$ is in $\mathsf{Formula}$-$\mathsf{XOR}[N^{1.01}]$, which completes the proof.

Note that the argument above provides a \emph{conditional} construction of highly efficient formulas for the original problem. Crucially, however, we can derive an \emph{unconditional} circuit upper bound from this argument: If we stop right before we replace the calls to $L'$ by an algorithm for $L'$ (this is what makes the reduction conditional), it unconditionally follows that $\MCSP[2^{n^{1/3}},2^{n^{2/3}}]$ can be computed by highly efficient $\mathsf{Formula}$-$\mathsf{XOR}$ circuits containing oracle gates of small fan-in, for some oracle. Similarly, one can argue that the problems in Items A1-E1 can be computed in the respective models by highly efficient Boolean devices containing oracles of small fan-in.

We stress that, as opposed to a magnification theorem, where one cares about the complexity of the oracle gates, in our discussion of the locality barrier we only need the fact that there is \emph{some way} of setting these oracles gates so that the resulting circuit or formula solves the original problem. (A definition of this model appears in Section \ref{ss:def_local}.) A more exhaustive interpretation of magnification theorems as construction of circuits with small fan-in oracles can be found in Appendix \ref{ss:kernelization}. 

On the other hand, we argue that the lower bound arguments from Items A3-E3 of the hardness magnification frontiers quite easily handle (in the respective models) the presence of oracles of small fan-in, \emph{regardless of the function} computed by these oracles. Using a more involved argument we can localize also lower bounds from items A4-D4. Consequently, these methods do not seem to be refined enough to prove the lower bounds required by A1-D1 without excluding oracle circuits that are unconditionally known to exist for the corresponding problems.

Following the example above, we state our results for the Magnification Frontier B.

\begin{theorem}[Locality Barrier for HM Frontier B]
  \label{thm:frontier_C} The following results hold.
 \begin{itemize}
\item[\emph{(B1}$^\mathcal{O}$\emph{)}] \emph{(Oracle Circuits from Magnification)} For any $\varepsilon > 0$, $\mathsf{MCSP}[2^{n^{1/3}}, 2^{n^{2/3}}] \in \mathsf{Formula}$-$\mathcal{O}$-$\mathsf{XOR}[N^{1.01}]$, where every oracle $\mathcal{O}$ has fan-in at most $N^\varepsilon$ and appears in the layer right above the $\mathsf{XOR}$ leaves.
		\item[\emph{(B3}$^\mathcal{O}$\emph{)}] \emph{(Extension of Lower Bound Techniques)} For any $\delta > 0$, $\mathsf{InnerProduct}$ over $N$ input bits cannot be computed by $N^{2-3\delta}$-size $\mathsf{Formula}$-$\mathcal{O}$-$\mathsf{XOR}$ circuits with at most $N^{2-3\delta}$ oracle gates of fan-in $N^{\delta}$ in the layer right above the $\mathsf{XOR}$ leaves, for any oracle $\mathcal{O}$.
\item[\emph{(B4}$^\mathcal{O}$\emph{)}] \emph{(Extension of  Lower Bound Techniques Below Magnification Threshold)} There is a universal constant $c$ such that for all constants $\eps > 0$ and $\alpha > 2$, $\MCSP[n^{c},2^{\eps/\alpha \cdot n}]$ cannot be computed by oracle formulas $F$ with $\SIZE_3(F) \le N^{2-\eps}$ and adaptivity $o(\log N/\log\log N)$.\footnote{That is, on any path from root to a leaf, there are at most $o(\log N / \log\log N)$ oracles.}
 \end{itemize}
\end{theorem}

Here, $\Size_t(F)$ denotes the size of the formula, if we replace every oracle $\mathcal{O}$ with fan-in $\beta$ in $F$ by a formula of size $\beta^t$ which reads all its inputs exactly $\beta^{t-1}$ times (see Section~\ref{ss:HS_lb} for the motivation of this definition). 
\medskip

The first two items of Theorem \ref{thm:frontier_C} are proved in Section \ref{ss:Formula_XOR_Tal}. The third item is proved in Section \ref{ss:HS_lb}.  While Theorem \ref{thm:frontier_C} does not specify that, we actually localize all proofs of the lower bounds from B3 and B4 we are aware of. Interestingly, the localization of B4 allows us to refute the Anti-Checker Hypothesis from \cite{OPS19_CCC} (and a family of potential hardness magnification theorems), cf.~Section \ref{ss:HS_lb}. We refer to Section  \ref{s:difficulties} for analogous statements describing the locality barrier in frontiers A, C, D, and E.\\

\noindent \textbf{Locality of Computations and Lower Bound Techniques.} The fact that many lower bound techniques extend to computational devices with oracles of small fan-in was observed already by Yao in 1989 on a paper on local computations \citep{Yao89}. According to Yao, a local function is one that can be efficiently computed using only localized processing elements. In our terminology, this corresponds to circuits with oracles of small fan-in. Among other results, 
\citep{Yao89} argues that Razborov's monotone circuit size lower bound for $k$-$\mathsf{Clique}$ \citep{Razborov:85a} and Karchmer and Wigderson's monotone formula size lower bound for $\mathsf{ST}$-$\mathsf{CONN}$ \citep{DBLP:journals/siamdm/KarchmerW90} extend to boolean devices with monotone oracles of bounded fan-in. Compared to Yao's work, our motivation and perspective are different. While Yao is particularly interested in lower bounds that can be extended in this sense (see e.g.~Sections 2 and 6 in \citep{Yao89}), here we view such extensions as a \emph{limitation} of the corresponding arguments, meaning that they are not refined enough to address the locality barrier.\footnote{On a more technical level, we are interested in the regime of barely super-linear size circuits and formulas, and our results do not impose a monotonicity constraint on the oracle.}

We note, however, that not every lower bound technique extends to circuits with small fan-in oracles.\footnote{Of course any such discussion depends on parameters such as number of oracles and their fan-ins, so whether a technique avoids or not the locality barrier is relative to a particular magnification theorem.} For instance, by the work of Allender and Kouck{\'{y}}~\cite{AK10} (also a more recent work by Chen and Tell~\cite{CT19_STOC}), the parity function $\mathsf{Parity}_n$ over $n$ input bits can be computed by a $\mathsf{TC}^0$ circuit of size $O(n)$ (number of wires) containing $\leq n^{1 - \varepsilon}$ oracle gates of fan-in $\leq n^\varepsilon$, provided that its depth $d = O(1/\varepsilon)$. On the other hand, it is known that $\mathsf{Parity}_n \notin \mathsf{TC}^0_d[n^{1 + c^{-d}}]$ for a constant $c > 0$ \citep{DBLP:journals/siamcomp/ImpagliazzoPS97} (again, the complexity measure is the number of wires). Since the latter lower bound is super-linear for every choice of $d$, it follows by the result of \citep{AK10,CT19_STOC} that it cannot be extended to circuits containing a certain number of oracles of fan-in $n^\varepsilon$, for a large enough depth $d$ that depends on $\varepsilon$. Incidentally, the hardness magnification theorems of \citep{AK10,CT19_STOC} do not achieve a magnification frontier.

In \cref{ss:frontier_MCSP_formula_xor} we identify one specific lower bound related to HM frontier D which is both above the magnification threshold and provably non-localizable, cf.~Theorem \ref{thm:E-lower-bound-for-FM}. In principle, there might be ways to overcome the locality barrier and match the lower bound with the magnification threshold. We refer to Section \ref{ss:conclusion} below for  additional discussion.\\


\noindent \textbf{On Lower Bounds Through Reductions.} The discussion above has focused on the possibility of \emph{directly} adapting existing lower bounds from Item 3 in HM frontier to establish the desired lower bound in Item 1. There is however an \emph{indirect} approach that one might hope to use: \emph{reductions}. For instance, in the context of the HM Frontier B discussed above, can we have a reduction from $\mathsf{InnerProduct}$ to $\mathsf{MCSP}[2^{n^{1/3}}, 2^{n^{2/3}}]$ that would allow us to show that $\mathsf{MCSP}[2^{n^{1/3}}, 2^{n^{2/3}}] \notin \mathsf{Formula}$-$\mathsf{XOR}[N^{1.01}]$? The first thing to notice is that, for this approach to make sense, the reduction needs to have a specific form 
so that composing the reduction with a candidate $\mathsf{Formula}$-$\mathsf{XOR}$ circuit for $\mathsf{MCSP}[2^{n^{1/3}}, 2^{n^{2/3}}]$ violates the hardness of $\mathsf{InnerProduct}$. 
Is there any hope to design a reduction of this form? 

The locality barrier presents a \emph{definitive answer} in this case. Indeed, it is immediate from the first two items of Theorem \ref{thm:frontier_C} that such a reduction \emph{does not exist}. For the same reason, it is not possible to use reductions to establish the required lower bounds in some other magnification frontiers, cf. Section \ref{ss:reductions}.

\subsection{Concluding Remarks and Open Problems}\label{ss:conclusion}

Hardness magnification shows that obtaining a refined understanding of weak computational models is an approach to major complexity lower bounds, such as separating $\mathsf{EXP}$ from $\mathsf{NC}^1$. As discussed in Sections \ref{ss:intro_background_motivation} and \ref{ss:non_natural} above, its different instantiations are connected to a few basic questions in Complexity Theory, including the power of non-monotone operations, learnability of circuit classes, and pseudorandomness.

One of our main conceptual contributions in this work is to identify a challenge when implementing this strategy for lower bounds. Quoting the influential article \citep{DBLP:journals/jcss/RazborovR97}} that introduced the natural proofs barrier,

\vspace{0.25cm}

{
\emph{``We do not conclude that researchers should give up on
proving serious lower bounds. Quite the contrary, by classifying a large number of techniques that are unable to do the
job we hope to focus research in a more fruitful direction.''}
}
\vspace{-0.35cm} 
\begin{flushright}
{ \small Razborov and Rudich \citep[Section 6]{DBLP:journals/jcss/RazborovR97}}
\end{flushright} 

\vspace{0.1cm}

\noindent We share a similar opinion with respect to hardness magnification and the obstruction identified in Section \ref{ss:locality_barrier}. While locality provides a unified explanation for the difficulty of adapting combinatorial lower bound techniques to exploit most (if not all) known magnification frontiers, it might be possible to discover new HM frontiers whose associated lower bound techniques in Item 3 are sensitive to the presence of small fan-in oracles. For instance, in the case of \emph{uniform} complexity lower bounds, this has been achieved in \citep{DBLP:conf/icalp/Oliveira19} via an indirect diagonalization that explores the theory of pseudorandomness.\footnote{In other words, the magnification theorem discussed in \citep{DBLP:conf/icalp/Oliveira19} admits a formulation for uniform randomized algorithms, and its proof provides an algorithm with oracle gates of small fan-in in the spirit of the oracle circuits discussed here. Nevertheless, the unconditional lower bound established in the same paper does not extend to algorithms with such oracle gates.} Alternatively, it might be possible to establish magnification theorems using a technique that does not produce circuits with small fan-in oracles. Even if one is pessimistic about these possibilities, we believe that an important contribution of the theory of hardness magnification is to break the divide between ``weak'' and ``strong'' circuit classes advocated by the natural proofs barrier, and that it deserves further investigation. 

We finish with a couple of technical questions related to our contributions. First, we would like to understand if it is possible to strengthen items (\emph{a}) and (\emph{b}) in Theorem \ref{thm:equivalences} to a wider range of parameters. For example, is hardness magnification for worst-case $\MCSP[n^c,2^{n^{\gamma}}]$ with $\gamma < 1$ non-naturalizable? The core of this question seems to be the problem of reducing worst-case $\MCSP$ from item (\emph{a}) to approximate $\MCSP$ from item (\emph{b}).

Another important direction is to show that hardness magnification avoids natural proofs also in the context of \emph{non}-meta-computational problems. Interestingly, many magnification theorems from \citep{OPS19_CCC} established for $\MCSP$ and variants were subsequently shown to hold for any  sparse language in $\mathsf{NP}$ \citep{Magnification_FOCS19}. Could it be the case that hardness magnification overcomes natural proofs in a much broader sense? 

Finally, it would be useful to investigate the locality of additional lower bound techniques.  Can we, for example, come up with non-localizable lower bounds similar to Theorem \ref{thm:E-lower-bound-for-FM} which would be above the magnification threshold and work for a problem more closely related to the one from the correponding HM frontier?

\section{Preliminaries}\label{s:preliminaries}

\subsection{Notation}

Given a Boolean function $f:\{0,1\}^n\rightarrow\{0,1\}$, $\ttable(f)$ denotes the $2^n$-bit string representing the truth table of $f$. On the other hand, for any
string $y \in \{0,1\}^{2^n}$, define $f_y$ as the function on $n$ inputs such that $\ttable(f_y) = y$. 

$\Circuit[s]$ denotes fan-in two Boolean circuits of size at most $s$ where we count the number of gates. $\Formula[s]$ denotes formulas over the basis $U_2$
(fan-in two ANDs and ORs) of size at most $s$ (counting the number of leaves) with input leaves labelled by literals or constants. 

For a circuit class $\mathcal{C}$, $\mathcal{C}[s]$ denotes circuits from $\mathcal{C}$ of size at most $s$.

A function $f:\{0,1\}^n\rightarrow\{0,1\}$ is $\gamma$-approximated by a circuit $C$, if $\Pr_x[C(x)=f(x)]\ge\gamma$. 

\subsection{Complexity of Learning}
\begin{definition}[Learning]
A circuit class $\cC$ is learnable over the uniform disribution by circuits in $\cD$ up to error $\eps$ with confidence $\delta$ if there are randomized oracle $\cD$-circuits $L^f$ such that for every boolean function $f:\{0,1\}^n\mapsto\{0,1\}$ computable by a circuit from $\cC$, when given oracle access to $f$, input $1^n$ and the internal randomness $w \in \{0,1\}^*$, $L^f$ outputs the description of a circuit satisfying 
  \begin{equation*}
    \Pr_w \{ L^f(1^n, w) \text{ } (1-\eps) \text{-approximates } f \} \geq \delta.
  \end{equation*}

\noindent $L^f$ uses  non-adaptive membership queries if the set of queries which $L^f$ makes to the oracle does not depend on the answers to previous queries.  If $\delta=1$, we omit mentioning the confidence parameter.
\end{definition}

\subsection{Natural Properties, MCSP, and its Variants}
Let $\cF_n$ be the set of all functions on $n$ variables. $\mathfrak{R} = \{ \cR_n \subseteq \cF_n \}_{n \in \N}$ is a combinatorial property of Boolean functions. 
\begin{definition}[Natural property \cite{DBLP:journals/jcss/RazborovR97}]
  \label{def:nat_prop}
  Let $\mathfrak{R} = \{ \cR_n \}$ be a combinatorial property, $\cC$ be a circuit class and $\Gamma$ be a complexity class. Then,
  $\mathfrak{R}$ is a $\Gamma$-natural property useful against $\cC[s(n)]$, if there exists an $n_0 \in \N$ such that the following hold: 
\begin{itemize}
\item[$\bullet$] \textbf{Constructivity} : For any function $f_n \in \cF_n$, the predicate $f_n \stackrel{?}{\in} \cR_n$ is computable in $\Gamma$ in the size
  of the truth table of $f_n$.
\item[$\bullet$] \textbf{Largeness} : For every $n \geq n_0$, $\Pr_{f_n \sim \cF_n} \{f_n \in \cR_n \} \geq \frac{1}{2^{O(n)}}$.
\item[$\bullet$] \textbf{Usefulness} : For every $n \geq n_0$, $\cR_n \cap \cC[s(n)] = \emptyset$. 
\end{itemize} 
\end{definition}

The following result which follows from \cite{DBLP:conf/coco/CarmosinoIKK16} connects the existence of natural properties useful against a class $\cC$ to designing learning algorithms for
$\cC$. 
\begin{theorem}[From Theorem 5.1 of \cite{DBLP:conf/coco/CarmosinoIKK16} and Lemma 14 of \cite{IW01}]
  \label{thm:cikk_learn}
  Let $R$ be a $\Ppoly$-natural property useful against $\Circuit[n^d]$ for some $d \geq 1$. Then, for each $\gamma\in (0,1)$, there are randomized, oracle circuits $\{D_n\}_{n \geq 1} \in \Circuit[2^{O(n^{\gamma})}]$ 
that learn $\Circuit[n^k]$ up to error $\frac{1}{n^k}$ using non-adaptive oracle queries to $f_n$, where $k = \frac{d\gamma}{a}$ and $a$ is a universal constant that does not depend on $d$ and $\gamma$.
\end{theorem}

\begin{definition}[Gap MCSP]
  \label{def:ag_mcsp}
  Let $s,t : \N \rightarrow \N$, where $s(n) \leq t(n)$ and $0 \leq \eps, \sigma < 1/2$. Define $\MCSP[(s,\sigma),(t, \eps)]$ on inputs of length $N = 2^n$, as the following promise problem :
  \begin{itemize}
  \item YES instances: $y \in \{0,1\}^{N}$ such that there exists a circuit of size $s(n)$ that $(1-\sigma)$-approximates $f_y$. 
  \item NO instances: $y \in \{0,1\}^{N}$ such that no circuit of size $t(n)$ $(1-\eps)$-approximates $f_y$.
  \end{itemize}
\end{definition}
We refer to $\MCSP[(s,0),(t,0)]$ as $\MCSP[s,t]$. 
Informally, if $\eps>0$, we say that $\MCSP[(s,0),(t, \eps)]$ is an {\em approximate} version of \MCSP. Otherwise, it is a {\em worst-case} version of \MCSP. 

\begin{remark}
 In Definition \ref{def:ag_mcsp}, if $s(n)=t(n)$, we also require that $\sigma < \eps$ for the yes and no instances to be disjoint. 
\end{remark}

\begin{definition}[Succinct \MCSP] For functions $s,t:\mathbb{N}\mapsto\mathbb{N}$, $\Succinct[s(n),t(n)]$ is the following problem. Given an input $\langle 1^n,1^s,
  (x_1,b_1),\dots,(x_t,b_t)\rangle$ where $x_i\in\{0,1\}^n, b_i\in\{0,1\}$, decide if there is a circuit $C$ of size $s$ such that $\bigwedge_{i=1,\dots,t} C(x_i)=b_i$.
\end{definition}

\subsection{Pseudorandom Generators}
\begin{definition}[Pseudorandom function families]
  \label{def:prf}
  For any circuit class $\cC$, size functions $s(n), t(n) \geq n$, family $\cG_n$ of $n$-bit Boolean functions and distribution $\cD_n$ over $\cG_n$, we say that a
  pair $(\cG_n, \cD_n)$ is a $(t(n), \eps(n))$-pseudorandom function family (PRF) in $\cC[s(n)]$, if each function in $\cG_n$ is in $\cC[s(n)]$ and for every randomized circuit $A^O \in \Circuit^O[t(n)]$, where $O$ denotes oracle access to a fixed Boolean function over $n$ inputs, we have
  \begin{equation*}
    \bigg \vert \Pr_{g \sim \cD_n, w} \{A^g(w) = 1\} - \Pr_{f \sim \cF_n, w} \{A^f(w) = 1 \} \bigg \vert \leq \eps(n)
  \end{equation*}
\end{definition}

\cite{DBLP:conf/coco/OliveiraS17} state an equivalence between the non-existence of PRFs in a circuit class $\cC$ and learning algorithms for $\cC$. In particular, we care about the following direction which they prove using a small-support version of Von-Neumann's Min-max Theorem.
\begin{theorem}[No PRFs in $\cC$ implies Learning Algorithm for $\cC$ \cite{DBLP:conf/coco/OliveiraS17}]
  \label{thm:no_prf_implies_learning}
  Let $t(n) \leq 2^{O(n)}$. Suppose that for every $k \geq 1$ and large enough $n$, there exists no $(\poly(t(n)),1/10)$-pseudorandom function families in $\cC[n^k]$. Then, for every $\eps >0, k \geq 1$ and large enough $n$, there is a randomized oracle circuit in $\Circuit^O [2^{n^\eps}]$ that learns every function $f_n \in \cC[n^k]$ up to error $1/n^k$ with confidence $1-1/n$, where $O$ denotes membership query access to $f_n$.
\end{theorem}

\subsection{Local Circuit Classes}\label{ss:def_local}

Our definition of local computation is somewhat similar to some definitions appearing in \citep{Yao89}.

\begin{definition}[Local circuit classes] \label{d:local_classes}
  Let $\mathcal{C}$ be a circuit class \emph{(}such as $\mathsf{AC}^0[s]$, $\mathsf{TC}_d^0[s]$, $\mathsf{Circuit}[s]$, etc\emph{)}. For functions $q, \ell, a \colon
  \mathbb{N} \to \mathbb{N}$, we say that a language $L$ is in $[q,\ell, a]$--$\,\mathcal{C}$ if there exists a sequence $\{E_n\}$ of oracle circuits for which the
  following holds:
  \begin{itemize}
  \item[\emph{(\emph{i})}] Each oracle circuit $E_n$ is a circuit from $\mathcal{C}$.
  \item[\emph{(\emph{ii})}] There are at most $q(n)$ oracle gates in $E_n$, each of fan-in at most $\ell(n)$, and any path from an input gate to an output gate
    encounters at most $a(n)$ oracle gates.
  \item[\emph{(\emph{iii})}] There exists a language $\mathcal{O} \subseteq \{0,1\}^*$ such that the sequence $\{E_n^\mathcal{O}\}$ \emph{(}$E_n$ with its oracle gates
    set to $\mathcal{O}$\emph{)} computes $L$.
\end{itemize}
\end{definition}

In the definition above, $q$ stands for \emph{quantity}, $\ell$ for \emph{locality}, and $a$ for \emph{adaptivity} of the corresponding oracle gates.

\subsection{Random Restrictions}

Let $\rho: [N] \to \{0,1,*\}$ be a \emph{restriction}, and $\boldsym{\rho}$ be a \emph{random restriction}, i.e., a distribution of restrictions. We say that $\boldsym{\rho}$ is $p$-regular if $\Pr[\boldsym{\rho}(i) = *] = p$ and $\Pr[\boldsym{\rho}(i) = 0] = \Pr[\boldsym{\rho}(i) = 1] = (1-p) / 2$ for every $i \in [N]$. We also say $\boldsym{\rho}$ is $k$-wise independent if any $k$ coordinates of $\boldsym{\rho}$ are independent. For a function $f: \{0,1\}^{N} \to \{0,1\}$, we use $f\restriction_{\rho}$ to denote the function $\{0,1\}^{|\rho^{-1}(*)|} \to \{0,1\}$ obtained by restricting $f$ according to $\rho$ in the natural way.

We need the following lemma stating that one can sample from a $k$-wise independent random restriction with a short seed, and moreover all restrictions have a small circuit description.

\begin{lemma}[\cite{ImpagliazzoMZ19,Vadhan12}]\label{lm:k-wise-independent-restrictions}
	There exists a $q$-regular $k$-wise independent random restriction $\boldsym{\rho}$ distributed over $\rho : [N] \to \{0,1,*\}$ samplable with $O(k \cdot \log(N)\log(1/q))$ bits. Furthermore, each output coordinate of the random restriction can be computed in time polynomial in the number of random bits.
\end{lemma}
\subsection{Technical Results}
\begin{lemma}[Hoeffding's inequality]
  \label{lem:hoeffding}
  Let $X_1, \dots, X_n$ be independent random variables such that $0 \leq X_i \leq 1$ for every $i \in [n]$. Let $X = \sum_{i=1}^n X_i$. Then, for any $\eps > 0$, we
  have
  \begin{equation*}
    \Pr \{ \vert X - \textbf{E}{X} \vert \geq \eps n \} \leq 2\exp(-2 \eps^2 n)
  \end{equation*}
\end{lemma}

\section{Magnification Frontiers}
\label{a:improved_magnification_MCSP}

\newcommand{\HMFrontierAC}[1][]{HM Frontier A{#1}\xspace}
\newcommand{\HMFrontierFormula}[1][]{HM Frontier B{#1}\xspace}
\newcommand{\HMFrontierClique}[1][]{HM Frontier E{#1}\xspace}
\newcommand{\HMFrontieralmostFormula}[1][]{HM Frontier C{#1}\xspace}
\newcommand{\HMFrontierRandFormula}[1][]{HM Frontier D{#1}\xspace}

\subsection{$\mathsf{EXP} \nsubseteq \mathsf{NC}^1$ and $\mathsf{AC}^0$-$\mathsf{XOR}$ Lower Bounds for $\mathsf{MKtP}$} \label{ss:frontier_ACXOR}

In this Section, we present the proofs of the new results stated in \HMFrontierAC.
Recall that $\Kt(x)$ is defined as the minimum over $|M| + \log t$ such that a program $M$ outputs $x$ in $t$ steps.
For thresholds $\theta, \theta' \colon \Nat \to \Nat$, we denote by $\MKtP[\theta(N), \theta'(N)]$
the promise problem whose YES instances consist of the strings $x \in \binset^N$ such that $\Kt(x) \le \theta(N)$
and NO instances consist of the strings such that $\Kt(x) > \theta'(N)$.

We start with the hardness magnification theorem of \HMFrontierAC{1}.
\begin{theorem}
 \label{thm_B1}
  There exists a constant $c$ such that, for every large enough constant $d > 1$,
  $$\MKtP[(\log N)^d, (\log N)^d + c \log N]\notin\AC\text{-}\mathsf{XOR} [N^{1.01}] \quad \text{implies}\quad \mathsf{EXP}\not\subseteq\NC. $$
\end{theorem}
\begin{proof}
  We prove the contrapositive.
  Assume that $\mathsf{EXP} \subseteq \NC$.
  First, recall that
  any $N$-bit-input polynomial-size $\NC$ circuit can be converted into a depth-$d'$ $\AC$ circuit of size $2^{N^{{O(1/d')}}}$ for every positive integer constant $d'$ (see, e.g., \cite[Lemma 8.1]{AHMPS08_siamcomp_journals}).

  Oliveira, Pich, and Santhanam \citep{OPS19_CCC} showed that there exists a problem $L \in \mathsf{EXP}$ such that 
  $\MKtP[\theta(N), \theta(N) + c \log N] \in \mathsf{AND}_{O(N)}$-$L_{O(\theta(N))}$-$\mathsf{XOR}$ for $\theta(N) \ge \log N$.
  (Here the subscript denotes the fan-in of a gate.)
  That is, the promise problem $\MKtP[\theta(N), \theta(N) + c \log N]$ can be computed by the following form of an $L$-oracle circuit:
  The output gate is an $\mathsf{AND}$ gate of fan-in $O(N)$,
  at the middle layer are $L$-oracle gates of fan-in $O(\theta(N))$,
  and
  at the bottom layer are $\mathsf{XOR}$ gates.
  Under the assumption that $\mathsf{EXP} \subseteq \NC$, we can replace $L$-oracle circuits with depth-$d'$ $\AC$ circuits of size $2^{(\log N)^{O(d / d')}}$,
  which is smaller than $N^{0.01}$ by choosing a constant $d'$ large enough.
  In particular, we obtain a depth-$(d' + O(1))$ almost linear size $\AC$ circuit with bottom $\mathsf{XOR}$ gates that computes $\MKtP[\theta(N), \theta(N) + c \log N]$.
\end{proof}

The rest of this section is devoted to proving 
the following $\AC$ lower bound for $\MKtP$,
which establishes \HMFrontierAC{4}.
\begin{theorem}
  \label{theorem_MKtPversusAC}
  For any $d = d(N)$,
  for some $\theta(N) = d \cdot \widetilde O(\log N)^3$
  and any $\theta'(N) = N / \omega( \log N )^d$,
  it holds that
  $\MKtP[ \theta(N), \theta'(N) ] \not\in \AC_d$.
\end{theorem}
\noindent
Note that \cref{theorem_MKtPversusAC} is only meaningful if $d = o(\log N/ \log \log N)$,
because otherwise the promise problem is not well-defined.

The idea of the proof is as follows:
Trevisan and Xue \cite{TrevisanX13_coco_conf} showed that there exists a pseudorandom restriction $\rho$ of seed length $\mathsf{polylog}(N)$ that shrinks every polynomial-size depth-2 circuit into shallow decision trees.
Moreover, the expected fraction of unrestricted variables $\rho^{-1}(*)$ is at least $p = \Omega(1 / \log N)$.
In particular, by composing $d$ independent pseudorandom restrictions $\rho_1, \cdots, \rho_d$,
every depth-$d$ circuit can be turned into a constant function, while still leaving at least $p^d$-fraction of inputs unrestricted.
The seed length required to sample $d$ independent pseudorandom restrictions is at most $d \times \mathsf{polylog}(N)$,
and thus $\mathsf{Kt}(0^N \circ \rho) \le \mathsf{polylog}(N)$.
We stress that the exponent of the seed length does not depend on $d$.
Since the circuit hit with the pseudorandom restriction becomes a constant function, 
it cannot distinguish $0^N \circ \rho$ with $U_N \circ \rho$, i.e., the distribution where the unrestricted variables of $\rho$ are replaced with the uniform distribution $U_N$.
Assuming that there remain sufficiently many unrestricted inputs (e.g., $N / O(\log N)^d \gg \mathsf{polylog}(N)$),
the latter distribution has a large $\mathsf{Kt}$ complexity, which is a contradiction to the fact that the circuit computes a gap version of $\MKtP$.

We note that
Cheraghchi, Kabanets, Lu, and Myrisiotis \cite{CheraghchiKLM19_eccc_journals} used the pseudorandom restriction method in order to obtain an exponential-size $\AC$ lower bound.
A crucial difference in this work is that
instead of optimizing the size of $\AC$ circuits, we aim at minimizing the threshold $\theta$ of $\MKtP[\theta]$.

  Following \cite{TrevisanX13_coco_conf},
  in order to generate a random restriction $\rho \in \{0, 1, *\}^N$
  that leaves a variable unrestricted with probability $2^{-q}$,
  we regard a binary string $w \in \binset^{(q+1)N}$ as a random restriction  $\rho_w$.  Specifically:
\begin{definition}
  For a string $w \in \binset^{(q+1)N}$,
  we define a restriction $\rho_w \in \{0, 1, *\}^N$
  as follows:
  Write $w$ as $(w_1, b_1) \cdots (w_N, b_N)$, where $w_i \in \binset^{q}$ and $b_i \in \binset$.
  For each $i \in [N]$,
  if $w_i = 1^q$ then set $\rho_w(i) := *$;
  otherwise, set $\rho_w(i) := b_i$.
\end{definition}
\noindent
  Note that this is defined so that
  $\Pr_{w}[\rho_w(i) = *] = 2^{-q}$ for every $i \in [N]$,
  when $w$ is distributed uniformly at random.

  Trevisan and Xue \cite{TrevisanX13_coco_conf} showed that H{\aa}stad's switching lemma can be derandomized by using a distribution that fools CNFs. To state this formally, we need the following definitions. Define a $t$-width CNF as one which has at most $t$ literals in each clause. We say that a distribution $\mathcal{D}$ over $\{0,1\}^n$ $\varepsilon$-fools a set of functions $\mathcal{S}_n$ over $n$ variables if for every $f \in \mathcal{S}_n$, $ \vert \Pr_{x \sim D} \{f(x)=1\} - \Pr_{x \sim U_n} \{f(x)=1\} \vert \leq \varepsilon$. Finally, define $\DT(f)$ as the depth of the smallest decision tree computing $f$. 
\begin{lemma}
  [Derandomized Switching Lemma {\cite[Lemma 7]{TrevisanX13_coco_conf}}]
  \label{lemma_DerandomizedSwitchingLemma}
  Let $\varphi$ be a $t$-width $M$-clause CNF formula over $N$ inputs.
  Let $p = 2^{-q}$ for some $q \in \Nat$.
  Assume that a distribution $\mathcal D$ over $\binset^{(q+1)N}$
  $\eps_0$-fools $M \cdot 2^{t(q+1)}$-clause CNFs.
  Then,
      $$\Pr_{w \sim \mathcal D} [ \DT(\varphi\restriction_{\rho_w}) > s ]
      \le 2^{s+t+1} (5pt)^s + \eps_0 \cdot 2^{(s+1)(2t + \log M)}.
      $$
\end{lemma}

\begin{theorem}
  [{Based on \cite{TrevisanX13_coco_conf} and \cite[Theorem 56]{Tal17_coco_conf}}]
  \label{theorem_PseudorandomRestriction}
  Let $s, M, d, N \in \mathbb N$ be positive integers.
  Let $p = 2^{-q}$ for some $q \in \mathbb N$ so that $1/128s \le p < 1/64s$.
  Assume that there is a pseudorandom generator
  $G \colon \binset^r \to \binset^{(q+1)N}$
  that $\eps_0$-fools CNFs of size $M \cdot 2^s \cdot 2^{s(q+1)}$.
  Then,
  there exists a distribution $\mathcal R$ of random restrictions
  that satisfies
  the following:
  \begin{enumerate}
    \item
      \label{item_ProbabilityCircuitShrinks}
      For every circuit $C$ of size $M$ and depth $d$ over $N$ inputs,
      $$\Pr_{\rho \sim \mathcal R} [ \DT(C\restriction_\rho) >s ]
      \le M \cdot \left(
        2^{-s + 1} + \eps_0 \cdot 2^{(s+1)(3s + \log M)}
      \right).
      $$
    \item
      \label{item_ALotOfStars}
      For any parameter $\delta < 1$,
      with probability at least $1 - N ( \delta + d \eps_0 )$,
      the number of unrestricted variables in $[N]$
      is at least $\lfloor N \cdot p^{d-1} / 64\log(1 / \delta) \rfloor$.
    \item
      \label{item_ConstructionOfPseudorandomRestriction}
      $\mathcal R$ can be generated by a seed of length $dr$ in polynomial time.
  \end{enumerate}
\end{theorem}

\begin{proof}
  We apply the derandomized switching lemma (\cref{lemma_DerandomizedSwitchingLemma}) $d$ times.
  In the first iteration, 
  we set $p := 1 / 64$ (and $q := 6$)
  and generate $\rho_{G(z)[1, \cdots, (6+1)N]}$.
  (Here we use the first $(6+1)N$ bits of $G(z)$ to generate $\rho_{G(z)}$.)
  This turns a circuit $C$ of size $M$ into a circuit whose 
  bottom fan-in is at most $s$.
  For every other iteration $i$ (where $i = 2, \cdots, d$), we set $p := 2^{-q}$ and 
  turn a circuit $C$ of depth $d-i+2$ into a circuit of depth $d-i+1$.
  Our final pseudorandom restriction $\rho \sim \mathcal R$ is defined by
  the composition of
  the $d$ independent pseudorandom restrictions $\rho_{G(z_1)[1, \cdots, (6+1)N]}, \rho_{G(z_2)}, \cdots, \rho_{G(z_d)}$.

  Our proof is essentially the same with \cite{Tal17_coco_conf},
  except that
  (1) we apply the switching lemma $d$ times (instead of $d-1$)
  in order to turn depth-$d$ circuits into shallow decision trees,
  and (2) 
  in \cite{TrevisanX13_coco_conf, Tal17_coco_conf},
  for the application of constructing a pseudorandom generator for $\AC$,
  fixed bits of pseudorandom restrictions must be generated by using
  truly random bits,
  whereas in our case we generate all the bits by using $G$.

  In more detail,
  for each $i \in [d]$,
  let $M_i$ be the number of the gates at level $i$ in $C$
  (i.e., the gates whose distance from the input gates is $i$).
  At the first iteration,
  we set $p := 1 / 64 = 2^{-6}$ and $q := 6$.
  We then generate $\rho^1 := \rho_{G(z_1)[1, \cdots, (6+1)N]}$
  by choosing a seed $z_1 \sim \binset^{r}$ uniformly at random.
  We regard $C$ as a depth-$(d+1)$ circuit of bottom fan-in $1$,
  and apply \cref{lemma_DerandomizedSwitchingLemma} to each gate
  at level 1 (in the original circuit $C$).
  The probability that there exists a gate at level $1$ in $C \restriction_{\rho^1}$
  that cannot be computed by a decision tree of depth $s$
  is bounded above by 
  $$
  M_1 \cdot \left( 2^{s+1+1} ( 5 / 64 )^s + \eps_0 \cdot 2^{(s+1)(2 + \log M)} \right).
  $$
  In the complement event, 
  each gate at level $1$ can be written as DNFs and CNFs of width $s$, and hence
  can be merged into some gate at level $2$.
  Thus a circuit $C \restriction_{\rho^1}$ can be turned into a circuit of depth $d$ and bottom fan-in $s$.
  Moreover, the number of gates at level $1$ is bounded by $M \cdot 2^s$,
  which is an invariant preserved during the iterations.

  For every other iteration $i$ ($i = 2, \cdots, d$),
  we generate $\rho^i := \rho_{G(z_i)}$ by choosing a seed $z_i \sim \binset^r$ uniformly at random.
  Using the invariant that the number of gates at level $i-1$
  is at most $M \cdot 2^s$,
  the probability
  that some gate at level $i$ in $C \restriction_{\rho^1 \cdots \rho^i}$ cannot be computed by a decision tree of depth $s$
  is bounded above by 
  $$
  M_i \cdot \left( 2^{s+s+1} (5ps)^s + \eps_0 \cdot 2^{(s+1)(2s + \log(M2^s))} \right).
  $$
  In the complement event, every gate at level $i$ can be written as
  width-$s$ CNFs or DNFs of size $2^s$, and hence can be merged into
  some gate at level $i+1$ (for $i < d$).
  At the last iteration (i.e., $i = d$),
  the circuit $C \restriction_{\rho^1 \cdots \rho^d}$ can be written as 
  a decision tree of depth $s$.
  We define the pseudorandom restriction $\rho$ as $\rho^d \circ \cdots \circ \rho^1$.
  \cref{item_ConstructionOfPseudorandomRestriction} is obvious from this construction.

  Overall, 
  the probability that $\DT(C \restriction_{\rho}) > s$
  is at most
  $M \cdot ( 2^{-s + 1} + \eps_0 \cdot 2^{(s+1)(3s + \log M)} ).$
  This completes the proof of \cref{item_ProbabilityCircuitShrinks}.

  To see \cref{item_ALotOfStars},
  we divide $N$ input bits into $k$ disjoint blocks $T_1, \cdots, T_k$
  of size at least $t$ (and hence $k = \lfloor N / t \rfloor$),
  where $t$ is a parameter chosen later.
  We claim that each block must contain at least one unrestricted variable
  in $\rho \sim \mathcal R$ with high probability (and hence $| \rho^{-1}(*) | \ge \lfloor N / t \rfloor$).
  Fix any block $T = T_i$ for some $i \in [k]$.
  As in \cite{Tal17_coco_conf}, one can easily observe that the condition that every variable in $T$ is restricted can be checked by a CNF of size at most $|T| \ (\le N)$.
  By a simple hybrid argument, the concatenation of $d$ independent
  pseudorandom distributions $G(z_1), \cdots, G(z_d)$ $d \eps_0$-fools
  CNFs (cf.\ \cite[Corollary 55]{Tal17_coco_conf}).
  Therefore,
  the probability that every variable in $T$ is restricted by $\rho \sim \mathcal R$
  is bounded by $(1 - p^{d-1}/64)^{t} + d\eps_0$,
  where the first term is an upper bound for the probability that 
  every variable in $T$ is restricted by a truly random restriction.
  Choosing $t = 64 \log (1 / \delta) / p^{d-1}$ and using a union bound,
  the probability that some block $T_i$ is completely fixed
  can be bounded above by $\lfloor N / t \rfloor \cdot (\delta + d \eps_0)$,
  which completes the proof of \cref{item_ALotOfStars}.
\end{proof}

\begin{corollary}
  \label{cor_PseudorandomRestriction}
  For every circuit $C$ of size $M \  (\ge N)$ and depth $d$ over $N$ inputs,
  there exists a restriction $\rho$ such that 
  \begin{enumerate}
    \item
      $C\restriction_\rho$ is a decision tree of depth at most
      $s := 2 \log 8M$,
    \item
      $|\rho^{-1}(*)| \ge N / O(\log M)^d$, and
    \item
      $\Kt(\rho) \le d \cdot \widetilde O( (\log M)^3 )$.
  \end{enumerate}
\end{corollary}
\begin{proof}
  Tal \cite[Theorem 52]{Tal17_coco_conf}
  showed that 
  there exists a polynomial-time pseudorandom generator $G$
  of seed length $r := \widetilde O(\log M_0 \cdot \log (M_0 / \eps_0))$
  that $\eps_0$-fools CNFs of size $M_0$.
  We set $M_0 := M \cdot 2^s \cdot 2^{s(q+1)}$, $s := 2 \log 8M$,
  and $\eps_0 := 2^{-9 s^2}$.
  Then 
  the seed length $r$ of $G$
  is 
  at most 
  $r = \widetilde O(\log M_0 \cdot \log (M_0 / \eps_0))
  = \widetilde O(\log M \cdot (\log M)^2 )$.
  Applying \cref{theorem_PseudorandomRestriction},
  the probability that $\DT(C \restriction_\rho) > s$
  is bounded by $\frac{1}{2}$.
  Choosing $\delta = 1 / 8N$,
  we also have that the probability that $|\rho^{-1}(*)|
  < \lfloor N \cdot p^{d-1}/64 \log(1/\delta) \rfloor$
  is at most $\frac{1}{4}$.
  Thus there exists some restriction $\rho$ in the support of $\mathcal R$
  such that
  $\DT(C\restriction_\rho) \le s$
  and 
  $|\rho^{-1}(*)| \ge \Omega(N \cdot p^{d-1} / \log N) \ge N / O(\log M)^d.$
\end{proof}

Using the assumption that a circuit computes $\MKtP$,
we show that 
shallow decision trees must be a constant function.
\begin{lemma}
  \label{lemma_MakeDTConstant}
  Let $C$ be a circuit and $\rho$ be a restriction such that
  $C\restriction_\rho$ is a decision tree of depth $s$.
  If $\MKtP[O(s \log N) + \Kt(\rho)] \subseteq C^{-1}(1)$,
  then $C \restriction_\rho \equiv 1$.
\end{lemma}
\begin{proof}
  We prove the contrapositive.
  Assume that $C \restriction_\rho \not\equiv 1$,
  which means that there is a path $\pi \colon [N] \to \{0, 1, *\}$
  of a decision tree $C \restriction_\rho$
  that assigns at most $s$ variables
  so that $C \restriction_{\rho \pi} \equiv 0$.
  Note that $\Kt(\pi) \le O(s \log N)$, because 
  one can specify each restricted variable of $\pi$ by using $O(\log N)$ bits.
  Thus we have
  $\Kt(0^N \circ \pi \circ \rho) \le O(s \log N) + \Kt(\rho)$.
  On the other hand, 
  $C(0^N \circ \pi \circ \rho) = C\restriction_{\rho\pi}(0^N) = 0$.
  Therefore, we obtain $\MKtP[O(s \log N) + \Kt(\rho)] \not\subseteq C^{-1}(1)$.
\end{proof}

Now we are ready to prove the main result of this section.
\begin{proof}
[Proof of \cref{theorem_MKtPversusAC}]
  Assume, by way of contradiction,  that there is a circuit $C$ 
  of size $M := N^{O(1)}$ and depth $d$
  that computes $\MKtP[d \cdot \widetilde O(\log N)^3, N / \omega( \log N )^d ] $.
  Using \cref{cor_PseudorandomRestriction},
  we take a restriction $\rho$
  such that $C \restriction_\rho$ is a decision tree of depth $s = O(\log N)$.
  By \cref{lemma_MakeDTConstant},
  we have $C \restriction_\rho \equiv 1$,
  under the assumption that $O(s \log N) + \Kt(\rho) \le \theta(N)$,
  which is satisfied by choosing $\theta(N)$ large enough.
  Now, by counting the number of inputs accepted by $C \restriction_\rho$,
  we obtain
  $$
  2^{N / O(\log N)^d} \le 2^{|\rho^{-1}(*)|} = |(C\restriction_\rho)^{-1}(1)| \le 2^{\theta'(N) + 1},
  $$
  where, in the last inequality,
  we used the fact that the number of strings whose $\Kt$ complexity is at most $\theta'(N)$ is at most $2^{\theta'(N) + 1}$.
  However, the inequality contradicts the choice of $\theta'(N)$.
\end{proof}

\subsection{$\mathsf{NQP} \nsubseteq \mathsf{NC}^1$ and $\mathsf{Formula}$-$\mathsf{XOR}$ or $\GapAND$-$\FM$ for $\MCSP$}\label{ss:frontier_MCSP_formula_xor}

This section is devoted to proving \HMFrontierFormula{1} and \HMFrontierRandFormula[1]. In fact, we provide two different proofs of \HMFrontierFormula{1}, one based on~\cite{OS18_mag_first}, another one based on~\cite{Magnification_FOCS19}.

In both proofs, the hardness magnification is achieved by constructing an oracle circuit for $\MCSP$. The most interesting part of the first proof is that it gives a \emph{conditional} construction assuming $\QPtime \subseteq \Ppoly$. While the oracle circuit construction can be made \emph{unconditional} (as in the second proof), it illustrates a potentially more applicable approach: proving the hardness magnification theorem while assuming the target circuit lower bound is false (i.e., $\NQPtime \subseteq \NC$).

\subsubsection{Reduction Based Approach from~\cite{OS18_mag_first}}

In the initial magnification theorem \cite[Theorem 1]{OS18_mag_first}, approximate $\MCSP$ was shown to admit hardness magnification phenomena. Here we present a similar hardness magnification theorem for a worst-case version of \MCSP.

A natural way of reducing worst-case \MCSP to approximate \MCSP is to apply error-correcting codes. Error-correcting codes map a hard Boolean function to a Boolean function which is hard on average. A problem with this approach is that error-correcting codes do not guarantee that an easy Boolean function will be mapped to an easy Boolean function. 
Our main idea is to enforce the latter property with an extra assumption $\QPtime\subseteq\Ppoly$.
Here, \QPtime denotes $\TIME[n^{\log^{O(1)} n}]$. Similarly, \NQPtime will stand for $\NTIME[n^{\log^{O(1)} n}]$.  

We will use the following explicit error-correcting code.
\begin{theorem}[Explicit linear error-correcting codes \citep{Justesen72_tit_journals, SipserS96_tit_journals}]\label{thm:ecc}
There exists a sequence $\{E_N\}_{N \in \mathbb{N}}$ of error-correcting codes $E_N \colon \{0,1\}^N \to \{0,1\}^{M(N)}$ with the following properties:
\begin{itemize}
\item $E_N(x)$ can be computed by a uniform deterministic algorithm running in time $\mathsf{poly}(N)$.
\item $M(N) = b \cdot N$ for a fixed $b \geq 1$.
\item There exists a constant $\delta > 0$ such that any codeword $E_N(x) \in \{0,1\}^{M(N)}$ that is corrupted on at most a $\delta$-fraction of coordinates can be uniquely decoded to $x$ by a uniform deterministic algorithm $D$ running in time $\mathsf{poly}(M(N))$.
\item Each output bit is computed by a parity function: for each input length $N \geq 1$ and for each coordinate $i \in [M(N)]$, there exists a set $S_{N,i} \subseteq [N]$ such that for every $x \in \{0,1\}^N$, $$E_N(x)_i = \bigoplus_{j \in S_{N,i}} \;x_j.$$
\end{itemize}
\end{theorem}

Under the assumption that $\QPtime\subseteq\Ppoly$,
we present an efficient reduction from worst-case \MCSP to approximate \MCSP:
Given the truth table of a function $f$, we simply map it to $E_N(\ttable(f))$.
The following lemma establishes the correctness of this reduction.

\begin{lemma}[Reducing worst-case \MCSP to approximate \MCSP]\label{lem:red}
Assume $\QPtime\subseteq\Ppoly$. Then the error-correcting code $E_N$ from Theorem \ref{thm:ecc} satisfies the following:

  \begin{enumerate}
    \item
      $f_n\in\Circuit[2^{n^{1/3}}]\Rightarrow$ $E_N(\ttable(f_n))\in\Circuit[2^{\sqrt{m}}]$,%
      \footnote{
        Here we identify $E_N(\ttable(f_n))$ with the function whose truth table is $E_N(\ttable(f_n))$.
      }
    \item
      $f_n\not\in\Circuit[2^{n^{2/3}}]\Rightarrow$ $E_N(\ttable(f_n))$ is hard to $(1-\delta)$-approximate by $2^{\sqrt{m}}$-size circuits,
  \end{enumerate}
  where $m = \Theta(n)$.
%
%
\end{lemma}

\begin{proof}
  For the first implication we consider the map
  $$C,i\mapsto E_N(\ttable(C))_i$$
  where $C$ is a circuit with $n$ inputs and size $2^{n^{1/3}}$, $i\in\{0,1\}^{m}$, and $m=\log |E_N|$.
  The map takes an input of length $2^{O(n^{1/3})}$,
  and
  is computable in time $2^{O(n)}$; hence the map is in $\QPtime\subseteq\Ppoly$.
  Thus, there exists a circuit $F$ of size $2^{O(n^{1/3})}$ that,
  taking
  the description of a circuit $C$ of size $2^{n^{1/3}}$ and $i \in \binset^m$ as input,
  outputs the $i$th bit of $E_N(\ttable(C))$.
  Therefore if $f_n$ is computed by a circuit $C$ of size $2^{n^{1/3}}$,
  the function $i \mapsto E_N(\ttable(f_n))_i$ is computable by a circuit $F(C, \text{-})$ of size $2^{O(n^{1/3})}<2^{\sqrt{m}}$. 

  The second implication is obtained in a similar way by considering the map
  $$C,i\mapsto D_N(\ttable(C))_i$$
  where $C$ is a circuit with $m=\log |E_N|$ inputs and size $2^{\sqrt{m}}$, $i\in\{0,1\}^n$ and $D_N$ is an efficient decoder of $E_N$. The new map is computable in time $2^{O(m)}$ and again is in $\QPtime\subseteq\Ppoly$. Therefore if $E_N(\ttable (f_n))$ is $(1-\delta)$-approximated by a circuit $C$ of size $2^{\sqrt{m}}$, $f_n$ is computable by a circuit of size $2^{O(\sqrt{m})}<2^{n^{2/3}}$.
\end{proof}

Since the error-correcting code of Theorem \ref{thm:ecc} can be computed by using one layer of XOR gates, we obtain the following corollary.
\begin{corollary}
  \label{cor_WorstMCSPtoApproxMCSP}
  If $\QPtime\subseteq\Ppoly$,
  then
  $\MCSP[2^{n^{1/3}},2^{n^{2/3}}]$ 
  is many-one-reducible to $\MCSP[(2^{\sqrt{n}},0), (2^{\sqrt{n}},\delta)]$ by using a linear-size circuit of $\mathsf{XOR}$ gates.
\end{corollary}

We are ready to prove the main result of this section:
\begin{theorem}[Magnification for worst-case \MCSP via error-correcting codes]\label{thm:mcspecc}
  \hspace{0em}
  \\
  Assume that $\MCSP[2^{n^{1/3}},2^{n^{2/3}}]\not\in \Formula$-$\mathsf{XOR}[N^{1 + \eps}]$ for some constant $\eps > 0$.
  Then either $\QPtime\not\subseteq\Ppoly$ or $\NP\not\subseteq\NC$.
  In particular, $\NQPtime\not\subseteq\NC$. 
\end{theorem}

\begin{proof}
  We prove the contrapositive.
  Assume that $\QPtime \subseteq \Ppoly$ and $\NP \subseteq \NC$.
  \cite[Lemma 16]{OS18_mag_first} shows that
  $\NP \subseteq \NC$ implies $\MCSP[(2^{\sqrt{n}},0), (2^{\sqrt{n}},\delta)] \in \Formula[N^{1 + \eps}]$ for any constant $\eps > 0$.
  By combining this with Corollary \ref{cor_WorstMCSPtoApproxMCSP},
  we obtain that $\MCSP[2^{n^{1/3}},2^{n^{2/3}}] \in \Formula$-$\mathsf{XOR}[O(N^{1 + \eps})]$.
\end{proof}

\subsubsection{Kernelization Based Approach from~\cite{Magnification_FOCS19}}

Now we give another proof of \HMFrontierFormula{1} by adapting techniques from~\cite{Magnification_FOCS19}. 
In fact, the following proof implies (under a straightforward adjustment of parameters) both \HMFrontierFormula[1] and \HMFrontierRandFormula[1].

\begin{theorem}[Magnification for worst-case \MCSP via kernelization for $\GapAND$-$\Formula$-$\mathsf{XOR}$]\label{thm:mcspker}
	\hspace{0em}
	Assume that $\MCSP[2^{n^{1/3}}]\not\in \GapAND_{O(N)}$-$\Formula$-$\mathsf{XOR}[N^{\eps}]$ for some constant $\eps > 0$. Then $\NQPtime\not\subseteq\NC$.
\end{theorem}
\begin{proofsketch}
	\newcommand{\Nyes}{N_{\sf yes}}
	\newcommand{\Syes}{S_{\sf yes}}
	The following proof is just an adaption of Theorem 3.4~of~\cite{Magnification_FOCS19}.
	
	Let $N = 2^n$ and $s = 2^{n^{1/3}} = 2^{(\log N)^{1/3}}$. Let $S = \MCSP[2^{n^{1/3}}]^{-1}(1)$ (that is, all yes instances of $\MCSP[2^{n^{1/3}}]$ on inputs of length $N$), and $m = |S|$. We have that $m \le s^{O(s)}$. Let $E_N$ be the error correcting code from Theorem~\ref{thm:ecc}. Recall that $E_N $ maps from $\{0,1\}^{N}$ to $\{0,1\}^{b \cdot N}$ for a constant $b$.
	
	Let $T = c_1 \cdot \log m$ for a large enough constant $c_1$. Suppose we pick $T$ random indexes $I = (i_1,i_2,\dotsc,i_T)$ from $[b \cdot N]$ independently and uniformly at random. Given $x \in \{0,1\}^{N}$, let $H_{I}(x) := (E_N(x)_{i_1},E_N(x)_{i_2},\dotsc,E_N(x)_{i_T})$.
	
	By a Chernoff bound and a union bound, we can see that with high probability over random choices of $I$, all inputs from $S$ are mapped into distinct strings in $\{0,1\}^{T}$ by $H_{I}$. We fix such a good collection of indexes $I_{\sf good}$.
	
	Now, consider the following language 
	\[
	L_{\sf check} : [b \cdot N]^{T} \times \{0,1\}^{T} \times [b \cdot N] \times \{0,1\} \to \{0,1\},
	\]
	which takes inputs $I$ (hash function coordinates), $w$ (hash value), $i$ (index), and $z$ (check-bit). $L_{\sf check}(I,w,i,z)$ guesses an input $y \in \{0,1\}^N$, and accepts if $H_{I}(y) = w$, $\MCSP[2^{n^{1/3}}](y) = 1$, and $E_N(y)_i = z$. It is easy to see that $L_{\sf check}$ is in $\NQPtime$.
	
	Given $x \in \{0,1\}^N$, we claim that $\MCSP[2^{n^{1/3}}](x) = 1$ if and only if $L_{\sf check}(I_{\sf good},H_{I_{\sf good}}(x),i,E_N(x)_i) = 1$ for all $i \in [b \cdot N]$. 
	
	\begin{enumerate}
		\item When $\MCSP[2^{n^{1/3}}](x) = 1$, on the particular guess $y = x$, we have $L_{\sf check}(I_{\sf good},H_{I_{\sf good}}(x),i,E_N(x)_i)$ accepts for all $i \in [b \cdot N]$.
		
		\item When $\MCSP[2^{n^{1/3}}](x) = 0$, we set $z = H_{I_{\sf good}}(x)$. By our choice of $I_{\sf good}$, there is at most one $x'$ satisfying both $\MCSP[2^{n^{1/3}}](x') = 1$ and $H_{I_{\sf good}}(x') = z$. If there is no such $x'$, then all $L_{\sf check}(I_{\sf good},H_{I_{\sf good}}(x),i,x_i)$ reject. Otherwise, we have $x \ne x'$. Let $i$ be an index such that $E_N(x)_i \ne E_N(x')_i$. Then $L_{\sf check}(I_{\sf good},H_{I_{\sf good}}(x),i,E_N(x)_i)$ rejects.
	\end{enumerate}

	Moreover, in the second case, since $E_N(x)$ is an error correcting code, $L_{\sf check}(I_{\sf good},H_{I_{\sf good}}(x),i,E_N(x)_i)$ indeed rejects at least for a constant fraction of $i \in [b \cdot N]$.
	
	Now suppose $\NQPtime \subseteq \NC$ for the sake of contradiction. Since $H_{I_{\sf good}}(x)$ can be computed by $T = N^{o(1)}$ many XOR gates ($I_{\sf good}$ is hardwired into the circuit), we can construct $b \cdot N$ $\Formula$-$\mathsf{XOR}[N^{o(1)}]$ circuits $C_1,C_2,\dotsc,C_{b \cdot N}$, such that if  $\MCSP[2^{n^{1/3}}](x) = 1$ then $C_i(x) = 1$ for all $x$, and otherwise $C_i(x) = 0$ for a constant fraction of $i$'s. 
	
	By a simple error reduction via random sampling, we can construct $m = O(N)$ $\Formula$-$\mathsf{XOR}[N^{o(1)}]$ circuits $D_1,D_2,\dotsc,D_{m}$, such that if $\MCSP[2^{n^{1/3}}](x) = 1$ then $D_i(x) = 1$ for all $x$, and otherwise $D_i(x) = 0$ for at least a $0.9$ fraction of inputs. Hence, we have $\MCSP[2^{n^{1/3}}] \in \GapAND_{O(N)}$-$\Formula$-$\mathsf{XOR}[N^{o(1)}]$, a contradiction to the assumption.
\end{proofsketch}

\begin{remark}
	Note that $\GapAND_{O(N)}$-$\Formula$-$\mathsf{XOR}[N^{\eps}]$ circuits are special cases of both $\Formula$-$\mathsf{XOR}[N^{1+\eps}]$ circuits and $\GapAND_{O(N)}$-$\Formula[N^{2+\eps}]$ circuits. Therefore, the above proof implies both \HMFrontierFormula[1] and \HMFrontierRandFormula[1].
\end{remark}

\subsection{$\mathsf{NP} \nsubseteq \mathsf{NC}^1$ and Almost-Formula Lower Bounds for $\MCSP$}\label{ss:almost_formulas_magnification}

Recall that near-quadratic \emph{formula} lower bounds are known for $\MCSP[2^{n^{o(1)}},2^{n^{o(1)}}]$. On the other
hand, 
 a hardness magnification obtained by a super efficient construction of anticheckers established in \cite{OPS19_CCC} states that $\NP\subseteq\Ppoly$ implies almost linear-size circuits for a worst-case version of parameterized $\MCSP[2^{n^{o(1)}},2^{n^{o(1)}}]$. Consequently, if we could make the hardness magnification work for formulas, 
$\NP\not\subseteq\NC$ would follow. We make a step in this direction by showing that $\NP\subseteq\NC$ implies the existence of almost-formulas of almost linear size solving the worst-case $\MCSP[2^{n^{o(1)}},2^{n^{o(1)}}]$, cf.~Theorem \ref{pnp}. This is established by a more detailed analysis of the proof from \cite{OPS19_CCC} extended with an application of the Valiant-Vazirani Isolation Lemma (cf.~\cite[Lemma 17.19]{AB09}) in the process of selecting anticheckers. We also observe that almost-formulas of subquadratic size cannot solve \PARITY, cf.~Theorem \ref{thm:parityflalikelb}. These results yield \HMFrontieralmostFormula{1} and \HMFrontieralmostFormula{3}.

We start the presentation with a lemma needed to derive \HMFrontieralmostFormula{1}.

\begin{lemma}[Anticheckers]\label{lem} Assume $\NP\subseteq\NC$. Then for any $\lambda\in (0,1)$ there are 
circuits $\{C_{2^n}\}_{n=1}^{\infty}$ of size $2^{n+O(n^\lambda)}$ which given $\ttable(f) \in \{0,1\}^N$, outputs $2^{O(n^\lambda)}$ $n$-bit strings $y_1,\dots,y_{2^{O(n^\lambda)}}$ together with bits $f(y_1),\dots,f(y_{2^{O(n^\lambda)}})$ forming a set of anticheckers for $f$, i.e. if $f$ is hard for circuits of size $2^{n^\lambda}$ then every circuit of size $2^{n^\lambda}/2n$ fails to compute $f$ on one of the inputs $y_1,\dots,y_{2^{O(n^\lambda)}}$. Moreover, each $y_i,f(y_i)$ is generated by a subcircuit of $C_{2^n}$ with inputs $y_1,\dots,y_{i-1},f(y_1),\dots,f(y_{i-1}),\ttable(f)$ whose only gates with fanout $>1$ are $y_1,\dots,y_{i-1},f(y_1),\dots,f(y_{i-1})$.
\end{lemma}

\proof This proof follows \cite{OPS19_CCC}. Our contribution here is the ``moreover" part, but we also give a more succinct self-contained proof. For each Boolean function $f$ the desired set of anticheckers is known to exist, the only problem is to find it with a circuit of the desired size and formula-like form. In order to do so, we will simulate the proof of the existence of anticheckers, but make the involved counting constructive by using linear hash functions and the assumption $\NP\subseteq\NC$. Additionally, for the ``moreover'' part of the lemma, we will employ the Valiant-Vazirani Isolation Lemma (cf.~\cite[Lemma 17.19]{AB09}) in the process of selecting good anticheckers. 
\medskip

Let $\lambda\in (0,1)$ and $f$ be a Boolean function with $n$ inputs hard for circuits of size $2^{n^\lambda}$. For $j$ $n$-bit strings $y_1,\dots,y_j$ and $s\in [0,1]$, define a predicate $$P_f(y_1,\dots,y_j)[s]\text{ iff } \leq s\text{ fraction of all circuits of size }2^{n^\lambda}/2n\text{ compute }f\text{ on }y_1,\dots,y_j.$$ Further, let $R_f(y_1,\dots,y_j)$ be the number of circuits of size $2^{n^\lambda}/2n$ which do not make any error on $y_1,\dots,y_j$ when computing $f$. Note that $P_f$ and $R_f$ depend on $j$ values of $f$, not on the whole $\ttable(f)$, but for simplicity we do not display them.
\smallskip

\sloppy Suppose that given $\ttable(f)$ we already generated $y_1,\dots, y_{i-1}, f(y_1),\dots,f(y_{i-1})$ such that $P_f(y_1,\dots,y_{i-1})[(1-1/4n)^{i-1}]$ holds. For $i=1$ the generated set is empty. We want to find $y_i, f(y_i)$ such that $P_f(y_1,\dots,y_i)[(1-1/4n)^i]$. In order to do so, we will construct a formula $F(y_1,\dots,y_i,f(y_1),\dots,f(y_i))$ of size $2^{O(n^\lambda)}$ (if $i\le 2^{O(n^\lambda)}$) such that under the assumption $R_f(y_1,\dots,y_{i-1})\ge 2n^2$, $$F(y_1,\dots,y_i,f(y_1),\dots,f(y_i))=1\quad\Rightarrow\quad P_f(y_1,\dots,y_i)[(1-1/4n)^i]$$ $$P_f(y_1,\dots,y_{i-1})[(1-1/4n)^{i-1}]\quad \Rightarrow\quad \exists y_i, F(y_1,\dots,y_i,f(y_1),\dots,f(y_i))=1.$$ Assume for now that we already have such a formula $F$. We firstly show how to find $y_i,f(y_i)$ given $F$ by an exhaustive search through all $n$-bit strings in combination with Valiant-Vazirani Lemma. 

Consider a $2^{O(n^\lambda)}$-size formula $F^{r,h}(y_1,\dots,y_{i-1},z,f(y_1),\dots,f(y_{i-1}),f(z))$ computing the following predicate \begin{equation}\label{Frh}F(y_1,\dots,y_{i-1},z,f(y_1),\dots,f(y_{i-1}),f(z))\wedge h(z)=0^r\end{equation} where $z\in\{0,1\}^n$, $r\le n+2$ and $h\in \mathcal{H}_{n,r}$ for a pairwise independent efficiently computable hash function collection $\mathcal{H}_{n,r}$ from $\{0,1\}^{n}$ to $\{0,1\}^r$. Formula $F^{r,h}$ exists since $\NP\subseteq\NC$. By Valiant-Vazirani Lemma, for fixed $y_1,\dots, y_{i-1},f(y_1),\dots,f(y_{i-1})$, if $h$ is chosen randomly from $\mathcal{H}_{n,r}$ and $r$ randomly from $\{2,\dots,n+1\}$, then with probability $\ge 1/8n$, there is a unique $z$ satisfying (\ref{Frh}). Therefore, the probability that none of $2^{O(n^\lambda)}$ many randomly chosen tuples $r,h$ guarantees a unique solution is $<(1-1/8n)^{2^{O(n^\lambda)}}\le 1/2^{2^{O(n^\lambda)}/8n}$. That is, there exist a set $\mathcal{R}$ of $2^{O(n^\lambda)}$ tuples $r,h$ such that for each $y_1,\dots,y_{i-1},f(y_1),\dots,f(y_{i-1})$, at least one tuple $r,h$ from $\mathcal{R}$ will guarantee a unique solution. 
Consequently, for each $y_1,\dots,y_{i-1},f(y_1),\dots,f(y_{i-1})$ for at least one $r,h\in \mathcal{R}$ the following $2^{n+O(n^\lambda)}$-size formula 
$$\bigvee_{k=1,\dots,2^n} (b^k_j\wedge F^{r,h}(y_1,\dots,y_{i-1},b^k,f(y_1),\dots,f(y_{i-1}),f(b^k)),$$ where $b^k_j$ is the $j$th bit of the $k$th $n$-bit string $b^k$ (in the lexicographic order), outputs the $j$th bit of a good antichecker $y_i$. 
Since $\NP\subseteq\NC$, we can select the right $y_i$ from the $2^{O(n^\lambda)}$ candidate strings corresponding to tuples $r,h$ from $\mathcal{R}$ by applying a formula of size $2^{O(n^\lambda)}$ on top of them. Having $y_i$, a formula of size $poly(n)2^{n}$ with access to $\ttable(f)$ can generate $f(y_i)$.\def\lookslinethicanbesimplified{Formula $D^{r,h}$ is constructed as follows. Since $\NP\subseteq\NC$, for bits $u,v$, we have $2^{O(n^\lambda)}$-size formulas $D_{u,v}(y_1,\dots,y_{i-1},f(y_1),\dots,f(y_{i-1}))$ computing the predicate $$\exists c\ne c', F^{r,h}(y_1,\dots,y_{i-1},c,f(y_1),\dots,f(y_{i-1}),u)\wedge F^{r,h}(y_1,\dots,y_{i-1},c',f(y_1),\dots,f(y_{i-1}),v).$$ $\NP\subseteq\NC$ also implies the existence of $2^{O(n^\lambda)}$-size formulas $S_{u,v}(y_1,\dots,y_{i-1},f(y_1),\dots,f(y_{i-1}))$ which output such $c\ne c'$ whenever they exist. For each $u,v$, we can thus check if $D_{u,v}$ is consistent with $\ttable(f)$ by generating $c,c'$ and then applying $poly(n)2^n$-size formulas $$\bigwedge_{k=1,\dots,2^n} [(c=b^k\rightarrow u=f(b^k))\wedge (c'=b^k\rightarrow v=f(b^k))].$$ This yields $2^{n+O(n^\lambda)}$-size formulas $D^{r,h}$ and completes the process of generating the right string $y_i$.}

Iteratively, a circuit of size $2^{n+O(n^\lambda)}$ will generate $y_1,\dots,y_{2^{O(n^\lambda)}},f(y_1),\dots,f(y_{2^{O(n^\lambda)}})$ such that $P_f(y_1,\dots,y_{2^{O(n^\lambda)}})[(1-1/4n)^{2^{O(n^\lambda)}}]$ as long as $R_f(y_1,\dots,y_{2^{O(n^\lambda)}})\ge 2n^2$. Deciding whether $R_f(y_1,\dots,y_i)\ge 2n^2$ is in $\NPtime\subseteq\NC$ (on input $y_1,\dots,y_i,f(y_1),\dots,f(y_i),1^{2^{n^\lambda}}$), so there are formulas of size $2^{O(n^\lambda)}$ for it. Since $(1-1/4n)^{2^{O(n^\lambda)}}\le 1/2^{2^{O(n^\lambda)}/4n}$, we reach $R_f(y_1,\dots,y_i)<2n^2$ with $i\le 2^{O(n^\lambda)}$. 
When this happens, the remaining $<2n^2$ circuits of size $2^{n^\lambda}/2n$ can be generated by an $\NPtime^{\coNP}$ algorithm, and since $\NPtime\subseteq\NC$, by a formula of size $2^{O(n^{\lambda})}$. Finally, for each of the remaining circuits we can find an $n$ bit string witnessing its error exhaustively by a formula of size $2^{n+O(n^\lambda)}$. Altogether, the desired anticheckers $y_1,\dots,y_{2^{O(n^\lambda)}}$ with bits $f(y_1),\dots,f(y_{2^{O(n^\lambda)}})$ will be generated by a circuit of size $2^{n+O(n^\lambda)}$. Note that this circuit will have the desired formula-like structure because its only gates with fanout bigger than 1 are those computing tuples $y_i,f(y_i)$. 

\begin{claim}\label{cl}
If $P_f(y_1,\dots, y_{i-1})[(1-1/4n)^{i-1}]$ and $R_f(y_1,\dots,y_{i-1})\ge 2n^2$, then for some $y_i$, $P_f(y_1,\dots,y_i)[(1-1/4n)^{i-1}(1-1/2n)]$.
\end{claim}

Claim \ref{cl} is proved by a standard counting argument, cf.~\cite[Claim 22]{OPS19_CCC}. Observe that with Claim \ref{cl} we can construct the desired formula $F$. Here we employ approximate counting with linear hash functions: if $X\subseteq \{0,1\}^m$ is a set of size $s$, there are matrices $A_1,\dots,A_{\log (4s^c)}$ such that each $A_j$ defines a linear function mapping a Cartesian power $X^c$ to $(s(1+\eps))^c/\log (4s^c)$, for $c=2(\eps^{-1}(\log\log s+\log\eps^{-1}))$. Moreover, for each $A_j$ there is $X^c_j\subseteq X^c$ satisfying $\forall x\in X^c_j\forall x'\in X^c\ (x\ne x' \rightarrow A_j(x)\ne A_j(x'))$, and $\bigcup_{j} X^c_j=X^c$. Mapping $x\in X^c$ to $A_j(x)$ in the $j$th block of size $(s(1+\eps))^c/\log (4s^c)$, for the first $A_j$ with $x\in X^c_j$, thus defines an injection from $X^c$ to $(s(1+\eps))^c$ which witnesses that the size of $X$ is $\le s(1+\eps)$. See e.g.~\cite[Section 3, 2nd paragraph]{DBLP:journals/jsyml/Jerabek09} for details.

Therefore, once we have $P_f(y_1,\dots,y_i)[(1-1/4n)^{i-1}(1-1/2n)]$ we can conclude that there are matrices $A_1,\dots,A_{2^{O(n^\lambda)}}$ defining an injective mapping of a Cartesian power (with exponent of rate $poly(n)$) of the set of all circuits of size $2^{n^\lambda}/2n$ that compute $f$ on $y_1,\dots,y_i$ to the same Cartesian power of $(1-1/4n)^{i-1}(1-1/2n)(1+1/4n)\le (1-1/4n)^i$ fraction of the set of all circuits of size $2^{n^\lambda}/2n$. The existence of such matrices, not only witnesses $P_f(y_1,\dots,y_i)[(1-1/4n)^i]$ but is also an $\NPtime^\coNP$ property
, and since $\NPtime\subseteq\NC$, decidable by a formula $F$ of size $2^{O(n^\lambda)}$.\qed

\bigskip

\begin{theorem}[Improved magnification via anticheckers]\label{pnp} Assume that $\MCSP[2^{n^{1/2}}/2n,2^{n^{1/2}}]$ is hard for circuits $C$ (with $2^n$ inputs) of size $2^{n+O(n^{1/2})}$ with the following form. Given $\ttable(f)$, subcircuits of $C$ generate $y_1,\dots,y_{2^{O(n^{1/2})}}, f(y_1),\dots,f(y_{2^{O(n^{1/2})}})$ so that each $y_i,f(y_i)$ is generated by a subcircuit of $C$ with inputs $y_1,\dots,y_{i-1},f(y_1),\dots,f(y_{i-1}),\ttable(f)$ whose only gates with fanout $>1$ are $y_1,\dots,y_{i-1},f(y_1),\dots,f(y_{i-1})$. Having $y_1,\dots,y_{2^{O(n^{1/2})}}, f(y_1),\dots,f(y_{2^{O(n^{1/2})}})$, $C$ applies a formula of size $2^{O(n^{1/2})}$ on top of these gates. 

Then $\NPtime\not\subseteq\NC$.
\end{theorem}

\proof If $\NPtime\subseteq\NC$, then $\MCSP[2^{n^{1/2}}/2n,2^{n^{1/2}}]$ can be solved by circuits of size $2^{n+O(n^{1/2})}$ of the required form: given a Boolean function $f$, apply Lemma \ref{lem} to generate a set of its anticheckers $y_1,\dots,y_{2^{O(n^{1/2})}}$ together with bits $f(y_1),\dots,f(y_{2^{O(n^{1/2})}})$ and using $\NPtime\subseteq\NC$ decide whether $f$ is hard for circuits of size $2^{n^{1/2}}/2n$ on $y_1,\dots,y_{2^{O(n^{1/2})}}$. \qed

\bigskip

Note that circuits from the assumption of hardness magnification via anticheckers, Theorem \ref{pnp}, are $2^{O(n^{1/2})}$-almost formulas of almost linear size which gives us \HMFrontieralmostFormula{1}. We can now complement it with \HMFrontieralmostFormula{3}.

Consider an $s$-almost formula. Each gate $G$ of $F$ with fanout larger than 1 is computed by a formula with inputs being either the original inputs of $F$ or gates of $F$ with fanout larger than 1. We call any maximal formula of this form a \emph{principal} formula of $G$.

\begin{theorem}\label{thm:parityflalikelb}
\PARITY$\not\in n^{\eps}$-$\mathsf{almost}$-$\Formula[n^{2-9\eps}]$, if $\eps<1$.
\end{theorem}

\proof[Proof Sketch.] For the sake of contradiction, assume \PARITY has $n^\eps$-almost formulas of size $n^{2-9\eps}$. Since there are only $n^{\eps}$ gates of fanout $>1$, we can replace these gates by appropriate constants and obtain formulas $F_n$ of size $n^{2-8\eps}$ computing \PARITY with probability $\ge 1/2+1/2^{n^{\eps}}$. In more detail, each formula $F_n$ checks if the principal formulas compute the fixed constants. If this is the case, then $F_n$ outputs the output of the original almost-formula (since gates with fan-out larger than 1 are fixed, the output can be computed by a formula). Otherwise, $F_n$ outputs a fixed constant, whichever is better on the majority of the remaining inputs. This does not increase the size of the resulting formula $F_n$ by more than a constant factor. As pointed out by Komargodski-Raz~\cite{KR_STOC_paper}, each boolean function $f$ on $n$ input bits can be approximated by a real polynomial of degree $O(t\sqrt{L(f)}\frac{\log n}{\log\log n})$ up to a point-wise additive error of $2^{-t}$, and this can be shown to imply that each formula of size $o((n/t)^2(\log\log n/\log n)^2)$ computes \PARITY over $n$ input bits with probability at most $1/2+1/2^{t+O(1)}$ (for large enough $t$). 
Taking $t=n^{2\eps}$ we get a contradiction. 
\qed

\subsection{$\mathsf{NP} \nsubseteq \mathsf{NC}^1$ and $\mathsf{AC}^0$ Lower Bounds for $(n-k)$-$\mathsf{Clique}$}\label{ss:frontier_clique}

In this Section, we discuss the proofs of some statements claimed in \HMFrontierClique from Section \ref{ss:intro_background_motivation}. Recall that we consider graphs on $n$ vertices that are described in the adjacency matrix representation. The input graph is therefore represented using $m = \Theta(n^2)$ bits. We begin with the proof of the magnification result in \HMFrontierClique{1}.

\begin{proposition}\label{p:hm_clique}
Let $k(n) = (\log n)^C$ for some constant $C$. If there exists $\varepsilon >0$ such that for every depth $d \geq 1$, $(n-k)$\emph{-}$\mathsf{Clique} \notin \mathsf{AC}^0_d[m^{1 + \varepsilon}]$, then $\mathsf{NP} \nsubseteq \mathsf{NC}^1$. 
\end{proposition}

\begin{proof}
We use a straightforward reduction to the magnification theorem for $k$-$\mathsf{Vertex}$-$\mathsf{Cover}$ established in \citep[Theorem 7]{OS18_mag_first}. (We state Proposition \ref{p:hm_clique} in a slightly weaker form just for simplicity.) Indeed, a graph $G$ on $n$ vertices has a vertex cover of size $\leq k$ if and only if $G$ has an independent set of size $\geq n - k$. In turn, the latter is true if and only if the complement graph $\overline{G}$ has a clique of size $\geq n-k$. Therefore, by negating input literals, the complexities of $(n-k)$\emph{-}$\mathsf{Clique}$ and $k$-$\mathsf{Vertex}$-$\mathsf{Cover}$ are equivalent with respect to $\mathsf{AC}^0$ circuits. For this reason, the hardness magnification theorem of \citep{OS18_mag_first} immediately implies Proposition \ref{p:hm_clique}.
\end{proof}

We state below conditional and unconditional lower bounds on the complexity of detecting very large cliques. The next proposition implies the lower bound claimed in \HMFrontierClique{4}.

\begin{proposition}[\citep{DBLP:journals/dm/AndreevJ08}; see also {\citep[Section 9.2]{DBLP:books/daglib/0028687}}]\label{p:mon_clique_lb}
For $k(n) \leq n/2$, every monotone circuit for $(n-k)$-$\mathsf{Clique}$ requires $2^{\Omega(k^{1/3})}$ gates.
\end{proposition}

\noindent Interestingly, the problem can be solved by (bounded depth) polynomial size monotone circuits if $k \leq \sqrt{\log n}$ \citep{DBLP:journals/dm/AndreevJ08}.

Finally, by the observation employed in the proof of Proposition \ref{p:hm_clique}, for non-monotone computations the complexities of detecting large cliques and small vertex covers are equivalent. A consequence of this is that one can show the following result, which implies the statement in \HMFrontierClique{2}.

\begin{proposition}
If \emph{ETH} for non-uniform circuits holds, then $(n-k)$\emph{-}$\mathsf{Clique} \notin \mathsf{P}/\mathsf{poly}$ as long as $\omega(\log n) \leq k \leq n/2$.
\end{proposition}

Indeed, under ETH the $k$-Vertex-Cover problem cannot be solved in time $2^{o(k)} \cdot \mathsf{poly}(m)$  (see \citep{DBLP:journals/jcss/ImpagliazzoPZ01} and \citep[Theorem 29.5.9]{DBLP:series/txcs/DowneyF13}). Further discussion on the conditional hardness of $k$-$\mathsf{Vertex}$-$\mathsf{Cover}$ that also applies to $(n-k)$-$\mathsf{Clique}$ appears in \citep{OS18_mag_first}.

\section{Hardness Magnification and Natural Proofs}\label{s:non-natural}

\subsection{Equivalences}\label{ss:proof_equivalences}

The main contribution of this section is new hardness magnification results showing non-learnability of circuit classes from slightly super-linear lower bounds for the approximate version of \MCSP and the gap version of \MCSP. We then use these magnification results to establish a series of equivalences.

\begin{lemma}[Hardness Magnification for Learnability from Lower Bounds for Approximate \MCSP]
  \label{lem:learn_lb}    
  Let $s, t : \N \rightarrow \N$ be size functions such that $n\le s(n)\le t(n)$ and $\eps, \delta$ be parameters such that $\eps<1/2$, $0 \leq \delta \leq 1/9$. 
If for infinitely many input lengths $N=2^n$, $\MCSP[(s,0),(t,\eps)] \notin\Circuit[N \cdot \poly(t(n)/\eps)]$, then for infinitely many inputs $n$, $\Circuit[s(n)]$ cannot be learnt up to error $\eps/2$ with confidence $1-\delta$ by $t(n)$-size circuits using non-adaptive membership queries over the uniform distribution.
\end{lemma}

We also show a related result which gives lower bounds for learnability of a circuit class $\cC$ using $\cC$-circuits by starting with a lower bound against worst-case $\MCSP$ instead of the average-case.
\begin{lemma}[Hardness Magnification for Learnability from Lower Bounds for Gap MCSP]
  \label{lem:wc_mcsp_learn_lb}
Let $c\ge 1$ be an arbitrary constant. If there is $\eps<1/2$, such that infinitely many input lengths $N = 2^n$, $\MCSP[n^c,2^n/n^c] \notin  \Circuit[N^{1+\varepsilon}]$, then for every $\gamma\in (0,1)$, for infinitely many inputs $n$, $\Circuit[n^c]$ cannot be learnt up to error $1/O(n^{2c})$ with confidence $1-1/n$ by $\Circuit[2^{O(n^\gamma)}]$-circuits using non-adaptive membership queries over the uniform distribution.
\end{lemma}

\begin{proof}[Proof of Theorem \ref{thm:equivalences}]
  The following implications establish the desired equivalences.\\

  \vspace{-0.3cm}

  \noindent $(a) \Longrightarrow (c)$: For the parameters $c, \gamma, \varepsilon$ given by $(a)$, we apply Lemma \ref{lem:learn_lb}
  for 
$s(n)=n^c$ and $t(n)=2^{n^\gamma}$, to see that for some $\gamma'>0$, $\Circuit[n^c]$ cannot be learned by circuits of size $2^{O(n^{\gamma'})}$ via non-adaptive queries up to an error $O(1/n^c)$. 

  \vspace{0.1in}
  
  \noindent $(c) \Longrightarrow (d)$: We show the contrapositive of this implication. Suppose that for every $d \geq 1$, there exists a $\Circuit[\poly(n)]$-natural property that is useful against $\Circuit[n^d]$ for all large enough $n$. By Theorem \ref{thm:cikk_learn}, for every $c \geq 1$, we can learn $\Circuit[n^c]$ by a sequence of oracle $\Circuit[2^{O(n^{1/2})}]$-circuits up to an error of $n^{-c}$, by choosing $d = 2ac$ for the constant $a$ from 
  Theorem  \ref{thm:cikk_learn}.

  \vspace{0.1in}
 \noindent $(d) \Longrightarrow (a)$, $(d) \Longrightarrow (b)$: Trivial, using the fact that random functions are hard.

  \vspace{0.1in}
  \noindent $(c) \Longrightarrow (e)$: Follows from the contrapositive of Theorem \ref{thm:no_prf_implies_learning}.
  
  \vspace{0.1in}
  \noindent $(e) \Longrightarrow (c)$: Follows from the non-uniform version of Proposition 29 in \cite{DBLP:conf/coco/OliveiraS17}, using essentially the same proof.

  \vspace{0.1in}
  \noindent $(b) \Longrightarrow (c)$: For the parameter $c$ given by $(b)$, we apply Lemma \ref{lem:wc_mcsp_learn_lb} 
to see that $\Circuit[n^c]$ cannot be learned by circuits of size $2^{O(n^\gamma)}$ via non-adaptive queries up to an error $O(1/n^c)$, for any $\gamma \in (0,1)$.

\end{proof}

We now complete the proof of Theorem \ref{thm:equivalences} by proving Lemmas \ref{lem:learn_lb} and \ref{lem:wc_mcsp_learn_lb}. 
\begin{proof}[Proof of Lemma \ref{lem:learn_lb}]
  For the promise problem $\MCSP[(s,0),(t,\eps)]$ over $N$ inputs, define
  \begin{equation*}
    \begin{split}
      \Pi_{yes} &= \{y \in \{0,1\}^N \mid \exists \text{ circuit of size} \leq s(n) \text{ that computes } f_y \} \\
      \Pi_{no} &= \{y \in \{0,1\}^N \mid \text{no circuit of size} \leq t(n) \text{ $(1-\eps)$-approximates } f_y \} \\
    \end{split}
  \end{equation*}

  We prove the contrapositive of the statement, by showing a reduction from $\MCSP[(s,0),(t,\eps)]$ to a learning algorithm for $\Circuit[s(n)]$
  using non-adaptive membership queries over the uniform distribution. For a fixed $\eps<1/2$ and $0 \leq \delta \leq 1/9$, let
  $\{D_n\}_{n \geq 1} \in \Circuit[t(n)]$ be the corresponding sequence of oracle circuits which learns $\Circuit[s(n)]$ up to error $\eps/2$, where $D_n$ makes non-adaptive queries to some function $f \in \Circuit[s(n)]$ over $n$ inputs. 

  Let $q = q(n) = \frac{200}{\eps^2}$. Define $F_N : \{0,1\}^N \times \{0,1\}^{n q(n)} \times \{0,1\}^{t(n)} \rightarrow \{0,1\}$ as the sequence of randomized
  circuits such that :
  \begin{equation*}
    \begin{split}
      z \in \Pi_{yes} &\implies \Pr_{y_1, w} \{F_N(z, y_1, w) = 1\} > 2/3 \\
      z \in \Pi_{no}  &\implies \Pr_{y_1, w} \{F_N(z, y_1, w) = 1\} < 1/3 \\
    \end{split}
  \end{equation*}


  \sloppy
  The reduction $F_N$ does the following. Let $Y = (x_1, \dots, x_{t(n)})$ be the set of queries made by $D_n$. $F_N$ runs the learner $D_n$ with input $w$ as its source
  of internal randomness and answers its oracle queries to $f_z$ by using the other input $z \in \{0,1\}^N$. If the output string of the learner cannot be interpreted
  as a $t(n)$-sized circuit, then $F_N$ outputs $0$. Otherwise, let $h$ be the $t(n)$-sized circuit on $n$ inputs, which can interpret the hypothesis output by the
  learner as a $t(n)$-sized circuit. $F_N$ then interprets the random input $y_1$ as a sequence of $q$ random examples $v_1, \dots, v_q \in \{0,1\}^n$ and computes $h$
  on each of these. It then forms a string $u \in \{0,1\}^{q}$, where for every $i \in [q], u_i = 1$ if and only if $h(v_i) = f_z(v_i)$. Finally, it uses a threshold gate on $T$ on ${q(n)}$ inputs to check if the Hamming weight of $u$ is at least $((1 - 3\eps/4)q)$.

  We now show the correctness of the reduction. If $z \in \Pi_{yes}$, then $f_z$ is computed by some circuit of size at most $s(n)$. Thus, for every random
  choice of $y_1$ and $w$, $D_n$ can learn the function $f_z$ and with probability at least $(1-\delta)$, output a hypothesis $h$ which has an error of at most
  $\eps/2$ with respect to $f_z$. Now, for the $q$ samples given by $y_1$, by an application of Hoeffding's inequality (Lemma \ref{lem:hoeffding}), the probability that the Hamming weight of $u \in \{0,1\}^{q}$ is lesser than $ \left( 1 - 0.6\eps \right)q$ is at most $2\exp(-2q \eps^2/100)$ which is at  most $1/4$ for our choice of $q$. When $\delta \leq 1/9$, we see that $T(u) = 1$ with probability at least 
  $(1-\delta)3/4 \geq 2/3$.

  On the other hand, if $z \in \Pi_{no}$, then no circuit of size at most $t(n)$ can even $(1-\eps)$-approximate $f_z$. Thus, for any random choice of $y_1$
  and $w$, any hypothesis $h$ which $D_n$ outputs is a circuit of size at most $t(n)$ and thus is at least $\eps$-far from $f_z$. By a similar application
  of Hoeffding's inequality, we see that the probability that the Hamming weight of $u \in \{0,1\}^{q}$ is greater than $\left( 1 - 0.9\eps \right)q$ is at most
  $2\exp(-2q \eps^2/100) \leq 1/4$. Therefore, $T(u) = 0$ with probability $2/3$. 

  For the next step, we need to derandomize the circuits $F_N$. Define $E_N$ as
  \begin{equation*}
    \begin{split}
      E_N : \{0,1\}^N \times &\left( \{0,1\}^{n \cdot q + t(n)} \right)^{R} \rightarrow \{0,1\} \\
      E_N(z, y^{(1)}, \dots, y^{(R)}) = & \mathsf{MAJ}_{R} (F_N(z, y^{(1)}), \dots, F_N(z, y^{(R)})) \\
    \end{split}
  \end{equation*}
  where $R = CN$ and each $y^{(j)} \in \{0,1\}^{n \cdot q + t(n)}$, for each $j \in [R]$. 

  When, $z \in \Pi_{yes}$, then using Hoeffding's inequality, we see that with probability at most $2^{-2N}$ (for suitably chosen $C$), the string
  $(F_N(z, y^{(1)}), \dots, F_N(z, y^{(R)}))$ has Hamming weight $\le 3R/5$. Similarly, when $z \in \Pi_{no}$, with probability at most $2^{-2N}$, the string
  $(F_N(z, y^{(1)}), \dots, F_N(z, y^{(R)}))$ has Hamming weight $\ge 2R/5$. Thus, the majority gate differentiates between the two cases except with probability at
  most $2^{-2N}$. We use Adleman's trick \cite{AB09} to fix a string $\alpha \in \{0,1\}^{R \cdot (n \cdot q + t(n))}$ which correctly derandomizes $F_N$ on all
  inputs in $\Pi_{yes}$ and $\Pi_{no}$ and call the resulting circuit as $E^*_N$ which computes the function $E^*_N : \{0,1\}^N \rightarrow \{0,1\}$.

  \sloppy We next compute the size of $E^*_N$. Each $F_N(z, y^{(i)})$ is fixed to $F_N(z, \alpha^{(i)})$, where $\alpha^{(i)} \in \{0,1\}^{(n \cdot q + t(n))}$ is the
  $i^{\text{th}}$ section of the hardwired random string $\alpha$. Observe that for the set of oracle queries $Y$ made by $D_n$, it is enough to use appropriate literals
  from the input $z$ whenever we need to access the truth table of $f_z$. Indeed, whenever $D_n$ uses a random example, the randomness comes from $\alpha^{(i)}$ which
  is fixed non-uniformly and whenever it makes a membership query, the set of queries $Y$ is fixed for $D_n$ because of its non-adaptivity. Recall that the size of the
  circuit $D_n$ is $t(n)$  
  and the hypothesis $h$ output by the learner can be interpreted as a circuit and efficiently computed by another
  circuit of size $\poly(t(n))$. Thus, the circuit size to compute $F_N(z, \alpha)$ is at most $\poly(t(n)\cdot q)$ 
and the total circuit size to construct $E^*_N$ is $O(N \cdot \poly(t(n)/\eps))$. 
\end{proof}

\begin{proof}[Proof sketch of Lemma \ref{lem:wc_mcsp_learn_lb}]
  We show a two-sided error randomized reduction from $\MCSP[n^c,2^n/n^c]$ to $\{D_n\}_{n \geq 1}$. Let $q = q(n) = O \left( n^{3c} \right)$. The reduction is 
  almost the same as that of Lemma \ref{lem:learn_lb}. Here we use a threshold gate on $q(n)$ inputs which answers $1$ whenever the Hamming weight of its input is greater than $(1 - 1/n^{1.5c})q(n)$.

  When the input to $\MCSP[n^c, 2^n/n^c]$ is a yes instance, with probability at least $(1-1/n)$, $D_n$ outputs a hypothesis $h_n \in \Circuit[2^{n^\gamma}]$ which has
  error at most $1/O(n^{2c})$. Now for the $q(n)$ samples drawn uniformly at random, the probability that $h$ agrees with the input instance on at least
  a $(1-1/n^{1.5c})q(n)$ samples is at least $(1-1/n)2/3$.

  When the input to $\MCSP[n^c, 2^n/n^c]$ is a no instance, any hypothesis $h$ which $D_n$ outputs must have error greater than $1/O(n^{c+2})$. Indeed, if the error
  is less than $O(1/n^{c+2})$, then by hardwiring all the error inputs by using circuits of size at most $O \left( \frac{2^n}{n^{c+2}} \cdot n \right)$ we get a circuit
  of size at most $2^n/n^{c}$, which is a contradiction to the promise of the no instance. By Hoeffding's inequality, the probability that $h$ agrees with the input instance on at most a $(1-1/n^{1.5c})q(n)$ samples is at least $2/3$. 

  The derandomization is the same as that of Lemma \ref{lem:learn_lb}, obtained by repeating the above reduction $R = O(N)$ times and computing the majority over the $R$ outputs of the reduction. The circuit size to compute $\MCSP[n^c, 2^n/n^c]$ is thus $O(N \cdot 2^{O(n^\gamma)} n^{3c}) = O(N^{1+\varepsilon})$, for $\varepsilon = o(1)$. 
\end{proof}

\subsection{Towards a More Robust Theory}\label{ss:redmcsp}

The question of non-naturalizability of hardness magnification for $\MCSP[n^c/2n,n^{c}]$ is connected to the question of basing hardness of learning on the assumption $\NP\not\subseteq\Circuit[2^{O(n^\gamma)}]$.

\begin{proposition} Assume that for every $\gamma\in (0,1)$ there is $d\ge 2$ such that $\NP\not\subseteq\Circuit[2^{O(n^\gamma)}]$ implies hardness of learning $\Circuit[n^d]$ by $2^{n^\gamma}$-size circuits with error $1/n^d$. Then, there is a constant $e$ such that for every $\gamma\in (0,1)$ and $c\ge 1$, $\MCSP[n^c/2n,n^c]\notin\Circuit[N^{1+e\gamma c}]$ implies that there is no $\Ppoly$-natural property against \Ppoly.
\end{proposition}

\proof By Theorem \ref{thm:cikk_learn}, \Ppoly-natural property against \Ppoly implies that for every $d$ there is $\gamma<1/d$ and $2^{n^{\gamma}}$-size circuits learning $\Circuit[n^d]$ with error $1/n^d$. By our assumption, this implies $\NP\subseteq\Circuit[2^{O(n^\gamma)}]$. We can now use $\NP\subseteq\Circuit[2^{O(n^\gamma)}]$ as the assumption in the proof of Theorem \ref{pnp} to conclude that there is a constant $e$ independent of $\gamma$ such that for $c\ge 1$, $\MCSP[n^c/2n,n^c]\in\Circuit[N^{1+e\gamma c}]$. \qed

\bigskip

A form of the opposite implication holds as well if we assume \NP-completeness of \MCSP. Moreover, instead of the non-naturalizability of hardness magnification, we need to assume a reduction from worst-case \MCSP to approximate \MCSP. 

\begin{definition}
A p-time algorithm $A$ $k$-reduces $\MCSP[s,t]$ to $\MCSP[(s,0),(t,\eps)]$ if it maps instances of $\MCSP[s,t]$ to instances of $\MCSP[(s,0),(t,\eps)]$ and

1. For $f\in\Circuit[s]$, $A(\ttable(f))$ is the truth-table of a Boolean function in $\Circuit[s^k]$.

2. For $f\not\in\Circuit[t]$, $A(\ttable(f))$ is not $(1-\eps)$-approximable by circuits of size $t^{1/k}$.
\end{definition}

\begin{proposition}\label{prop:redcoms}
Assume there is a p-time algorithm $k$-reducing $\MCSP[s,t]$ to $\MCSP[(s,0),(t,\eps)]$ and that for all $0<\alpha<\beta<1$, $\MCSP[2^{\alpha n},2^{\beta n}]$ is \NP-complete. If for every sufficiently small $\alpha>0$ there is $\beta<1/k$ and a $2^{\beta n}$-time algorithm learning $\Circuit[2^{\alpha n}]$ with error $\eps$, then \Ptime= \NP.
\end{proposition}

\proof Let $A$ bet the p-time $k$-reduction from the statement and $\alpha>0$ be sufficiently small. Assume we can learn in $2^{\beta n}$-time $\Circuit[2^{k\alpha n}]$ with error $\eps$ and $k\alpha<\beta<1/k$. This implies that $\MCSP[(2^{k\alpha n},0),(2^{\beta n},\eps)]$ can be solved in p-time. Since $A$ reduces an \NP-complete problem $\MCSP[2^{\alpha n},2^{k\beta n}]$ to $\MCSP[(2^{k\alpha n},0),(2^{\beta n},\eps)]$, this shows that \Ptime= \NP.  \qed

\section{The Locality Barrier}\label{s:difficulties}

\subsection{Lower Bounds Above Magnification Threshold}

\subsubsection{The Razborov-Smolensky Polynomial Approximation Method}\label{ss:ACp_Polynomial_Approx}

In this section, we observe that the lower bound techniques of Razborov and Smolensky \citep{Razborov87, DBLP:conf/stoc/Smolensky87} can be ``localized.''
The following proposition instantiates the locality barrier for \HMFrontierAC.

\begin{proposition}[Locality Barrier for \HMFrontierAC]
  \label{p:frontier_B} The following results hold.
 \begin{itemize}
\item[\emph{(A1}$^\mathcal{O}$\emph{)}] \emph{(Oracle Circuits from Magnification)}
  $\mathsf{MKtP}[n^c,2n^c] \in \mathsf{AND}$-$\mathcal{O}$-$\mathsf{XOR}[N^{1.01}]$. More precisely, $\mathsf{MKtP}[n^c,2n^c]$ is computed by circuits with $N^{1.01}$ gates and of the following form: the output gate is an AND gate of fan-in $O(N)$, at the middle layer are oracle gates of fan-in $poly(n)$, and at the bottom layer are XOR gates.
\item[\emph{(A3}$^\mathcal{O}$\emph{)}] \emph{(Extension of  Lower Bound Techniques)}
  For a constant $d$, assume that $O_1, \cdots, O_d \in \Nat$
     satisfy $\prod_{i=1}^d O_i \le \sqrt{N}/\omega(\log N)^{d}$.
     Then $\mathsf{Majority}$ cannot be computed by
     a depth-$d$ polynomial-size oracle $(\AC[\oplus])^\mathcal{O}$ circuit
     whose oracle gates on the $i$-th level have fan-in at most $O_i$.
 \end{itemize}
\end{proposition}

The first item is immediate from the proof of Theorem \ref{thm_B1} in Section \ref{ss:frontier_ACXOR}.
In what follows, we prove the second item of Proposition \ref{p:frontier_B}.

Recall that the proof techniques of Razborov and Smolensky \citep{Razborov87, DBLP:conf/stoc/Smolensky87} consist of two parts:
The first lemma shows that any low degree polynomial cannot approximate $\mathsf{Majority}$.
(A simple proof sketch can be found in, e.g., \cite{Kopparty11_stoc_conf}.)
\begin{lemma}
  \label{lemma_MajorityVersusLowDegree}
  For any polynomial $p \in \mathbb{F}_2[x_1, \cdots, x_N]$ of degree $\le \sqrt{N} / 4$,
  $$
  \Pr_{x \sim \binset^N} [ p(x) \neq \mathsf{Majority}(x) ] \ge \frac{1}{4}.
  $$
\end{lemma}
The second lemma shows that $\AC[\oplus]$ circuits can be approximated by low degree polynomials.
We show that this argument can be localized.
\begin{lemma}
  \label{lemma_RazborovApprox}
  Let $C$ be a depth-$d$ polynomial-size oracle $\AC[\oplus]$ circuit
  whose oracle gates on the $i$-th level have fan-in at most $O_i$.
  Then there exists a polynomial $p \in \mathbb{F}_2[x_1, \cdots, x_N]$
  of degree $\le O(\log N)^{d} \cdot \prod_{i=1}^d O_i$
  such that
  $\Pr_{x \sim \binset^N} [ p(x) \neq C(x) ] < \frac{1}{4}.$
\end{lemma}

\begin{proof}
  [Proof Sketch]
  We convert each layer of the circuit $C$ into a low degree probabilistic polynomial $p$
  that approximates $C$.

  Consider the $i$-th level of a circuit $C$.
  $\mathsf{NOT}$, $\mathsf{OR}$, $\mathsf{AND}$, and $\mathsf{XOR}$ gates can be converted into a probabilistic polynomial of degree $O(\log N)$ and error $1/\poly(N)$ in the standard way \citep{Razborov87}.
  In order to represent an oracle gate $\mathcal{O}$ as a low-degree polynomial,
  we simply take the multilinear extension of the oracle gate $\mathcal{O}$.
  Note that, at the $i$-th level, the fan-in of the oracle gate $\mathcal O$ is bounded by $O_i$;
  thus the oracle gate at the $i$-th level can be represented as a polynomial of degree $\le O_i$.
  Thus, in either cases, any gate at $i$-th level can be represented as a probabilistic polynomial of degree $\max \{ O(\log N), O_i \}$.
   Continuing this for $i = 1, \cdots, d$ and composing resulting polynomials, we obtain a probabilistic polynomial of degree $\prod_{i=1}^d \max \{ O(\log N), O_i \}$
  that approximates $C$. This implies via standard techniques the existence of a (deterministic) polynomial of the same degree that correctly computes the circuit on most inputs.
\end{proof}

These two lemmas immediately imply the $\mathsf{Majority}$ lower bound for $(\AC[\oplus])^\mathcal{O}$:
\begin{proof}
  [Proof of \emph{(A3}$^\mathcal{O}$\emph{)} of Proposition \ref{p:frontier_B}]
  Suppose that there exists a depth-$d$ polynomial-size oracle $\AC[\oplus]$ circuit that computes $\mathsf{Majority}$ and satisfies the condition of Proposition \ref{p:frontier_B}.
  By Lemma \ref{lemma_RazborovApprox}, there exists a polynomial $p$ of degree at most $O(\log N)^{d} \cdot \prod_{i=1}^d O_i \le o(\sqrt{N})$
  that approximates $\mathsf{Majority}$.
  However, this contradicts Lemma \ref{lemma_MajorityVersusLowDegree}.
\end{proof}

Finally, we mention that an incomparable bound can be obtained by using a lower bound for $\AC[\oplus]$ interactive compression games.
\begin{proposition}
[{\citep[Corollary 5.3]{DBLP:conf/coco/OliveiraS15}}]
    \emph{(A3}$^\mathcal{O}$\emph{)} 
    $\mathsf{Majority} \not\in(\AC[\oplus])^\mathcal{O}[\mathsf{poly}(n)]$ if the total number of input wires
    in the circuit feeding the $\mathcal O$-gates is $N / (\log N)^{\omega(1)}$.
\end{proposition}

\subsubsection{The $\mathsf{Formula}$-$\mathsf{XOR}$ Lower Bound of \citep{Tal17}}\label{ss:Formula_XOR_Tal}

This section captures an instantiation of the locality barrier for HM Frontier B. Throughout this section we use the $\{-1,1\}$ realization of the Boolean domain (that is, $-1$ represents True and $1$ represents False). Let $\mathsf{Formula}$-$\mathsf{XOR}$ on variables $x_1, \dots, x_n$ be the class of formulas where the input leaves are labeled by parity functions of arbitrary arity over $x_1, \dots, x_n$. 

\begin{proposition}[Locality Barrier for HM Frontier B]
	\label{prop:frontier_C} 
	The following results hold.
	\begin{itemize}
		\item[\emph{(B1}$^\mathcal{O}$\emph{)}] \emph{(Oracle Circuits from Magnification)} For any $\varepsilon > 0$, $\mathsf{MCSP}[2^{n^{1/3}}, 2^{n^{2/3}}] \in \mathsf{Formula}$-$\mathcal{O}$-$\mathsf{XOR}[N^{1.01}]$, where every oracle $\mathcal{O}$ has fan-in at most $N^\varepsilon$ and appears in the layer right above the $\mathsf{XOR}$ leaves.
		\item[\emph{(B3}$^\mathcal{O}$\emph{)}] \emph{(Extension of Lower Bound Techniques)} For any $\delta > 0$, $\mathsf{InnerProduct}$ over $N$ input bits cannot be computed by $N^{2-3\delta}$-size $\mathsf{Formula}$-$\mathcal{O}$-$\mathsf{XOR}$ circuits with at most $N^{2-3\delta}$ oracle gates of fan-in $N^{\delta}$ in the layer right above the $\mathsf{XOR}$ leaves, for any oracle $\mathcal{O}$.
	\end{itemize}
\end{proposition}

To prove item 2 of Proposition \ref{prop:frontier_C}, we adapt Tal's \cite{Tal17} lower bound for bipartite formulas\footnote{A bipartite formula on variables $x_1, \dots, x_n, y_1, \dots, y_n$ is a formula such that each leaf computes an arbitrary function in either $(x_1, \dots, x_n)$ or $(y_1, \dots, y_n)$. $\mathsf{Formula}$-$\mathsf{XOR}$ circuits are a subset of bipartite formulas as one can always write $\oplus(x_1, \dots, x_{2n})$ as the parity of $\oplus(x_1, \dots, x_{n})$ and $\oplus(x_{n+1}, \dots, x_{2n})$.}, for which we need the following results.
\begin{lemma}[\cite{Rei11, Tal17}]
	\label{lem:rei_approx_degree}
	Let $F$ be a De Morgan formula of size $s$ which computes $f : \{-1,1\}^n \rightarrow \{\-1,1\}$. Then, there exists a multilinear polynomial $p$ over $\mathbb{R}$ of degree $O(\sqrt{s})$, such that for every $x \in \{-1,1\}^n$, $p(x) \in [F(x)-1/3, F(x)+1/3]$. 
\end{lemma}

For any function $f : \{-1,1\}^n \rightarrow \{-1,1\}$, $f$ is $\varepsilon$-correlated with a parity $p_S(x) = \prod_{i \in S} x_i$, if $\vert \mathbb{E}_{x \in \{-1,1\}^n} [f(x) \cdot p_S(x)] \vert \geq \varepsilon$. We have
\begin{lemma}
	\label{lem:formula_xor_poly_approx}
	For any $\delta > 0$, let $F(x_1, \dots, x_n)$ be a $\mathsf{Formula}$-${\mathcal{O}}$-$\mathsf{XOR}$ formula of size $s$, where every oracle $\mathcal{O}$ has fan-in at most $n^\delta$ and appears in the layer right above the $\mathsf{XOR}$ leaves. Then the following hold true:
	\begin{enumerate}
		\item There exists a multi-linear polynomial $p(x_1, \dots, x_n)$ over $\mathbb{R}$ with at most $s^{O(\sqrt{s})} \cdot 2^{n^\delta \cdot O(\sqrt{s})}$ monomials such that for every $x \in \{-1,1\}^n$, $\text{sign}(p(x)) = F(x)$.
		\item There exists a parity function $f_T(x_1, \dots, x_n)$ which is at least $\left( \frac{1}{s^{O(\sqrt{s})} \cdot 2^{n^\delta O(\sqrt{s})}} \right)$-correlated with $F$.
	\end{enumerate}
\end{lemma}

\begin{proof}
	We assume that the oracle function is a Boolean function on $n^\delta$ inputs. Let $t \leq s/n^\delta$ be the number of oracle gates in $F$. Let $p_1, \dots, p_s$ be the leaves of $F$, where each $p_i$ is an $\mathsf{XOR}$ gate over $x_1, \dots, x_n$ and every oracle gate $g_1, \dots, g_t$ is such that $g_i(x) = \mathcal{O}(p_{i1}(x), \dots, p_{i \ell}(x))$, where $\ell = n^\delta$ and $p_{ij} \in \{p_1,\dots,p_t\}$ for every $i \in [t], j \in [\ell]$.
	
	Let $F'$ be a De morgan formula obtained by replacing oracle gates in $F$ with new variables $z_i$ (for notational simplicity we assume that every leaf is an input to some oracle gate), for $i \in [t]$. We now use Lemma \ref{lem:rei_approx_degree} on $F'$ to get a degree $d = O(\sqrt{t})$ polynomial $q(z)$ such that for every $z \in \{-1,1\}^t$, $\text{sign}(q(z)) = F'(z)$. Expanding $q(z)$ as a multilinear polynomial :
	\begin{equation*}
	q(z) = \sum_{S \subseteq [t], \vert S \vert \leq d} \hat{q}(S) \prod_{i \in S} z_i
	\end{equation*}
	
	To prove the first item, we replace each $z_i$ by the original leaf and we get that for every $x \in \{-1,1\}^n$,
	\begin{equation*}
	\begin{split}
	F(x) &= \text{sign} \left( \sum_{S \subseteq [t], \vert S \vert \leq d} \hat{q}(S) \prod_{i \in S} g_i(x) \right) \\
	&= \text{sign} \left( \sum_{S \subseteq [t], \vert S \vert \leq d} \hat{q}(S) \prod_{i \in S} \left( \sum_{U \subseteq [\ell]} \hat{\mathcal{O}}(U) \prod_{j \in U} p_{ij} (x)\right) \right) \\
	&= \text{sign} \left( \sum_{\substack{S \subseteq [t] \\ S = \{i_1, \dots, i_{\vert S \vert}\} \\ \vert S \vert \leq d}}  \sum_{U_{i_1}, \dots, U_{i_{\vert S \vert}} \subseteq [\ell]} \hat{q}(S) \cdot \left( \prod_{1 \leq k \leq \vert S \vert} \hat{\mathcal{O}}(U_{i_k}) \prod_{j \in U_{i_k}} p_{i_k j}(x) \right) \right) \\
	\end{split}
	\end{equation*}
	where the second equality uses the fact that any Boolean function on $\ell$ inputs can be represented by a multilinear polynomial of degree at most $\ell$ where each coefficient is between $[-1,1]$. Clearly the number of monomials is at most $s^{O(\sqrt{s})} \cdot 2^{n^\delta \cdot O(\sqrt{s})}$.
	
	To prove the second item, firstly observe that for every $z \in \{-1,1\}^t$, $q(z) \cdot F'(z) \in [2/3, 4/3]$, because $\vert q(z) - F'(z) \vert \leq 1/3$ for every $z$. This also means that for the polynomial $r(x) = q(g_1(x), \dots, g_t(x))$, $\mathbb{E}_{x \in \{-1,1\}^n}[r(x) \cdot F(x)] \geq 2/3$. 
	
	Given that $\hat{q}(S) = \mathbb{E}_{z \in \{-1,1\}^s} \left[ q(z) \prod_{i \in S} z_i \right]$, we see that $\vert \hat{q}(S) \vert \leq 4/3$. We have
	\begin{equation*}
	\begin{split}
	2/3 &\leq \mathop{\mathbb{E}}_{x \in \{-1,1\}^n}[F(x) \cdot r(x)] \\
	&= \mathop{\mathbb{E}}_{x \in \{-1,1\}^n} \left[ F(x) \cdot \sum_{S \subseteq [t], \vert S \vert \leq d} \hat{q}(S) \prod_{i \in S} g_i(x) \right] \\
	&\leq \sum_{\substack{S \subseteq [t] \\ S = \{i_1, \dots, i_{\vert S \vert}\} \\ \vert S \vert \leq d}}  \sum_{U_{i_1}, \dots, U_{i_{\vert S \vert}} \subseteq [\ell]} \hat{q}(S) \cdot \prod_{1 \leq k \leq \vert S \vert} \hat{\mathcal{O}}(U_{i_k}) \cdot \mathop{\mathbb{E}}_{x \in \{-1,1\}^n} \left[ F(x) \cdot \prod_{1 \leq k \leq \vert S \vert} \prod_{j \in U_{i_k}} p_{i_k j}(x) \right]  \\
	\end{split}
	\end{equation*}
	Since, $\vert \hat{q}(S) \vert \leq 4/3$ for every $S \subseteq [t]$ and $\vert \hat{\mathcal{O}}(U) \vert \leq 1$ for every $U \subseteq [\ell]$, we see that there exists a set $S$ of size at most $d$ and sets $U_{i_1}, \dots, U_{i_{\vert S \vert}}$ such that $\Big\vert \mathbb{E}_{x \in \{-1,1\}^n} \left[ F(x) \cdot \prod_{1 \leq k \leq \vert S \vert} \prod_{j \in U_{i_k}} p_{i_k j}(x) \right] \Big\vert \geq \frac{1}{t^{O(\sqrt{t})} \cdot 2^{n^\delta O(\sqrt{t})}}\geq \frac{1}{s^{O(\sqrt{s})}\cdot 2^{n^\delta O(\sqrt{s})}}$. Taking $p_T$ be the parity of the parities given by $p_T = \prod_{1 \leq k \leq \vert S \vert} \prod_{j \in U_{i_k}} p_{i_k j}(x)$, we see that $p_T$ is $\frac{1}{s^{O(\sqrt{s})} \cdot 2^{n^\delta O(\sqrt{s})}}$-correlated with $F$.
\end{proof}

Define the Inner Product modulo 2 function, $\mathsf{InnerProduct}_n : \{-1,1\}^n \times \{-1,1\}^n \rightarrow \{-1,1\}$ as $IP_n(x,y)=(-1)^{\sum_{i=1}^n (1-x_i)(1-y_i)/4}$.

\begin{proof}[Proof Sketch of Proposition \ref{prop:frontier_C}]
	The first item follows from an inspection of the proof of 
	Theorem \ref{thm:mcspker} in Section \ref{ss:frontier_MCSP_formula_xor}. Theorem \ref{thm:mcspecc} gives the same oracle circuit construction (with different oracles) under the assumption $\mathsf{QP}\subseteq \Ppoly$.
	
	The second item follows from Lemma \ref{lem:formula_xor_poly_approx}. We observe that three different techniques used to show $\mathsf{Formula}$-$\mathsf{XOR}$ lower bounds localize. Firstly, Tal's lower bound based on sign rank shows that the sign rank of any $\mathsf{Formula}$-$\mathsf{XOR}$ circuit $F$ is at most the number of monomials in the polynomial $p$ given by the first item of lemma \ref{lem:formula_xor_poly_approx}. Since this is at most $s^{O(\sqrt{s})} \cdot 2^{n^\delta \cdot O(\sqrt{s})}$ and $\mathsf{InnerProduct}$ has a sign rank which is at least $2^{n/2}$~\cite{Forster02}, the lower bound follows. Secondly, Tal's lower bound based on the discrepancy of a function also localizes, as he shows that the discrepancy of $F$ is at least a constant times the correlation of $F$ with the parity $f_T$ given by item 2 of Lemma \ref{lem:formula_xor_poly_approx}, which is at least $\Omega \! \left( \frac{1}{s^{O(\sqrt{s})} \cdot 2^{n^\delta O(\sqrt{s})}} \right)$, whereas the discrepancy of the inner product is at most $1/2^{n/2}$ (cf.~\citep[Lemma 14.5]{DBLP:books/daglib/0028687}), thus proving the given lower bound for inner product. Finally, we also observe that the lower bound technique of showing high correlation of $F$ with some parity $f_T$ and the fact that inner product has exactly $2^{-n/2}$-correlation with any parity also localizes to give the same lower bound. 
\end{proof}

\subsubsection{Almost-Formula Lower Bounds}\label{ss:Almost_Formulas_LB}

This section captures an instantiation of the locality barrier for HM Frontier C. We recall the following definition. Consider an $s$-almost formula. Each gate $G$ of $F$ with fanout larger than 1 is computed by a formula with inputs being either the original inputs of $F$ or gates of $F$ with fanout larger than 1. We call any maximal formula of this form a \emph{principal} formula of $G$.

\begin{theorem}[Locality Barrier for HM Frontier C]\label{thm:fclb} The following results hold.
	\begin{itemize}
		\item[\emph{(C1}$^\mathcal{O}$\emph{)}] \emph{(Oracle Circuits from Magnification)} $\MCSP[2^{n^{1/2}}/2n,2^{n^{1/2}}]$ is computable by $2^{O(n^{1/2})}$-almost formulas of size $2^{n+O(n^{1/2})}$ with oracles of fanin $2^{O(n^{1/2})}$ at the bottom layer of principal formulas computing gates with fanout larger than 1.
		\item[\emph{(C3}$^\mathcal{O}$\emph{)}] \emph{(Extension of Lower Bound Techniques)} For every $\eps<1$, \PARITY is not in $n^{\eps}$-$\mathsf{almost}$-$\Formula[n^{2-9\eps}]$ even if the almost-formulas are allowed to use arbitrary oracles of fanin $<n^{\eps}$ at the bottom layer of principal formulas computing gates with fanout larger than $1$.
	\end{itemize}
\end{theorem}

\proof The first item follows by inspecting the proof of Theorem \ref{pnp}. It is not hard to see that $\MCSP[2^{n^{1/2}}/2n,2^{n^{1/2}}]$ is computable by $2^{O(n^{1/2})}$-almost formulas $F_N$ of size $2^{n+O(n^{1/2})}$ with local oracles of fanin $2^{O(n^{1/2})}$. Moreover, the only gates of fanout larger than 1 are the gates computing anticheckers $y_1,\dots,y_{2^{O(n^{1/2})}}$ with bits $f(y_1),\dots, f(y_{2^{O(n^{1/2})}})$. We want to show that the local oracles are at the bottom of principal formulas generating gates with fanout larger than 1. In order to achieve this we need to modify 
formulas $F_N$ a bit. 

First, note that $F_N$ contains an oracle which is applied on top of anticheckers $y_1,\dots,y_{2^{O(n^{1/2})}}$ with bits $f(y_1),\dots, f(y_{2^{O(n^{1/2})}})$. In order to ensure that 
this oracle is at the bottom of a principal formula computing a gate with fanout bigger than 1 we simply add dummy negation gates to the output gate and the gates computing anticheckers $y_1,\dots,y_{2^{O(n^{1/2})}}$ with bits $f(y_1),\dots, f(y_{2^{O(n^{1/2})}})$, if necessary.

Second, note that each $y_{i+1}$, $f(y_{i+1})$ is generated as follows: 1. if $R_f(y_1,\dots,y_i)\ge 2n^2$ then a subformula $F'$ generates anticheckers $y_{i+1}$, $f(y_{i+1})$, and 2. if $R_f(y_1,\dots,y_i)< 2n^2$ then a subformula $F''$ generates anticheckers $y_{i+1}$, $f(y_{i+1})$. In both cases we replace predicates $R_f(y_1,\dots,y_i)< 2n^2$ by oracles. In case 1, subformulas of $F'$ with oracles at the bottom compute predicates $F^{r,h}$ from the proof of Lemma \ref{lem}. This process generates a set of $2^{O(n^{1/2})}$ potential anticheckers. $F'$ chooses the right antichecker by applying another oracle. In order to ensure that this top oracle is at the bottom of a principal formula, we add dummy negation gates to the gates generating the potential anticheckers. This increases the number of gates with fanout larger than 1 only by $2^{O(n^{1/2})}$. 
In case 2, $y_{i+1}$, $f(y_{i+1})$ is generated by oracles outputting circuits which have not been killed yet and evaluating them on all possible inputs. Here we ensure that the oracles are at the bottom by asking them to perform both tasks: choose the next alive circuit and evaluate it on a given input. The oracle selecting the right antichecker from the set of potential anticheckers is treated in the same way as in case 1. All in all, we obtain the desired oracle almost formulas.
\smallskip

The second item is proved analogously to Theorem \ref{thm:parityflalikelb}.  For the sake of contradiction assume \PARITY has $n^\eps$-almost formulas of size $n^{2-9\eps}$ with local oracles at the bottom of principal formulas. Since there are only $n^{\eps}$ gates of fanout $>1$, we can replace these gates by constants and obtain formulas $F_n$ of size $n^{2-8\eps}$ with local oracles at the bottom computing \PARITY with probability $\ge 1/2+1/2^{n^{\eps}}$. Let $L'(f)$ be the size (i.e. the number of leafs) of the smallest formula with local oracles at the bottom computing $f$. Since oracles have fanin $<n^{\eps}$ and are located at the bottom, each function $f:\{-1,1\}^n\mapsto\{-1,1\}$ can be approximated by a polynomial of degree $O(t\sqrt{L'(f)}\frac{\log n}{\log\log n}n^{\eps})$ up to point-wise error of $2^{-t}$. This implies that each formula of size $o((n/t)^2(\log\log n/\log n)^2(1/n^{\eps})^2)$ with local oracles at the bottom computes \PARITY with probability at most $1/2+1/2^{t+O(1)}$ (for large enough $t$). Taking $t=n^{2\eps}$ we get a contradiction. \qed

\subsubsection{$\GapAND$-$\FM$ Lower Bounds}\label{ss:GapAND_Formulas_LB}

This section captures an instantiation of the locality barrier for HM Frontier D.

\begin{theorem}[Locality Barrier for HM Frontier D]\label{thm:GapAND-locality-barrier}
	The following results hold.
	\begin{enumerate}
		\item[\emph{(D1}$^\mathcal{O}$\emph{)}] \emph{(Oracle Circuits from Magnification)} 
		$\MCSP[2^{\sqrt{n}}] \in \GapAND_{O(N)}$-$O_{N^{o(1)}}$-$\FM[N^2]$.
		
		\item[\emph{(D3}$^\mathcal{O}$\emph{)}] \emph{(Extension of Lower Bound Techniques)} 
		$\Andreev_N \notin \GapAND_{O(N)}$-$O_{N^{\beta}}$-$\FM[N^{3-\eps}]$, for $0 < \beta < \eps < 1$.
		
		\item[\emph{(D3}$^\mathcal{O}$\emph{)}] Furthermore, 
		$\MCSP[2^n / n^4] \notin \GapAND_{O(N)}$-$O_{N^{\beta}}$-$\FM[N^{3-\eps}]$, for $0 < \beta < \eps < 1$.
	\end{enumerate}
	
\end{theorem}

Item 1 of the theorem above follows directly from Theorem~\ref{thm:mcspker}.

Next we show that the classical $N^{3-o(1)}$ formula size lower bound for the Andreev's function~\cite{Hastad98,Tal14} localizes, even in the presence of a $\GapAND$ gate of bounded fan-in at the top of the formula.

\begin{proofof}{Item 2}
Let $m = N/2$, recall that $\Andreev_N$ is defined on a $2 m$-bit string $z = x \circ y$, where $x,y \in \{0,1\}^m$. For simplicity, we assume $m$ is a power of $2$ in the following.

$\Andreev_N(x,y)$ first partitions $x$ into $\log m$ blocks $x_1,x_2,\dotsc,x_{\log m}$, each of length $m / \log m$. After that, it computes $i \in \{0,1\}^{\log m}$ as $i = \PARITY(x_1)\circ\PARITY(x_2)\circ\dotsc\PARITY(x_{\log m})$. It then treats $i$ as an integer from $[m]$, and outputs $y_i$.

Now, suppose there is a $\GapAND_{O(N)}$-$O_{N^{\beta}}$-$\FM[N^{3-\eps}]$ formula for $\Andreev_N$. Suppose we fix the $y$ variables to a string $w \in \{0,1\}^m$, and apply a random restriction keeping exactly one variable from each block alive to $x$ variables, then w.p. $0.9$, we obtain a $\GapAND_{O(N)}$-$O_{N^\beta}$-$\FM[N^{1-\eps} \cdot \polylog(N)]$ formula computing $f_w : \{0,1\}^{\log m} \to \{0,1\}$~\cite{Tal14}.

That is, for all $w \in \{0,1\}^m$, there exists an $O_{N^\beta}$-$\FM[N^{1-\eps} \cdot \polylog(N)]$ formula $0.8$-approximating $f_w$. Note that there are at most $2^{N^{1-\eps + \beta}\cdot \polylog(N)}$ such $O_{N^\beta}$-$\FM[N^{1-\eps} \cdot \polylog(N)]$ formulas, and there are $2^{N}$ possible $w$'s (Note that $O$ is a fixed oracle which does not depend on $w$). Since each $O_{N^\beta}$-$\FM[N^{1-\eps} \cdot \polylog(N)]$ formula can only $0.8$-approximate $2^{\alpha \cdot N}$ many functions from $\{0,1\}^{\log m} \to \{0,1\}$ for a constant $\alpha < 1$, there must exist a $w$ such that $f_w$ cannot be $0.8$-approximated by such formulas, contradiction.
\end{proofof}

Next, we observe that the $N^{3-o(1)}$ formula lower bound for $\MCSP$~\cite{CheraghchiKLM19_eccc_journals} also localizes.

\begin{proofof}{Item 3}

We first observe that the PRG construction of~\cite{CheraghchiKLM19_eccc_journals} also works for oracle formulas. (We omit the details of this proof.)
	
\begin{claim}[\cite{CheraghchiKLM19_eccc_journals}]\label{claim:PRG}
	For $0 < \beta < \eps < 1$, there is $M = N^{1-\Omega_{\beta,\eps}(1)}$ and a PRG $G:\{0,1\}^{M} \to \{0,1\}^{N}$ such that the following hold.
	
	\begin{enumerate}
		\item For each fixed $z \in \{0,1\}^{M}$, $G(z)$, when interpreted as a function from $\{0,1\}^{\log N} \to \{0,1\}$, can be computed by a circuit of size $N^{1-\Omega(1)}$.
		\item For all $O_{N^{\beta}}$-$\FM[N^{3-\eps}]$ formulas $C$, we have
		\[
		\left| \Pr_{z \in \{0,1\}^N}[C(z) = 1] - \Pr_{z \in \{0,1\}^M}[C(G(z)) = 1] \right| \;\le\; 0.01.
		\]
	\end{enumerate}
\end{claim}

	Now, suppose $\MCSP[2^n/n^4]$ on $N = 2^n$ bits can be computed by a $\GapAND_{O(N)}$-$O_{N^\beta}$-$\FM[N^{3-\eps}]$ formula $C$. Let $C_1,C_2,\dotsc,C_{b \cdot N}$ be the $O_{N^\beta}$-$\FM[N^{3-\eps}]$ sub-formulas of $C$ under the top $\GapAND$ gate, where $b$ is a constant.
	
	We know that
	\[
	\Pr_{z \in \{0,1\}^N}[\MCSP[2^n/n^4](z) = 1] = o(1).
	\]
	Since $C$ computes $\MCSP[2^n/n^4]$, and $C(x) = 0$ implies $C_i(x) = 0$ for at least a $0.9$ fraction of $i \in [b \cdot N]$. We have that
	\[
	\Pr_{i \in [b \cdot N],\,z \in \{0,1\}^N} \left[C_i(z) = 1 \right] \le 0.2.
	\]
	
	On the other side, by the definition of $\MCSP[2^n/n^4]$, and the Item (1) of~\cref{claim:PRG}, it follows that
	\[
	\Pr_{z \in \{0,1\}^M}[\MCSP[2^n/n^4](G(z)) = 1] = 1.
	\]
	Again, since $C$ computes $\MCSP[2^n/n^4]$, and $C(x) = 1$ implies $C_i(x) = 1$ for all $i \in [b \cdot N]$. We have that
	\[
		\Pr_{i \in [b \cdot N],\,z \in \{0,1\}^M} \left[C_i(G(z)) = 1 \right] = 1.
	\]
	
	Therefore, there must exist an $i$ such that
	\[
	\left| \Pr_{z \in \{0,1\}^N}[C_i(z) = 1] - \Pr_{z \in \{0,1\}^M}[C_i(G(z)) = 1] \right| \ge 0.5,
	\]
	which is a contradiction to Item (2) of~\cref{claim:PRG}.
\end{proofof}

Finally, we show that there is a language in $\mathsf{E}$ which cannot be computed by $\GapAND_{O(N)}$-$\FM[N^{3-\eps}]$ formulas, but it \emph{can} be computed by an $O_{N^{o(1)}}$-$\FM[N^2]$ formula. Therefore, this lower bound does not localize in the sense of Theorem~\ref{thm:GapAND-locality-barrier}.

\begin{theorem}\label{thm:E-lower-bound-for-FM}
	There is a language $L \in \mathsf{E}$, such that $L \notin \GapAND_{O(N)}$-$\FM[N^{3-\eps}]$ for all constants $\eps > 0$, but $L \in O_{N^{o(1)}}$-$\FM[N^2]$.
\end{theorem}
\begin{proof}
	The function $L$ is very similar to the Andreev's function. On an input $x$ of length $N$, let $m = \log N$ (we assume $N$ is a power of $2$ for simplicity). To avoid the second input to $\Andreev_N$, we want to find a function $f_{\sf hard}: \{0,1\}^{m} \to \{0,1\}$ which cannot be $0.8$-computed by $N^{1-\eps/2}$ formulas in $2^{O(N)}$ time (such a function exists by a simple counting argument). To find $f_{\sf hard}$, we simply enumerate all possible functions $f : \{0,1\}^m \to \{0,1\}$, and check whether it can be $0.8$-approximated by an $N^{1-\eps/2}$-size formula.
	
	There are $2^{2^m} = 2^{N}$ possible functions on $m$ bits, and $(N^{1-\eps/2})^{O(N^{1-\eps/2})} = 2^{N^{1-\eps/2} \cdot \polylog(N)}$ many formulas of $N^{1-\eps/2}$ size. Hence, a straightforward implementation of the algorithm runs in $2^{O(N)}$ time.
	
	Next, $L$ partitions $x$ into $m$ blocks $x_1,x_2,\dotsc,x_{m}$, each of length $N / m$. After that, it computes $i \in \{0,1\}^{m}$ as $i = \PARITY(x_1)\circ\PARITY(x_2)\circ\dotsc\PARITY(x_{m})$. It then outputs $f_{\sf hard}(i)$.
	
	Now, suppose there is a $\GapAND_{O(N)}$-$\FM[N^{3-\eps}]$ for $L$. We apply a random restriction keeping exactly one variable from each block alive, then w.p. $0.9$, we obtain a $\GapAND_{O(N)}$-$\FM[N^{1-\eps} \cdot \polylog(N)]$ formula for $f_{\sf hard}$~\cite{Tal14}, which implies that there is an $N^{1-\eps} \cdot \polylog(N)$-size formula $0.8$-approximating $f_{\sf hard}$, contradiction.
	
	Finally, it is easy to verify that $L \in \mathsf{E}$ and $L \in O_{N^{o(1)}}$-$\FM[N^2]$.
\end{proof}

\subsubsection{$\mathsf{AC}^0$ Lower Bounds via Random Restrictions}\label{ss:AC_Random_Restrictions}

This section states and proves a result capturing an instantiation of the locality barrier for HM Frontier E.

\begin{proposition}[Locality Barrier for HM Frontier E]
	\label{p:frontier_A} The following results hold.
	\begin{itemize}
		\item[\emph{(E1}$^\mathcal{O}$\emph{)}] \emph{(Oracle Circuits from Magnification)} For each $k = (\log n)^C$ and every large enough depth $d$, $(n-k)$\emph{-}$\mathsf{Clique} \in (\mathsf{AC}^0_d)^\mathcal{O}[m^{1 + \varepsilon_d}]$, where $\varepsilon_d \to 0$ as $d \to \infty$, and the corresponding circuit employs a single oracle gate $\mathcal{O}$ of fan-in at most $O((\log n)^{4C})$.
		\item[\emph{(E3}$^\mathcal{O}$\emph{)}] \emph{(Extension of Lower Bound Techniques)} $\mathsf{Parity} \notin (\mathsf{AC}^0)^\mathcal{O}[\mathsf{poly}(n)]$ if the total number of input wires in the circuit feeding the $\mathcal{O}$-gates is $n/(\log n)^{\omega(1)}$. 
	\end{itemize}
\end{proposition}

\begin{proof}
	The first item is established by inspection of the proof of Proposition \ref{p:hm_clique}, which relies on the circuit construction from \citep{OS18_mag_first} and a straightforward translation between vertex cover and clique detection. Recall that the circuit in \citep{OS18_mag_first} simulates a well-known kernelization algorithm for $k$-$\mathsf{Vertex}$-$\mathsf{Cover}$. This algorithm produces a graph $H$ containing $O(k^2)$ vertices and a new parameter $k_H \leq k$. This graph can be described by a string of length $O(k^4)$, and the pair $(H,k_H)$ becomes the input string to the single oracle $\mathcal{O}$ that is necessary in the oracle circuit construction. (If $\mathcal{O}$ solves vertex cover, the resulting oracle circuit correctly solves $(n-k)$-$\mathsf{Clique}$.)
	
	The second item easily follows by simulating oracle circuits via interactive compression games (see e.g.~\citep[Section 5]{DBLP:conf/coco/OliveiraS15}). In other words, one can view a circuit with oracles as an interactive protocol between two parties, where one of them has unbounded computational power, and the other is restricted to computations in a fixed circuit class. The total number of wires feeding the oracle gates corresponds to the number of bits sent to the unbounded party. The desired lower bound for oracle circuits then follows immediately from the main result from \citep{DBLP:conf/focs/ChattopadhyayS12}, which shows that the random restriction method can be extended to establish limitations on circuits with oracle gates of large fan-in.
\end{proof}

Informally, the main difficulty with the use of random restrictions in connection to HM Frontier E is that as soon as one simplifies a boolean circuit so that the oracle gate $\mathcal{O}$ is directly fed by input literals, one can fix just $(\log n)^{O(C)}$ input variables and eliminate this gate. Sacrificing such a small number of coordinates won't affect a typical  worst-case lower bound based on the random restriction method.


\subsubsection{Lower Bounds Through Reductions}\label{ss:reductions} 


Consider a reduction of \PARITY to $\MCSP[2^{n^{1/2}}/2n,2^{n^{1/2}}]$ by subquadratic-size $n^\eps$-almost formulas with $n^{\eps'}$ MCSP (possibly non-local) oracles at the bottom of each principal formula computing a gate with fanout $>1$. By Theorem \ref{thm:fclb}, such a reduction would imply $\MCSP[2^{n^{1/2}}/2n,2^{n^{1/2}}]\notin n^\eps$-$\mathsf{almost}$-$\Formula[N^{1.1}]$ assuming that after replacing all oracles by $n^\eps$-almost formulas of size $N^{1.1}$ the total size of the resulting circuit remains $<N^{2-9(\eps+\eps')}$. In combination with hardness magnification, this would give us $\NP\not\subseteq\NC$. Unfortunately, Theorem \ref{thm:fclb} rules this possibility out.

\begin{corollary}\label{cor:redafla} \PARITY is not computable by subquadratic-size $n^\eps$-almost formulas with $n^{\eps'}$ oracle gates computing $\MCSP[2^{n^{1/2}}/2n,2^{n^{1/2}}]$, assuming that after replacing all oracles by $n^\eps$-almost formulas of size $N^{1.1}$ the total size of the resulting circuit remains $<N^{2-9(\eps+\eps')}$ for $\eps+\eps'<1$.
\end{corollary}

\proof Assume the reduction in question exists. By Theorem \ref{pnp}, for every $\eps>0$ and all sufficiently big $n$, $\MCSP[2^{n^{1/2}}/2n,2^{n^{1/2}}]$ is computable by $N^{1.1}$-size $n^{\eps}$-almost formulas with local oracles at the bottom of principal formulas computing gates with fanout $>1$. By the assumption, if we replace the \MCSP oracles in the reduction by almost-formulas with local oracles, the resulting circuit is an $n^{\eps+\eps'}$-almost formula of size $N^{2-9(\eps+\eps')}$ with oracles of bounded fan-in. This contradicts the second item of Theorem \ref{thm:fclb}. \qed

\bigskip

Analogous arguments rule out the possibility of establishing strong lower bounds via reductions also in other HM frontiers.

\subsection{Lower Bounds Below Magnification Threshold}

The localizations presented in this section show that one cannot obtain strong circuit lower bounds by ``lowering the threshold'' in certain hardness magnification proofs. As a consequence of one of our results (Theorem \ref{thm:local-HS17} in Section \ref{ss:local_HS}), we also refute the Anti-Checker Hypothesis  from~\cite{OPS19_CCC}.

\subsubsection{$\mathsf{AC}^0$ Lower Bounds via Pseudorandom Restrictions}\label{ss:AC_pseudo_random_restrictions}

In this section we show that the $\AC$ lower bounds proved for $\MCSP$ ($\MKtP$) via pseudorandom restrictions~\cite{CheraghchiKLM19_eccc_journals} (see also Section~\ref{ss:frontier_ACXOR}) localize in a very strong sense. 

We use $\AC_d[O_1,O_2,\dotsc,O_d]$ to denote $\AC_d$ circuits extended with arbitrary oracles, such that oracle gates on the $i$-th level (the gates whose distance from the inputs is $i$) have fan-in at most $O_i$.

\begin{theorem}\label{thm:localize-MCSP-AC0}
	There is a constant $c$ such that for all $\eps > 0$, constants $d$, and $O_1,O_2,\dotsc,O_d$ such that $\prod_{i=1}^{d} O_i \le N / (\log N)^{\omega(1)}$, $\MCSP[n^c, n^{2c}] \notin \AC_d[O_1,O_2,\dotsc,O_d][\poly(N)]$.
\end{theorem}

\vspace{0.1cm}

\begin{remark}
	We remark that the constraint on oracles in the above theorem is incomparable to the second item of Proposition~\ref{p:frontier_A}. Here we focus on the maximum oracle fan-in at each level, while there the focus is on the total fan-in of all oracles. A lower bound result for an explicit problem with parameters similar to Theorem \ref{thm:localize-MCSP-AC0} is not known for $\mathsf{AC}^0$ oracle circuits extended with parity gates \emph{(}see \emph{\citep{DBLP:conf/coco/OliveiraS15}} for results in this direction\emph{)}.
\end{remark}

\vspace{0.1cm}

We are going to apply Lemma~\ref{lemma_DerandomizedSwitchingLemma}, together with the following well-known results on $k$-wise independence fooling CNFs.

\vspace{0.1cm}

\begin{lemma}[\cite{Bazzi09,Tal17_coco_conf}]\label{lm:k-wise-independence-fool-CNFs}
	$k=O(\log(M/\eps) \cdot \log(M))$-wise independent distribution $\eps$-fools $M$-clauses CNFs.
\end{lemma}

\vspace{0.1cm}

Combining Lemma~\ref{lemma_DerandomizedSwitchingLemma} and Lemma~\ref{lm:k-wise-independence-fool-CNFs}, we have the following lemma.

\vspace{0.1cm}

\begin{lemma}\label{lm:shrink-CNF}
	Let $\varphi$ be a $t$-width $M$-clause CNF formula over $N$ inputs, $p = 2^{-q}$ for some $q \in \Nat$, and $\eps_0 > 0$ be a real. There is a $p$-regular, 
	\[
	k = \Theta(\log(M \cdot 2^{t(q+1)}/\eps_0) \cdot \log(M \cdot 2^{t(q+1)}) \cdot q^{-1})\text{-wise}
	\]
	independent random restriction $\boldsym{\rho}$ such that
	$$
	\Pr_{\rho \sim \boldsym{\rho} } [ \DT(\varphi\restriction_{\rho}) > s ]
	\le 2^{s+t+1} (5pt)^s + \eps_0 \cdot 2^{(s+1)(2t + \log M)}.
	$$
	\noindent Moreover, $\boldsym{\rho}$ is samplable with $O(t\cdot q\cdot\polylog(M,N) \cdot \log(1/\eps_0))$ bits, and each output coordinate of the random restriction can be computed in time polynomial in the number of random bits.
\end{lemma}

\vspace{0.1cm}

The moreover part follows from standard construction of $k$-wise independent distributions~(see e.g.~\cite{Vadhan12}).

We also need the following lemma which states that an arbitrary oracle with inputs being small-size decision trees shrinks to a small-size decision tree with high probability, under suitable pseudorandom restrictions.

\begin{lemma}\label{lm:shrink-oracle}
	Let $O:\{0,1\}^{T} \to \{0,1\}$ be an arbitrary function, and $D_1,D_2,\dotsc,D_T$ be $T$ $k$-query decision trees on variables $x_1,x_2,\dotsc,x_N$. Let $F := O \circ (D_1,D_2,\dotsc,D_T)$ be their compositions. For $s \in \mathbb{N}$, and all $k(s+1)$-wise independent $1/(T \cdot k^2)$-regular random restriction $\boldsym{\rho}$, we have
	\[
	\Pr_{\rho \sim \boldsym{\rho}}[\DT(F\restriction_{\rho}) > s] \le \left( \frac{k(s+1)}{2e^2} \right)^{-(s+1)}.
	\]
\end{lemma}
\begin{proof}
	
	\newcommand{\event}{\mathcal{E}}
	\newcommand{\algo}{\mathbb{A}}
	\newcommand{\Eval}{\textsf{Eval}}
	
	We focus on the following particular decision tree for evaluating $\{D_1,D_2,\dotsc,D_T\}$ with respect to a restriction $\rho : [N] \to \{0,1,*\}$:
	
	\begin{framed}
		
		\centering Algorithm $\Eval(\rho,D_1,D_2,\dotsc,D_T)$.
		
		\begin{itemize}
			\item For $i$ from 1 to T:
			\begin{itemize}
				\item Simulate decision tree $D_i$ with restriction $\rho$. That is, when $D_i$ queries an index $j$, we feed $\rho_j$ to $D_i$ if $\rho_j \in \{0,1\}$, and query the $j$-th bit otherwise.
			\end{itemize}
			\item Let $\alpha_i$ be the output of the $i$-th decision tree, we output $\alpha = (\alpha_1,\alpha_2,\dotsc,\alpha_T)$.
		\end{itemize}
	\end{framed}
	
	To obtain a decision tree for $F\restriction_{\rho}$, we can run $\Eval(\rho,D_1,D_2,\dotsc,D_T)$ to obtain $\alpha$ first and output $F(\alpha)$ at the end.
	
	Let $\widetilde{\DT}(F\restriction_{\rho})$ be the query complexity of the above decision tree. Since $\DT(F\restriction_{\rho}) \le \widetilde{\DT}(F\restriction_{\rho})$ ($\DT(F\restriction_{\rho})$ is the minimum complexity among all decision trees computing $F\restriction_{\rho}$), it suffices to bound
	\[    
	\Pr_{\rho \sim \boldsym{\rho}}[\widetilde{\DT}(F\restriction_{\rho}) > s].
	\]
	
	Consider the event that $\widetilde{\DT}(F\restriction_{\rho}) > s$, it is equivalent to that there exists a string $w \in \{0,1\}^s$, such that if we fix the first $s$ queried unrestricted bits in $\rho$ according to $w$, $\Eval$ ends up querying $>s$ bits. (Note that since we only care about whether $\widetilde{\DT}(F\restriction_{\rho}) > s$, we can force the algorithm to abort if it tries to make the $(s+1)$-th query.)
	
	Now, suppose we fix the string $w$, then the number of queries made by $\Eval$ only depends on $\rho$. Suppose the algorithm has queried at least $s+1$ bits, we let $D'_1,D'_2,\dotsc,D'_{t}$ ($t \le s+1$) be the decision trees in which the algorithm made queries during the first $s+1$ queries.  This implies that if we run $\Eval(\rho,D'_1,D'_2,\dotsc,D'_t)$ with respect to the same string $w$, the algorithm also makes at least $s+1$ queries.
	
	Now, since $\boldsym{\rho}$ is $k(s+1)$-wise independent. The probability that $\Eval(\rho,D'_1,D'_2,\dotsc,D'_t)$ makes at least $s+1$ queries with respect to the fixed string $w$ is bounded by
	\begin{align*}
	&(T \cdot k^2)^{-(s+1)} \cdot \binom{t \cdot k}{s+1} \\
	\le &(T \cdot k^2)^{-(s+1)} \cdot \left( \frac{t \cdot k \cdot e}{s+1} \right)^{s+1} \\
	\le &\left( \frac{T \cdot k \cdot (s+1)}{t \cdot e} \right)^{-(s+1)} \le \left( \frac{T \cdot k}{e} \right)^{-(s+1)}.
	\end{align*}
	
	Putting everything together, we have
	\begin{align*}
	&\Pr_{\rho \sim \boldsym{\rho}}[\widetilde{\DT}(F\restriction_{\rho}) > s] \\
	\le &2^{s} \cdot \left( \frac{T \cdot k}{e} \right)^{-(s+1)} \cdot \sum_{t=0}^{s+1} \binom{T}{t} \\
	\le &2^{s} \cdot \left( \frac{T \cdot k}{e} \right)^{-(s+1)} \cdot \left( \frac{T \cdot e}{s+1} \right)^{s+1} \\
	\le &\left( \frac{k \cdot (s+1)}{2 e^2} \right)^{-(s+1)}.
	\end{align*}
\end{proof}

\begin{remark}
	Clearly, Lemma~\ref{lm:shrink-oracle} also holds when $\boldsym{\rho}$ is $k(s+1)$-wise independent and $p$-regular, for $p \le \frac{1}{T \cdot k^2}$.
\end{remark}


Now we are ready to prove Theorem~\ref{thm:localize-MCSP-AC0}.

\begin{proofof}{Theorem~\ref{thm:localize-MCSP-AC0}}
	
	We assume $N$ and $\log N$ are both powers of $2$ for simplicity.  Let $p = 1/\log^5 N$, $\eps_0 = 2^{-\log^6 N}$, $s = t = 10 \log^2 N$, $M = 2^{s} \cdot N^{\log N}$, and $\boldsym{\rho}$ be the $k$-wise independent $p$-regular random restriction guaranteed by Lemma~\ref{lm:shrink-CNF}. Note that we have $k = \omega(\log^6 N)$ and $k = \log^{O(1)} N$.
	
	Let $C \in \AC_d[O_1,O_2,\dotsc,O_d]$ be a circuit with $S$ gates computing $\MCSP[n^{c},n^{2c}]$. For each $i \in [d]$, let $S_i$ be the number of gates at level $i$ (i.e., the gates whose distance from the input gates is $i$). Recall that $O_i$ is the maximum oracle fan-in at level $i$. We are going to prove the stronger claim that $S = \Omega(N^{ \log N})$. Now, suppose for the sake of contradiction that $S \le N^{\log N} / 8$.

	Now we proceed in $d$ iterations. We will ensure that at the end of the $i$-th iteration, all gates at level $i$ become $s$-query decision trees with high probability. At the $i$-th iteration, we apply $\rho$
	\[
	\tau_i = \lceil \log_{1/p} O_i \rceil + 1
	\]
	times. It is straightforward to see that the composition of $\tau_i$ independent restrictions from $\boldsym{\rho}$ is a $k$-wise independent $p_i$-regular random restriction for $p_i = p^{\tau_i} \le \frac{1}{O_i \cdot \log^5 N}$.
	
	Note that each oracle gate at original level $i$ has inputs computed by $s$-query decision trees (at the first step, one can treat the input variables as $1$-query decision trees).    By Lemma~\ref{lm:shrink-oracle} and noting that $k \ge s(s+1)$ and $O_i \cdot \log^5 N \ge O_i \cdot s^2$, with probability at least
	\[
	1 - S_i \cdot \left(\frac{s (s+1)}{2 e^2} \right)^{-(s+1)} \ge 1 - S_i \cdot N^{-\log N},
	\]
	all oracle gates at level $i$ become $s$-query decision trees after these $\tau_i$ restrictions.
	
	Similarly, note that each AND / OR gate at level $i$ are equivalent to a CNF or DNF with width-$s$ and size at most $2^s \cdot S$. By Lemma~\ref{lm:shrink-CNF}, again with probability at least
	\begin{align*}
	   &1 - S_i \cdot \left( 2^{s+t+1} (5p t)^{s} + \eps_0 \cdot 2^{(s+1)(2 t + \log M)} \right) \\
	\ge&1 - S_i \cdot \left( 2^{20\log^2 N+1} (5 \cdot (1 /\log^5 N) \cdot 10 \log^2 N)^{10 \log^2 N} + 2^{-\log^6 N} \cdot 2^{(10\log^2 N+1)(20 \log^2 N + \log(N^{\log N} \cdot 2^{10\log^2 N}))} \right) \\
	\ge&1 - S_i \cdot N^{-\log N},
	\end{align*}
	all AND / OR gates at level $i$ become $s$-query decision tree after these $\tau_i$ restrictions.
	
	\newcommand{\tautotal}{\tau_{\sf total}}
	
	Finally, note that in total we have applied $\boldsym{\rho}$ at most
	\[
	\tautotal = 2 d + \log_{1/p} \left(\prod_{i=1}^{d} O_i\right) = \log_{1/p} N - \omega(1)
	\]
	times, and the final output gate shrinks to an $s$-query decision tree with probability at least
	\[
	1 - 2 \cdot S \cdot N^{-\log N}.
	\]
	Since $S \le N^{\log N} / 8$, with probability at least $3/4$, after all these restrictions, $C$ is equivalent to an $s$-query decision tree.
	
	\newcommand{\pend}{p_{\sf end}}
	\newcommand{\Nremain}{N_{\sf remain}}
	
	Now let $\pend = p^{\tautotal} = N^{-1} \cdot p^{-\omega(1)}$. By Chebyshev's inequality, the number of unrestricted variables at the end of the restriction is at least $\Nremain = \frac{1}{2} \cdot \pend \cdot N = (\log N)^{\omega(1)}$ with probability at least $1/2$. Therefore, with probability at least $1/4$, at the end of the restrictions, it holds that the remaining circuit $C$ is equivalent to an $s$-query decision tree $D$, and the number of unrestricted variables is at least $\Nremain$.
	
	Suppose we fix all these remaining unrestricted variables to be $0$ to get an input $x^*$, since each restriction from $\boldsym{\rho}$ can be computed by a $\poly(n)$-size circuit, $x^*$ has a circuit of $\poly(n) \cdot \log N = \poly(n) \le n^c$ size (now we set $c$). Let $S$ be the set of input variables that $D$ queries on the input $x^*$. Note that there are at least $2^{\Nremain - |S|}$ ways of assigning values to unrestricted variables while keeping variables in $S$ all $0$. And we can see that $F$'s output on $x^*$ is the same as its output on all of these assignments. But there must exist at least one assignment such the $\MCSP$ value is at least $(\log N)^{2c} = n^{2c}$ ($2^{\Nremain - |S|} = 2^{n^{\omega(1)}}$), contradiction to the assumption that $C$ computes $\MCSP[n^{c},n^{2c}]$.
\end{proofof}

\subsubsection{The Nearly Quadratic Formula Lower Bound of  \citep{HS17}}\label{ss:HS_lb}

In this section, we prove that the nearly quadratic formula lower bound of~\citep{HS17} localizes, and thereby proving the third item of Theorem~\ref{thm:frontier_C}. This localization indeed refutes a family of possible approaches to establish circuit lower bounds through hardness magnification via ``lowering the threshold''.

More concretely, consider the following hypothesized approach. Suppose we can compute $\MCSP[2^{\sqrt{n}}]$ by a formula $F$ with $\NP$ oracles, such that when we replace every oracle $O$ with fan-in $\beta$ in $F$ by a formula of size $\beta^k$ which reads all its inputs exactly $\beta^{k-1}$ times, the size of the new formula is less than $N^{1.99}$. Then we know that $\NP$ cannot be computed by formulas of size $n^k$ which reads all its inputs exactly $n^{k-1}$ times, as otherwise we get an $N^{1.99}$-size formula for $\MCSP[2^{\sqrt{n}}]$, which is a contradiction to the lower bound in~\citep{HS17}. If this holds for all $k > 0$, then we would have $\NP \not\subset \Formula[n^k]$ for all $k$.

In the following, by localizing~\citep{HS17}, we show that there is no such oracle formula construction for $\MCSP$ even if the oracles can be arbitrary. This excludes magnification theorems obtained by approaches that unconditionally produce circuits with oracles, and essentially addresses a question from \citep{OPS19_CCC}. It also suggests that the consideration of almost-formulas in HM Frontier C is unavoidable.  

\subsection*{A Size Measure on Oracle Formulas and A Potential Approach to Formula Size Lower Bound} We first introduce a size measure $\Size_t$ on oracle formulas to formalize the previous discussion.

For a parameter $t$ and an oracle formula $F$, we define $\Size_t(F)$ as the size of the formula, if we replace every oracle $O$ with fan-in $\beta$ in $F$ by a formula of size $\beta^t$ which reads all its inputs exactly $\beta^{t-1}$ times.

More formally, 
\[
\SIZE_t(F) := \begin{cases}
\SIZE_t(F_1) + \SIZE_t(F_2) &\quad \text{$F = F_1 \wedge F_2$  or $F = F_1 \vee F_2$,}\\
\beta^{t-1} \cdot \left( \sum_{i=1}^{\beta} \SIZE_t(F_i) \right) &\quad \text{$F=O(F_1,F_2,\dotsc,F_\beta)$}.
\end{cases}
\]

\newcommand{\ctiny}{c_{\sf tiny}}
\newcommand{\cdrop}{c_{\sf drop}}

\begin{proposition}\label{prop:approach}
	For a constant $k > 0$, if there is an $\NP$ oracle formula $F$ (all oracles are languages in $\NP$) for $\MCSP[2^{\sqrt{n}}]$ such that $\SIZE_{k+1}(F) \le N^{2-\eps}$ for a constant $\eps > 0$, then $\NP \not\subseteq \Formula[n^k]$.
\end{proposition}
\begin{proof}
	Suppose $\NP \subseteq \Formula[n^k]$ for the sake of contradiction. Then in particular each $\NP$ language can be computed by a size-$n^{k+1}$ formula which reads all its inputs exactly $n^{k}$ times by adding some dummy nodes in the formula. Therefore, by replacing all $\NP$ oracles in $F$ by such formulas, we have an $N^{2-\eps}$-size formula for $\MCSP[2^{\sqrt{n}}]$, in contradiction to the lower bound in~\citep{HS17}.
\end{proof}

\subsection*{Localization of~\citep{HS17}}\label{ss:local_HS}

Our following theorem shows that the above approach is not viable even with $k = 3$ by localizing~\citep{HS17}, with a moderate constraint on the adaptivity of the oracle circuits.

\begin{theorem}\label{thm:local-HS17}
	There is a universal constant $c$ such that for all constants $\eps > 0$ and $\alpha > 2$, $\MCSP[n^{c},2^{(\eps/\alpha) \cdot n}]$ cannot be computed by oracle formulas $F$ with $\SIZE_3(F) \le N^{2-\eps}$ and adaptivity $o(\log N/\log\log N)$ (that is, on any path from root to a leaf, there are at most $o(\log N / \log\log N)$ oracles).
\end{theorem}

\begin{remark}
	It is not hard to see that the adaptivity can be at most $O(\log N)$ given the condition $\SIZE_3(F) \le N^{2-\eps}$.
\end{remark}

Before proving Theorem~\ref{thm:local-HS17}, we first show it refutes the Anti-Checker Hypothesis (restated below) from~\cite{OPS19_CCC}.
\medskip

\noindent {\bf The Anti-Checker Hypothesis.} {\it For every $\lambda\in (0,1)$, there are $\eps>0$ and a collection $\mathcal{Y}=\{Y_1,\dots,Y_\ell\}$ of sets $Y_i\subseteq\{0,1\}^n$, where $\ell=2^{(2-\eps)n}$ and each $|Y_i|=2^{n^{1-\eps}}$, for which the following holds.
	
	If $f:\{0,1\}^n\mapsto \{0,1\}$ and $f\notin \Circuit[2^{n^\lambda}]$, then some set $Y\in \cal{Y}$ forms an anti-checker for $f$: For each circuit $C$ of size $2^{n^\lambda}/10n$, there is an input $y\in Y$ such that $C(y)\ne f(y)$.}

\medskip

\begin{corollary}\label{cor:anti-checker-hypothesis-wrong}
	The Anti-Checker Hypothesis is false.
\end{corollary}
\begin{proof}
	It is easy to see that, assuming the Anti-Checker Hypothesis, we can solve $\MCSP[2^{n^{1/3}},2^{n^{2/3}}]$ with a formula $F$ of $N^{2-\eps}$ size which uses $N^{2-\eps}$ oracles of fan-in $\poly(n)2^{n^{1-\eps}} = \polylog(N) \cdot 2^{(\log N)^{1-\eps}} = N^{o(1)}$ only at the layer above the inputs, for some $\eps > 0$. However, since $\SIZE_3(F) \le N^{2-\eps + o(1)}$, $F$ cannot compute $\MCSP[2^{n^{1/3}},2^{n^{2/3}}]$ by Theorem~\ref{thm:local-HS17}, contradiction.
\end{proof}

Now we are ready to prove Theorem~\ref{thm:local-HS17}.

\newcommand{\brho}{\boldsym{\rho}}

\begin{proofof}{Theorem~\ref{thm:local-HS17}}
	
	Let $k = \log^3 N$, and $\boldsym{\rho}$ be the $k$-wise independent $(1/\sqrt{k})$-regular random restriction guaranteed by Lemma~\ref{lm:k-wise-independent-restrictions}.
	
	For an oracle formula $F$ and a sub-formula $G$ of it, we say $G$ is a maximal sub-formula if $G$ is an entire subtree rooted at either the root, an oracle gate, or a gate whose father is an oracle.
	
	We are going to apply $t = \Theta(\log_k N)$ independent pseudorandom restrictions $\brho_1,\brho_2,\dotsc,\brho_{t}$, each distributed identically to $\brho$, where $t$ will be set precisely later.
	
	\subsection*{The Overall Proof Structure} 
	
	To analyze the size of the oracle formula under the random restriction sequence $\brho_1,\brho_2,\dotsc,\brho_{t}$, we define a potential function $\Phi$ inductively for all maximal sub-formulas of the given formula $F$. As it will be clear from the definition, $\Phi$ is not only a function of the structure of the oracle formula, but also depends on the history of the pseudorandom restrictions. 
	
	Formally, for each maximal sub-formula $G$ of the given formula $F$, and for each integer $0 \le i \le t$, we define a random variable $\Phi_{G,i}$, which denotes the potential function of $G$ after the first $i$ pseudorandom restrictions and only depends on $\brho_1,\brho_2,\dotsc,\brho_i$.

	\paragraph*{Definition of Tiny formulas and Blow up.} For an oracle formula, if the top gate is an oracle, we say it is tiny if it depends on at most $\log N$ variables. Otherwise, we say it is tiny if it depends on at most $\ctiny \cdot k$ variables, for a constant $\ctiny$ to be specified later.
	
	After each pseudorandom restriction, for a formula with an oracle gate at the top, when it depends on at most $b = 20$ variables, we blow it up to a formula of size $B=2^b$ (note that if there are two oracle gates $u,v$ such that $u$ and $v$ both depend on at most $b$ variables and $u$ is an ancestor of $v$, then it suffices to only blow up $u$). 
	
	The above two definitions (tiny formulas and the process of blowing up) may seem not easy to understand at first. Let us explain the motivation behind them. The key difficulty of the proof is to handle the oracle gates properly. The process of blowing up ensures that whenever an oracle becomes too small, we just replace it with a  constant size normal formula, so it becomes easier to deal with.
	
	The definition of tiny formulas is more subtle. As it will be clearly in Case II and Case III of the inductive definition of $\Phi$, setting the threshold of being tiny to $\log N$ for oracle formulas with top oracle gates ensures that the corresponding event of becoming tiny happens with high probability, which is indeed crucial in our proof.
	
	\paragraph*{The properties of $\Phi$.} We require the following properties on $\Phi$. 
	
	\begin{enumerate}
		\item For an oracle formula $F$, $\Phi$ is multiplied by a factor of $\frac{c_F}{k}$ under $\boldsym{\rho}$ in expectation, where $c_F$ depends on $F$ but it is upper bounded by a universal constant.
		
		\item With probability $1-p_F$, for all stages, and all maximal sub-formulas $G$ of $F$, $\Phi = 0$ for $G$ implies that $G$ is tiny, where $p_F$ depends on $F$ but it is upper bounded by $N^{-2}$.
		
		\item It holds that either $\Phi = 0 $ or $\Phi \ge 1$. Together with the second item, it implies that if the oracle formula is not tiny then $\Phi \ge 1$.
	\end{enumerate}
	
	
	With these carefully designed properties of $\Phi$, the overall proof is straightforward. We first show that $\Phi$ of $F$ is closely related to $\SIZE_3(F)$, and our conditions on the oracle formula imply that $\Phi$ of the whole oracle formula is bounded by $N^{2-\eps + o(1)}$ at the beginning. Then after roughly $t \approx \log_k(N^{2-\eps + o(1)})$ rounds of restrictions from $\boldsym{\rho}$, $\Phi$ becomes $0$ with a good probability, which also implies the whole oracle formula becomes tiny (only depend on $\polylog(N)$ bits).
	
	But then we argue that after $t$ rounds of restrictions from $\boldsym{\rho}$, with high probability the number of unrestricted variables is still at least $N^{\Omega(1)}$. Using a similar argument as from~\cite{HS17,OS18_mag_first,OPS19_CCC}, we show that the remained tiny oracle formula cannot compute $\MCSP[n^c,2^{\eps/\alpha \cdot n}]$ on the remaining variables, which concludes the proof.
	
	\subsection*{The Inductive Definition of the Potential Function $\Phi$}
	
	In the following, we gradually develop the definition of the potential function $\Phi$. We remark that Case I and Case II below are actually special cases of Case III and Case IV respectively. We discuss them first in the hope that they provide some intuitions and make it easier to understand the more complicated Case III and Case IV.
	
	\subsubsection*{Case I: $\Phi$ for a Pure Formula}
	
	We begin with the simplest case of pure formulas $F$ (formulas with no oracles) of size $S$. We define 
	\[
	\Phi = \begin{cases}
	S \quad& \text{ $S \ge 100 \cdot k$,}\\
	0 \quad& \text{ otherwise.}
	\end{cases}
	\]
	
	It follows from the shrinkage lemma \citep{Hastad98}, formula decomposition \citep[Claim 6.2]{Tal14}, and the $k$-wise independence of $\boldsym{\rho}$ that, when $S \ge 100 \cdot k$, the \emph{expected size} of $S$ drops by a factor of at least $k / c_{\sf Tal}$, for a universal constant $c_{\sf Tal}$ (we can set $c_F = c_{\sf Tal}$). Otherwise, the formula is tiny. It is straightforward to verify that all three properties of $\Phi$ are satisfied (we can set $p_F = 0$ in this case). 
	
	\subsubsection*{Case II: $\Phi$ for a Pure Oracle}
	
	Next we consider the case that $F$ is a pure oracle $O$ with fan-in $T$ (pure oracle means each input to $O$ is just a variable). We set 
	\[
	\Phi = T^{2} \cdot k^3
	\]
	at the beginning. And set $\Phi \leftarrow \Phi / k$ after each $\boldsym{\rho}$. Whenever it happens $\Phi < 1$, we set $\Phi = 0$ afterwards. Here, we can simply set $c_F = 1$. 
	
	Now we argue that with probability at least $1 - N^{-5}$ (that is, we set $p_F = N^{-5}$), when $\Phi  = 0$, $O$ only depends on at most $\log N$ variables and therefore becomes tiny.
	
	Note that $\Phi = 0$ means at least $\log_k T^{2}$ rounds of random restrictions have been applied.\footnote{Note that for this argument, a potential function of $T \cdot k$ already suffices. We use $T^2 \cdot k^3$ here to make it consistent with Case III.} Their composition is a $k$-wise independent restriction which keeps a variable unrestricted with probability at most $T^{-1}$. Therefore, the probability that the number of alive variable is larger than $\log N$ is smaller than
	\[
	\binom{T}{\log N} \cdot T^{-\log N} \le \left( \frac{e \cdot T}{\log N} \right)^{\log N} \cdot T^{-\log N} \le \left( \frac{e}{\log N} \right)^{\log N} \le N^{-5}.
	\]
	
	Note that in the above inequalities we can safely assume $T > \log N$.
	
	\subsubsection*{Case III: $\Phi$ for an Oracle Formula with an Oracle Top Gate}
	
	Then we move to the case of a maximal sub-formula $F$ with an oracle top gate $O$ with fan-in $T$. Let $\Phi_i$ be the corresponding potential function of the maximal sub-formula with root being the $i$-th input to $O$. We set
	\[
	\Phi = 
	\max\left(\sum_{i=1}^{T} \Phi_i, 1/k \right) \cdot T^{2} \cdot k^4,
	\] 
	at the beginning.
	
	When $\sum_{i=1}^{T} \Phi_i > 0$, we still let $\Phi = \left( \sum_{i=1}^{T} \Phi_i \right) \cdot T^2 \cdot k^4$.
	When $\sum_{i=1}^{T} \Phi_i$ first becomes $0$ (this could happen before the first restriction, if $\sum_{i=1}^{T} \Phi_i = 0$ at the beginning), we set $\Phi = T^{2} \cdot k^3$ and decrease it by a factor of $k$ during each later restriction, and set it to $0$ if it becomes $< 1$.
	
	Here, we set $c_F$ to be the maximum of $c_{F'}$ for all maximal sub-formulas $F'$ whose root is an input to the top oracle gate $O$ in $F$.
	
	First let us argue that $\Phi$ is multiplied by a factor of $\frac{c_F}{k}$ after each $\boldsym{\rho}$ in expectation. When $\sum_{i=1}^{T} \Phi_i = 0$, it is evident from the way we set $\Phi$ (note that $c_F \ge 1$). When $\sum_{i=1}^{T} \Phi_i > 0$, it follows from the induction as each $\Phi_i$ is multiplied by a factor of $\frac{c_F}{k}$ after each $\boldsym{\rho}$ in expectation. In the borderline case when $\sum_{i=1}^{T} \Phi_i > 0$ before $\boldsym{\rho}$ and becomes $0$ afterwards. One can see $\Phi$ drops from at least $T^2 \cdot k^4$ to at most $T^2 \cdot k^3$.
	
	Moreover, when $\sum_{i=1}^{T} \Phi_i = 0$, with probability at least $1 - \sum_{i=1}^{T} p_{F_i}$ ($F_i$ is the $i$-th sub-formula whose root is an input to the top oracle gate $O$ in $F$) all the sub-formulas are tiny, so at this time the oracle depends on at most $O(T \cdot k)$ variables.
	
	Therefore, when $\Phi$ drops to $0$, with probability at least $1 - \sum_{i=1}^{T} p_{F_i} - N^{-5}$ the whole oracle formula becomes tiny, by a calculation similar to the pure oracle case. Therefore, we can set $p_F = \sum_{i=1}^{T} p_{F_i} + N^{-5}$.

	\subsubsection*{Case IV: $\Phi$ for a Formula with Oracle Leaves}
	
	Finally, we deal with the most complicated case when the maximal sub-formula $F$ is a formula with oracle leaves. Suppose $F$ is a formula of size $S$ with $m$ oracle leaves. Let $\Phi_i$ be the potential function of the sub-formula corresponding to the $i$-th oracle leaf. Also, let $\cdrop$ be the maximum of the $c_F$'s of all the sub-formulas corresponding to the oracle leaves.
	
	\newcommand{\Nact}{N_{\sf active}}
	
	The difficulty in analyzing this case is that there could be many oracles which are tiny but have not blown up yet, and we have to keep track of the number of such oracles. Let $\Nact$ be the number of remaining active tiny oracles (oracles which are tiny but have not blown up). Clearly, $\Nact \le S$ at the beginning.
	
	We set
	\[
	\Phi = S + \Nact \cdot k^2 + \sum_{i=1}^{m} \Phi_i \cdot k^4,
	\]
	at the beginning. When $S \le 100 \cdot k$ happens, we change $\Phi$ to be 
	\[
	\Nact \cdot k^2 + \sum_{i=1}^{m} \Phi_i \cdot k^4
	\]
	afterwards.
	
	After each $\boldsym{\rho}$, if $S \ge 100 \cdot k$, the expected size of $S$ becomes at most
	\[
	c_1 \cdot S/k + c_2 \cdot k \cdot \left( \sum_{i=1}^{m} \Phi_i + \Nact\right),
	\]
	for two universal constants $c_1$ and $c_2$. This  bound holds because, by Claim 4.4 of~\cite{ImpagliazzoMZ19}, a formula of size $S$ can be decomposed into $6 S/k$ sub-formulas, each of size at most $k$, and each formula has at most $2$ sub-formula children. 
	
	The number of active oracle leaf (who is not blown up) is at most $\sum_{i=1}^{m} \Phi_i + \Nact$. Hence, at least $6 \cdot S/k - \sum_{i=1}^{m} \Phi_i - \Nact$ sub-formulas do not contain an active oracle leaf, and their total expected size is $O(S/k)$ after $\boldsym{\rho}$ (by Lemma 4.1 and Lemma 4.3 of~\cite{ImpagliazzoMZ19}, and~\cite{Tal14}).    For those sub-formulas containing active oracle leaves, their total size is at most $(\sum_{i=1}^{m} \Phi_i + \Nact) \cdot O(k)$ after $\boldsym{\rho}$ (this takes account of the worst case situation that all these active oracle leaves blow up).
	
	Also, we can see that after $\boldsym{\rho}$, $\Nact$ becomes at most
	\[
	\Nact/k^2 + \sum_{i=1}^{m} \Phi_i
	\]
	in expectation. This is because for a tiny active oracle depending on at most $\log N$ variables, the probability that it does not blow up after $\boldsym{\rho}$ is at most
	\[
	\binom{\log N}{b} \cdot k^{-b/2} \le (\log N)^{b - 1.5 \cdot b} = (\log N)^{-10} \le 1/k^2.
	\]    
	
	By induction, we also have that $\sum_{i=1}^{m} \Phi_i$ is multiplied by a factor of $\frac{\cdrop}{k}$ in expectation as well after each $\boldsym{\rho}$. Therefore, after $\boldsym{\rho}$, the expectation of $\Phi$ can be bounded by
	\begin{align*}
	&c_1 \cdot S/k + c_2 \cdot k \cdot \left( \sum_{i=1}^{m} \Phi_i + \Nact\right) + \left(
	\Nact/k^2 + \sum_{i=1}^{m} \Phi_i\right) \cdot k^2 + \left(\sum_{i=1}^{m} \Phi_i\right) \cdot \frac{\cdrop}{k} \cdot k^4\\
	\le& \;\; S \cdot \frac{c_1}{k} + \Nact \cdot k^2 \cdot \frac{c_2 + 1/k}{k} + \sum_{i=1}^{m} \Phi_i \cdot k^4 \cdot \left(\frac{\cdrop + c_2/k^2 + 1/k }{k} \right).
	\end{align*}
	
	We can set
	\[
	c_F = \max(c_1,c_2+1/k, \cdrop + c_2/k^2 + 1/k).
	\]
	
	Recall that when $S \le 100 \cdot k$ happens, we change $\Phi$ to be 
	\[
	\Nact \cdot k^2 + \sum_{i=1}^{m} \Phi_i \cdot k^4
	\]
	afterwards.
	
	By the previous discussion, after this $\Phi$ still drops by a factor of $k / c_F$ in expectation after each $\boldsym{\rho}$. Note that when $\Phi = 0$, we can see the size of the whole formula is smaller than $B \cdot 100 \cdot k = O(k)$, therefore it is tiny (here we set $\ctiny = B \cdot 100$). This is because $\Phi = 0$ implies $S \le 100 \cdot k$ happened at some point, and also $\Nact = \sum_{i=1}^m \Phi_i = 0$. They together imply that all oracles have blown up, and the size bound follows since each oracle adds at most $B$ leaves. 
	
	Let $F_i$ be the sub-formula with root being the $i$-th oracle leaf. In this case, we can set $p_F = \sum_{i=1}^m p_{F_i}$.
	
	\subsection*{The $\MCSP$ Lower Bound}
	
	Let $F$ be an oracle formula with $\SIZE_3(F) \le N^{2-\eps}$ and adaptivity $\tau =o(\log N/\log\log N)$. We first need to verify that $c_F$ is upper bounded by a universal constant. One can upper bound
	\[
	c_F \le \max(c_1,c_2 + 1/k,c_{\sf Tal}) + \tau \cdot (c_2/k^2 + 1/k) \le \max(c_1,c_2 + 1/k, c_{\sf Tal}) + o(1) = O(1).
	\]
	We can also upper bound $p_F$ by $p_F \le N^{-5} \cdot N^2 = N^{-3}$.
	
	By the inductive definition of the potential function $\Psi$ on maximal sub-formulas, it is not hard to show that
	\[
	\Phi \le \SIZE_3(F) \cdot k^{O(\tau)} \le N^{2-\eps + o(1)}.
	\]
	Note that this inequality crucially employs the definition of $\SIZE_3(\cdot)$.
	
	After each $\boldsym{\rho}$, $\Phi$ is reduced by a factor of $k/c_F$. After 
	\[
	t = \lceil \log_{k/c_F} \Phi \rceil + 2
	\]
	rounds of $\boldsym{\rho}$, the expected $\Phi$ of the overall formula becomes $<1/10$, which means with probability $0.9 - p_F \ge 0.8$ it is tiny and only depends on at most $O(k) = O(\log^3 N)$  variables.
	
	Note that by definition
	\[
	(k / c_F)^{t} \le \Phi \cdot k^3,
	\]
	and therefore
	\[
	k^{t} \le \Phi \cdot k^3 \cdot (c_F)^t \le N^{2-\eps + o(1)},
	\]
	as $(c_F)^t = (c_F)^{O(\log N /\log\log N)} = N^{o(1)}$.
	
	The composition of $t$ independent $\boldsym{\rho}$ keeps a variable unrestricted with probability $k^{-t/2} \ge N^{-1 + \eps/2 - o(1)}$, and is clearly pairwise independent. By Chebyshev's inequality, after $t$ restrictions from $\boldsym{\rho}$, with probability $0.5$, at least
	\[
	1/2 \cdot N \cdot N^{-1 + \eps/2 - o(1)} \ge N^{\eps/2 - o(1)}
	\]
	variables remain active. So with probability at least $0.3$, after $t$ restrictions from $\boldsym{\rho}$, the remaining formula $F$ only depends on $O(\log^3 N)$ variables, and the number of remaining unrestricted variables is at least $N^{\eps/2 - o(1)}$.
	
	Suppose we fix all these remaining unrestricted variables to be $0$ to get an input $x^*$. Since each restriction from $\boldsym{\rho}$ can be computed by a $\poly(n)$-size circuit, $x^*$ has a circuit of $\poly(n) \cdot t = \poly(n) \le n^{c}$ size (here we set $c$). Let $S$ be the set of input variables that $F$ depends on. Note that there are at least $2^{N^{\eps/2 - o(1)} - |S|}$ ways of assigning values to unrestricted variables while keeping variables in $S$ all $0$. Since $F$ only depends on $S$, $F$'s output on $x^*$ is the same as its output on all of these assignments. But there must exist at least one assignment such the $\MCSP$ value is at least $N^{\eps/\alpha} = 2^{(\eps /\alpha) \cdot n}$ as $\alpha > 2$. Therefore, $F$ cannot compute $\MCSP[n^{c},2^{(\eps / \alpha) \cdot n}]$.
\end{proofof}

\section*{Acknowledgements}

Part of this work was completed while some of the authors were visiting the Simons Institute for the Theory of Computing. We
are grateful to the Simons Institute for their support. This work was supported in part by the European Research Council under the European Union's Seventh Framework Programme (FP7/2007-2014)/ERC Grant
Agreement no.~615075. J\'{a}n Pich was supported in part by Grant 19-05497S of GA \v{C}R.  Lijie Chen is supported by NSF CCF-1741615 and a Google Faculty Research Award. Igor C. Oliveira was supported in part by a Royal Society University Research Fellowship.\footnote{Most of this work was completed while Igor C. Oliveira was affiliated with the University of Oxford.} 

\small

\bibliographystyle{alpha}	
\bibliography{refs}	

\newcommand{\etalchar}[1]{$^{#1}$}
\begin{thebibliography}{CMMW19}

\bibitem[AB09]{AB09}
Sanjeev Arora and Boaz Barak.
\newblock {\em Computational complexity: a modern approach}.
\newblock Cambridge University Press, 2009.

\bibitem[ABX08]{DBLP:conf/focs/ApplebaumBX08}
Benny Applebaum, Boaz Barak, and David Xiao.
\newblock On basing lower-bounds for learning on worst-case assumptions.
\newblock In {\em Symposium on Foundations of Computer Science \emph{(FOCS)}},
  pages 211--220, 2008.

\bibitem[AHM{\etalchar{+}}08]{AHMPS08_siamcomp_journals}
Eric Allender, Lisa Hellerstein, Paul McCabe, Toniann Pitassi, and Michael~E.
  Saks.
\newblock {Minimizing Disjunctive Normal Form Formulas and $\mathrm{AC}^0$
  Circuits Given a Truth Table}.
\newblock {\em {SIAM} J. Comput.}, 38(1):63--84, 2008.

\bibitem[AJ08]{DBLP:journals/dm/AndreevJ08}
Alexander~E. Andreev and Stasys Jukna.
\newblock Very large cliques are easy to detect.
\newblock {\em Discrete Mathematics}, 308(16):3717--3721, 2008.

\bibitem[Ajt83]{ajtai1983sup}
Mikl{\'o}s Ajtai.
\newblock $\sum_1^1$-formulae on finite structures.
\newblock {\em Annals of Pure and Applied Logic}, 24(1):1--48, 1983.

\bibitem[AK10]{AK10}
Eric Allender and Michal Kouck{\'{y}}.
\newblock Amplifying lower bounds by means of self-reducibility.
\newblock {\em J. {ACM}}, 57(3):14:1--14:36, 2010.

\bibitem[Baz09]{Bazzi09}
Louay M.~J. Bazzi.
\newblock Polylogarithmic independence can fool {DNF} formulas.
\newblock {\em {SIAM} J. Comput.}, 38(6):2220--2272, 2009.

\bibitem[CGJ{\etalchar{+}}18]{DBLP:journals/jcss/CheraghchiGJWX18}
Mahdi Cheraghchi, Elena Grigorescu, Brendan Juba, Karl Wimmer, and Ning Xie.
\newblock {AC}$^0 \circ$ {MOD}$_2$ lower bounds for the {Boolean Inner
  Product}.
\newblock {\em J. Comput. Syst. Sci.}, 97:45--59, 2018.

\bibitem[CIKK16]{DBLP:conf/coco/CarmosinoIKK16}
Marco~L. Carmosino, Russell Impagliazzo, Valentine Kabanets, and Antonina
  Kolokolova.
\newblock Learning algorithms from natural proofs.
\newblock In {\em Conference on Computational Complexity \emph{(CCC)}}, pages
  10:1--10:24, 2016.

\bibitem[CJW19a]{Magnification_FOCS19}
Lijie Chen, Ce~Jin, and Ryan Williams.
\newblock Hardness magnification for all sparse {NP} languages.
\newblock In {\em Symposium on Foundations of Computer Science \emph{(FOCS)}},
  2019.

\bibitem[CJW19b]{CJW_tight_threshold}
Lijie Chen, Ce~Jin, and Ryan Williams.
\newblock Sharp threshold results for computational complexity.
\newblock In {\em Unpublished manuscript}, 2019.

\bibitem[CKLM19]{CheraghchiKLM19_eccc_journals}
Mahdi Cheraghchi, Valentine Kabanets, Zhenjian Lu, and Dimitrios Myrisiotis.
\newblock Circuit lower bounds for {MCSP} from local pseudorandom generators.
\newblock In {\em International Colloquium on Automata, Languages, and
  Programming \emph{(ICALP)}}, pages 39:1--39:14, 2019.

\bibitem[CMMW19]{CMMW_CCC_paper}
Lijie Chen, Dylan~M. McKay, Cody~D. Murray, and R.~Ryan Williams.
\newblock Relations and equivalences between circuit lower bounds and
  {K}arp-{L}ipton theorems.
\newblock In {\em Computational Complexity Conference \emph{(CCC)}}, 2019.

\bibitem[COS17]{DBLP:conf/stoc/ChenOS17}
Xi~Chen, Igor~Carboni Oliveira, and Rocco~A. Servedio.
\newblock Addition is exponentially harder than counting for shallow monotone
  circuits.
\newblock In {\em Symposium on Theory of Computing \emph{(STOC)}}, pages
  1232--1245, 2017.

\bibitem[CS12]{DBLP:conf/focs/ChattopadhyayS12}
Arkadev Chattopadhyay and Rahul Santhanam.
\newblock Lower bounds on interactive compressibility by constant-depth
  circuits.
\newblock In {\em Symposium on Foundations of Computer Science \emph{(FOCS)}},
  pages 619--628, 2012.

\bibitem[CT19]{CT19_STOC}
Lijie Chen and Roei Tell.
\newblock Bootstrapping results for threshold circuits ``just beyond'' known
  lower bounds.
\newblock In {\em Symposium on Theory of Computing \emph{(STOC)}}, 2019.

\bibitem[DF13]{DBLP:series/txcs/DowneyF13}
Rodney~G. Downey and Michael~R. Fellows.
\newblock {\em Fundamentals of Parameterized Complexity}.
\newblock Texts in Computer Science. Springer, 2013.

\bibitem[For02]{Forster02}
J{\"{u}}rgen Forster.
\newblock A linear lower bound on the unbounded error probabilistic
  communication complexity.
\newblock {\em J. Comput. Syst. Sci.}, 65(4):612--625, 2002.

\bibitem[FSS84]{DBLP:journals/mst/FurstSS84}
Merrick~L. Furst, James~B. Saxe, and Michael Sipser.
\newblock Parity, circuits, and the polynomial-time hierarchy.
\newblock {\em Mathematical Systems Theory}, 17(1):13--27, 1984.

\bibitem[GKRS19]{DBLP:conf/innovations/GoosKRS19}
Mika G{\"{o}}{\"{o}}s, Pritish Kamath, Robert Robere, and Dmitry Sokolov.
\newblock Adventures in monotone complexity and {TFNP}.
\newblock In {\em Innovations in Theoretical Computer Science Conference
  \emph{(ITCS)}}, pages 38:1--38:19, 2019.

\bibitem[H{\aa}s98]{Hastad98}
Johan H{\aa}stad.
\newblock The shrinkage exponent of de morgan formulas is 2.
\newblock {\em {SIAM} J. Comput.}, 27(1):48--64, 1998.

\bibitem[Hir18]{Hir18}
Shuichi Hirahara.
\newblock Non-black-box worst-case to average-case reductions within {NP}.
\newblock In {\em Symposium on Foundations of Computer Science \emph{(FOCS)}},
  pages 247--258, 2018.

\bibitem[HS17]{HS17}
Shuichi Hirahara and Rahul Santhanam.
\newblock On the average-case complexity of {MCSP} and its variants.
\newblock In {\em Computational Complexity Conference \emph{(CCC)}}, pages
  7:1--7:20, 2017.

\bibitem[IMZ19]{ImpagliazzoMZ19}
Russell Impagliazzo, Raghu Meka, and David Zuckerman.
\newblock Pseudorandomness from shrinkage.
\newblock {\em J. {ACM}}, 66(2):11:1--11:16, 2019.

\bibitem[IPS97]{DBLP:journals/siamcomp/ImpagliazzoPS97}
Russell Impagliazzo, Ramamohan Paturi, and Michael~E. Saks.
\newblock Size-depth tradeoffs for threshold circuits.
\newblock {\em {SIAM} J. Comput.}, 26(3):693--707, 1997.

\bibitem[IPZ01]{DBLP:journals/jcss/ImpagliazzoPZ01}
Russell Impagliazzo, Ramamohan Paturi, and Francis Zane.
\newblock Which problems have strongly exponential complexity?
\newblock {\em J. Comput. Syst. Sci.}, 63(4):512--530, 2001.

\bibitem[IW01]{IW01}
Russell Impagliazzo and Avi Wigderson.
\newblock Randomness vs time: Derandomization under a uniform assumption.
\newblock {\em J. Comput. Syst. Sci.}, 63(4):672--688, 2001.

\bibitem[Jer09]{DBLP:journals/jsyml/Jerabek09}
Emil Jer{\'{a}}bek.
\newblock Approximate counting by hashing in bounded arithmetic.
\newblock {\em J. Symb. Log.}, 74(3):829--860, 2009.

\bibitem[Juk12]{DBLP:books/daglib/0028687}
Stasys Jukna.
\newblock {\em Boolean Function Complexity - Advances and Frontiers}.
\newblock Springer, 2012.

\bibitem[Jus72]{Justesen72_tit_journals}
J{\o}rn Justesen.
\newblock {Class of constructive asymptotically good algebraic codes}.
\newblock {\em {IEEE} Trans. Information Theory}, 18(5):652--656, 1972.

\bibitem[Kop11]{Kopparty11_stoc_conf}
Swastik Kopparty.
\newblock {On the complexity of powering in finite fields}.
\newblock In {\em Proceedings of the Symposium on Theory of Computing
  \emph{(STOC)}}, pages 489--498, 2011.

\bibitem[KR13]{KR_STOC_paper}
Ilan Komargodski and Raz Ran.
\newblock Average-case lower bounds for formula size.
\newblock In {\em Symposium on Theory of Computing \emph{(STOC)}}, 2013.

\bibitem[Kra11]{KForc11}
Jan Kraj\'i\v{c}ek.
\newblock {\em Forcing with random variables and proof complexity}.
\newblock Cambridge University Press, 2011.

\bibitem[KW90]{DBLP:journals/siamdm/KarchmerW90}
Mauricio Karchmer and Avi Wigderson.
\newblock Monotone circuits for connectivity require super-logarithmic depth.
\newblock {\em {SIAM} J. Discrete Math.}, 3(2):255--265, 1990.

\bibitem[LW13]{LW13}
Richard~J. Lipton and Ryan Williams.
\newblock Amplifying circuit lower bounds against polynomial time, with
  applications.
\newblock {\em Computational Complexity}, 22(2):311--343, 2013.

\bibitem[MMW19]{MMW_STOC_paper}
Dylan~M. McKay, Cody~D. Murray, and R.~Ryan Williams.
\newblock Weak lower bounds on resource-bounded compression imply strong
  separations of complexity classes.
\newblock In {\em Symposium on Theory of Computing \emph{(STOC)}}, 2019.

\bibitem[MP17]{DBLP:journals/eccc/MullerP17}
Moritz M{\"{u}}ller and J{\'{a}}n Pich.
\newblock Feasibly constructive proofs of succinct weak circuit lower bounds.
\newblock {\em Electronic Colloquium on Computational Complexity
  \emph{(ECCC)}}, 24:144, 2017.

\bibitem[Oli19]{DBLP:conf/icalp/Oliveira19}
Igor~Carboni Oliveira.
\newblock Randomness and intractability in {K}olmogorov complexity.
\newblock In {\em International Colloquium on Automata, Languages, and
  Programming \emph{(ICALP)}}, pages 32:1--32:14, 2019.

\bibitem[OPS19]{OPS19_CCC}
Igor~Carboni Oliveira, J{\'{a}}n Pich, and Rahul Santhanam.
\newblock Hardness magnification near state-of-the-art lower bounds.
\newblock In {\em Computational Complexity Conference \emph{(CCC)}}, 2019.

\bibitem[OS15]{DBLP:conf/coco/OliveiraS15}
Igor~Carboni Oliveira and Rahul Santhanam.
\newblock Majority is incompressible by {AC}$^0[p]$ circuits.
\newblock In {\em Conference on Computational Complexity \emph{(CCC)}}, pages
  124--157, 2015.

\bibitem[OS17]{DBLP:conf/coco/OliveiraS17}
Igor~Carboni Oliveira and Rahul Santhanam.
\newblock Conspiracies between learning algorithms, circuit lower bounds, and
  pseudorandomness.
\newblock In {\em Computational Complexity Conference \emph{(CCC)}}, pages
  18:1--18:49, 2017.

\bibitem[OS18]{OS18_mag_first}
Igor~Carboni Oliveira and Rahul Santhanam.
\newblock Hardness magnification for natural problems.
\newblock In {\em Symposium on Foundations of Computer Science \emph{(FOCS)}},
  pages 65--76, 2018.

\bibitem[Raz85]{Razborov:85a}
Alexander~A. Razborov.
\newblock Lower bounds on the monotone complexity of some {B}oolean functions.
\newblock {\em Doklady Akademii Nauk SSSR}, 281:798--801, 1985.
\newblock English translation in: Soviet Mathematics Doklady 31:354--357, 1985.

\bibitem[Raz87]{Razborov87}
Alexander~A. Razborov.
\newblock Lower bounds on the size of constant-depth networks over a complete
  basis with logical addition.
\newblock {\em Mathematicheskie Zametki}, 41(4):598--607, 1987.

\bibitem[Raz15]{Razb15}
Alexander~A. Razborov.
\newblock Pseudorandom generators hard for $k$-{DNF} resolution and polynomial
  calculus resolution.
\newblock {\em Annals of Mathematics}, 181(2):415--472, 2015.

\bibitem[Rei11]{Rei11}
Ben~W. Reichardt.
\newblock Reflections for quantum query algorithms.
\newblock In {\em Symposium on Discrete Algorithms \emph{(SODA)}}, pages
  560--569, 2011.

\bibitem[RR97]{DBLP:journals/jcss/RazborovR97}
Alexander~A. Razborov and Steven Rudich.
\newblock Natural proofs.
\newblock {\em J. Comput. Syst. Sci.}, 55(1):24--35, 1997.

\bibitem[Smo87]{DBLP:conf/stoc/Smolensky87}
Roman Smolensky.
\newblock Algebraic methods in the theory of lower bounds for {B}oolean circuit
  complexity.
\newblock In {\em Symposium on Theory of Computing \emph{(STOC)}}, pages
  77--82, 1987.

\bibitem[Sri03]{Srinivasan03}
Aravind Srinivasan.
\newblock On the approximability of clique and related maximization problems.
\newblock {\em J. Comput. Syst. Sci.}, 67(3):633--651, 2003.

\bibitem[SS96]{SipserS96_tit_journals}
Michael Sipser and Daniel~A. Spielman.
\newblock {Expander codes}.
\newblock {\em {IEEE} Trans. Information Theory}, 42(6):1710--1722, 1996.

\bibitem[Tal14]{Tal14}
Avishay Tal.
\newblock Shrinkage of {De} {Morgan} formulae by spectral techniques.
\newblock In {\em Symposium on Foundations of Computer Science \emph{(FOCS)}},
  pages 551--560, 2014.

\bibitem[Tal17a]{Tal17}
Avishay Tal.
\newblock Formula lower bounds via the quantum method.
\newblock In {\em Symposium on Theory of Computing \emph{(STOC)}}, pages
  1256--1268, 2017.

\bibitem[Tal17b]{Tal17_coco_conf}
Avishay Tal.
\newblock {Tight Bounds on the Fourier Spectrum of $\mathrm{AC}^0$}.
\newblock In {\em Computational Complexity Conference \emph{(CCC)}}, pages
  15:1--15:31, 2017.

\bibitem[TX13]{TrevisanX13_coco_conf}
Luca Trevisan and Tongke Xue.
\newblock {A Derandomized Switching Lemma and an Improved Derandomization of
  $\mathrm{AC}^0$}.
\newblock In {\em Conference on Computational Complexity \emph{(CCC)}}, pages
  242--247, 2013.

\bibitem[Vad12]{Vadhan12}
Salil~P. Vadhan.
\newblock Pseudorandomness.
\newblock {\em Foundations and Trends in Theoretical Computer Science},
  7(1-3):1--336, 2012.

\bibitem[Yao89]{Yao89}
Andrew~Chi{-}Chih Yao.
\newblock Circuits and local computation.
\newblock In {\em Symposium on Theory of Computing \emph{(STOC)}}, pages
  186--196, 1989.

\end{thebibliography}

\normalsize

\appendix

\section{Review of Hardness Magnification in Circuit Complexity}
\label{a:review_magnification}
\subsection{Previous Work}\label{ss:previous_work}

We focus on some representative examples. For definitions and more details, check Section \ref{s:preliminaries} or consult the original papers.\\

\noindent \textbf{Srinivasan \citep{Srinivasan03} (Informal).} If there exists $\varepsilon > 0$ such that $n^{1 - o(1)}$-approximating $\mathsf{MAX}$-$\mathsf{CLIQUE}$ requires boolean circuits of size at least $m^{1 + \varepsilon}$ (where $m = \Theta(n^2)$), then $\mathsf{NP} \nsubseteq \mathsf{Circuit}[\mathsf{poly}]$.\\

\noindent \textbf{Allender-Kouck{\'{y}} \citep{AK10} and Chen-Tell \citep{CT19_STOC}.} The following results hold.
\begin{itemize}
\item Let $\Pi \in \{\mathsf{BFE}, \mathsf{W}_{\mathsf{S}_5}, \mathsf{W}5\text{-}\mathsf{STCONN}\}$. Suppose that for each $c > 1$ there exist infinitely many $d \in \mathbb{N}$ such that $\mathsf{TC}^0$ circuits of depth $d$ require more than $n^{1 + c^{-d}}$ wires to solve $\Pi$. Then, $\mathsf{NC}^1 \nsubseteq \mathsf{TC}^0$.
\item Suppose that for each $c > 1$ there exist infinitely many $d \in \mathbb{N}$ such that $\mathsf{MAJ}$ cannot be computed by $\mathsf{ACC}^0$ circuits of depth $d$ with $n^{1 + c^{-d}}$ wires. Then $\mathsf{MAJ} \notin \mathsf{ACC}^0$, and consequently $\mathsf{TC}^0 \nsubseteq \mathsf{ACC}^0$.
\end{itemize} 

\vspace{0.25cm}

\noindent \textbf{Lipton-Williams \citep{LW13}.} If there is $\varepsilon > 0$ such that for every $\delta > 0$ we have $\mathsf{CircEval} \notin \mathsf{Size}\text{-}\mathsf{Depth}[n^{1 + \varepsilon}, n^{1 - \delta}]$, then for every $k \geq 1$ and $\gamma > 0$ we have $\mathsf{CircEval} \notin \mathsf{Size}\text{-}\mathsf{Depth}[n^k, n^{1 - \gamma}]$ (in particular $\mathsf{P} \nsubseteq \mathsf{NC}$).\\

\noindent \textbf{Oliveira-Santhanam \citep{OS18_mag_first}.} The following results hold.
\begin{itemize}
\item Let $s(n) = n^k$ and $\delta(n) = n^{-k}$, where $k \in \mathbb{N}$. If $\mathsf{MCSP}[(s,0),(s,\delta)] \notin \mathsf{Formula}[N^{1 + \varepsilon}]$ for some $\varepsilon > 0$, then there is $L \in \mathsf{NP}$ over $m$-bit inputs and $\delta > 0$ such that $L \notin \mathsf{Formula}[2^{m^\delta}]$.
\item Suppose there exists $k \geq 1$ such that for every $d \geq 1$ there is $\varepsilon_d > 0$ such that $\mathsf{MCSP}[(s,0),(s,\delta)] \notin \mathsf{AC}^0_d[N^{1 + \varepsilon_d}]$, where $s(n) = n^k$ and $\delta(n) = n^{-k}$. Then $\mathsf{NP} \nsubseteq \mathsf{NC}^1$. 
\item Let $k(n) = n^{o(1)}$. If there exists $\varepsilon > 0$ such that $k$-$\mathsf{Vertex}$-$\mathsf{Cover} \notin \mathsf{DTISP}[m^{1 + \varepsilon}, m^{o(1)}]$, where the input is an $n$-vertex graph represented by an adjacency matrix of bit length $m = \Theta(n^2)$, then $\mathsf{P} \neq \mathsf{NP}$.
\item Let $k(n) = (\log n)^C$, where $C \in \mathbb{N}$ is arbitrary. If for every $d \geq 1$ there exists $\varepsilon > 0$ such that $k$-$\mathsf{Vertex}$-$\mathsf{Cover} \notin \mathsf{AC}^0_d[m^{1 + \varepsilon}]$, then $\mathsf{NP} \nsubseteq \mathsf{NC}^1$.
\end{itemize}

\vspace{0.25cm}

\noindent \textbf{Oliveira-Pich-Santhanam \citep{OPS19_CCC} and McKay-Murray-Williams \citep{MMW_STOC_paper} (Informal).} If there exists $\varepsilon > 0$ such that
for every small enough $\beta > 0$,
\begin{itemize}
\item $\mathsf{MCSP}[2^{\beta n}] \notin \mathsf{Circuit}[N^{1 + \varepsilon}]$, then $\mathsf{NP} \nsubseteq \mathsf{Circuit}[\mathsf{poly}]$.
\vspace{-0.05cm}
\item $\mathsf{MKtP}[2^{\beta n}] \notin \mathsf{TC}^0[N^{1 + \varepsilon}]$, then $\mathsf{EXP} \nsubseteq \mathsf{TC}^0[\mathsf{poly}]$.
\item $\mathsf{MKtP}[2^{\beta n}] \notin U_2$-$\mathsf{Formula}[N^{3 + \varepsilon}]$, then $\mathsf{EXP} \nsubseteq \mathsf{Formula}[\mathsf{poly}]$.
\item $\mathsf{MKtP}[2^{\beta n}] \notin B_2$-$\mathsf{Formula}[N^{2 + \varepsilon}]$, then $\mathsf{EXP} \nsubseteq \mathsf{Formula}[\mathsf{poly}]$.
\item $\mathsf{MKtP}[2^{\beta n}] \notin \mathsf{Formula}$-$\mathsf{XOR}[N^{1 + \varepsilon}]$, then $\mathsf{EXP} \nsubseteq \mathsf{Formula}[\mathsf{poly}]$.
\item $\mathsf{MKtP}[2^{\beta n}] \notin \mathsf{BP}[N^{2 + \varepsilon}]$, then $\mathsf{EXP} \nsubseteq \mathsf{BP}[\mathsf{poly}]$.
\item $\mathsf{MKtP}[2^{\beta n}] \notin (\mathsf{AC}^0[6])[N^{1 + \varepsilon}]$, then $\mathsf{EXP} \nsubseteq \mathsf{AC}^0[6]$.
\end{itemize} 
\noindent Many results for $\mathsf{MKtP}$ admit analogues for $\mathsf{MrKtP}$, which considers a randomized version of $\mathsf{Kt}$ complexity introduced by \citep{DBLP:conf/icalp/Oliveira19}. An advantage of $\mathsf{MrKtP}$ is that strong unconditional lower bounds against uniform computations are known, while the hardness of problems such as $\mathsf{MCSP}$ and  $\mathsf{MKtP}$ currently relies on cryptographic assumptions.

\vspace{0.5cm}

\noindent \textbf{Chen-McKay-Murray-Williams \citep{CMMW_CCC_paper}.} The following results hold.
\begin{itemize}
\item If there is $\varepsilon > 0$, $c \geq 1$, and an $n^c$-sparse language $L \in \mathsf{NP}$ such that $L \notin \mathsf{Circuit}[n^{1 + \varepsilon}]$, then $\mathsf
  {NE} \nsubseteq \mathsf{Circuit}[2^{\delta \cdot n}]$ for some $\delta > 0$.
\item  If there is $\varepsilon  > 0$ such that for every $\beta > 0$ there is a $2^{n^\beta}$-sparse language $L \in \mathsf{NTIME}[2^{n^\beta}]$ such that $L \notin
  \mathsf{Circuit}[n^{1 + \varepsilon}]$, then $\mathsf{NEXP} \nsubseteq \mathsf{Circuit}[\mathsf{poly}]$.
\end{itemize} 

\vspace{0.25cm}

More recently, \citep{Magnification_FOCS19} established that many hardness magnification theorems for problems such as $\mathsf{MCSP}$ and $\mathsf{MKtP}$ hold in fact under the assumption that a \emph{sufficiently sparse and explicit language} admits weak lower bounds. We refer to their work for more details.

\subsection{Hardness Magnification Through the Lens of Oracle Circuits}\label{ss:kernelization}

We can view the results from Appendix \ref{ss:previous_work} as \emph{unconditional} upper bounds on the size of small fan-in oracle circuits solving the corresponding problems, for a certain choice of oracle gates. In a magnification theorem, it is important to upper bound the uniform complexity of the oracle gates. For our discussion, this is not going to be relevant. 

We repeat here a definition from Section \ref{s:preliminaries}, for convenience of the reader.
\begin{definition}[Local circuit classes]
  Let $\mathcal{C}$ be a circuit class \emph{(}such as $\mathsf{AC}^0[s]$, $\mathsf{TC}_d^0[s]$, $\mathsf{Circuit}[s]$, etc\emph{)}. For functions $q, \ell, a \colon
  \mathbb{N} \to \mathbb{N}$, we say that a language $L$ is in $[q,\ell, a]$--$\,\mathcal{C}$ if there exists a sequence $\{E_n\}$ of oracle circuits for which the
  following holds:
  \begin{itemize}
  \item[\emph{(\emph{i})}] Each oracle circuit $E_n$ is a circuit from $\mathcal{C}$.
  \item[\emph{(\emph{ii})}] There are at most $q(n)$ oracle gates in $E_n$, each of fan-in at most $\ell(n)$, and any path from an input gate to an output gate
    encounters at most $a(n)$ oracle gates.
  \item[\emph{(\emph{iii})}] There exists a language $\mathcal{O} \subseteq \{0,1\}^*$ such that the sequence $\{E_n^\mathcal{O}\}$ \emph{(}$E_n$ with its oracle gates
    set to $\mathcal{O}$\emph{)} computes $L$.
\end{itemize}
\end{definition}

In the definition above, $q$ stands for \emph{quantity}, $\ell$ for \emph{locality}, and $a$ for \emph{adaptivity} of the corresponding oracle gates.

The fact that existing magnification theorems produce such circuits is a consequence of the algorithmic nature of the underlying proofs, which show how to reduce an instance of a problem to shorter instances of another related problem. By inspection of each proof, it is possible to establish a variety of upper bounds. We explicitly state some of them below.

\begin{proposition}\label{p:magnification_as_oracle_circuits}
The following results hold.
\begin{itemize}
\item \emph{\citep{AK10}} For every $\Pi \in \{\mathsf{BFE}, \mathsf{W}_{\mathsf{S}_5}, \mathsf{W}5\text{-}\mathsf{STCONN}\}$ and every $\beta > 0$, $\Pi_n \in \left[ O\left( n^{1-\beta} \right), n^{\beta}, O\!\left( \frac{1}{\beta} \right) \right]$--$\,\mathsf{TC}^0[O(n)]$.
\item \emph{\citep{LW13}} For every $\delta > 0$, $\mathsf{CircEval}_n \in \left[ n \cdot \poly(\log n), n^\delta, n^{1-\delta} \right]$--$\,\Circuit[n \cdot \poly(\log n)]$.
\item \emph{\citep{OS18_mag_first}} For every constructive function $n \leq s(n) \leq 2^n/\poly(n)$ and parameter $0 < \delta(n) < 1/2$, $\mathsf{MCSP}[(s,0),(s,\delta)] \in [N,\poly(s/\delta),1]$--$\,\mathsf{Formula}[N \cdot \poly(s/\delta)]$.
\item \emph{\citep{OS18_mag_first}} Let $k = (\log n)^C$, where $C \in \mathbb{N}$. Then $k$-$\mathsf{Vertex}$-$\mathsf{Cover} \in [1, (\log n)^{4C}, 1]$--$\,\mathsf{AC}^0_d[m^{1 + \varepsilon}]$, where $\varepsilon_d \rightarrow 0$ as $d \rightarrow \infty$.
\item \emph{\citep{OPS19_CCC}} For every $\beta > 0$ and for every constructive function $s(n) \leq 2^{\beta n}$, $\mathsf{Gap}$-$\mathsf{MKtP} \in
  [N, \mathsf{poly}(s), 1]$--$\,\mathsf{Formula}$-$\mathsf{XOR}[N \cdot \poly(s)]$.
\item \sloppy \emph{\citep{OPS19_CCC}} For every constructive function $s(n) \leq 2^n/\poly(n)$, it follows that $\mathsf{Gap}$-$\mathsf{MCSP} \in [N \cdot \mathsf{poly}(s), \mathsf{poly}(s), \mathsf{poly}(s)]$--$\,\mathsf{Circuit}[N \cdot \poly(s)]$.
\item \sloppy \emph{\citep{MMW_STOC_paper}} For every constructive function $s(n) \leq 2^n/\poly(n)$, we have $\MCSP[s(n)] \in [O(N/\poly (s)), \poly (s), O(n/\log (s))]$--$\,\Circuit[N / \poly(s)]$.
\end{itemize}
\end{proposition}

We stress however that not every hardness magnification theorem needs to lead to an unconditional construction of efficient oracle circuits. (All the proofs that we know of produce such circuits though.)

\end{document}